\documentclass[sn-mathphys]{sn-jnl}%

\usepackage{threeparttable}
\usepackage{bm}
\usepackage{color}
\newtheorem{lemma}{Lemma} 
\usepackage[section]{placeins}

\newtheorem{observation}{Observation}

\jyear{2021}%

\theoremstyle{thmstyleone}%
\newtheorem{theorem}{Theorem}%
\newtheorem{proposition}[theorem]{Proposition}%

\theoremstyle{thmstyletwo}%
\newtheorem{example}{Example}%

\theoremstyle{thmstylethree}%

\begin{document}
	
	\title[Non-independence of modular additions in ARX ciphers]{Towards non-independence of modular additions in searching differential trails of ARX ciphers: new automatic methods with  application to SPECK and Chaskey}

	\author[1,2]{ \sur{Haiwen Qin}}%
	
	\author[1,2]{\sur{Baofeng Wu}}%

	\affil[1]{\orgdiv{State Key Laboratory of Information Security}, \orgname{Institute of Information Engineering, Chinese Academy of Sciences}, \orgaddress{\city{Beijing}, \postcode{100093}, \country{China}}}
	
	\affil[2]{\orgdiv{School of Cyber Security}, \orgname{University of Chinese Academy of Sciences}, \orgaddress{\city{Beijing}, \postcode{100049}, \country{China}}}

	\abstract{
		ARX-based ciphers, constructed by the modular addition,  rotation and XOR operations, have been receiving a lot of attention in the design of lightweight symmetric ciphers.
		For their differential cryptanalysis, most automatic search methods of differential trails adopt the assumption of independence of modulo additions.
		However, this assumption does not necessarily hold when the trail includes consecutive modular additions (CMAs). %
		It has already been found that in this case  some differential trails searched by automatic methods before are actually impossible, but the study is not in 
		depth yet, for example, few effort has been paid to exploiting the root causes of non-independence between CMAs and accurate calculation of probabilities of the valid trails. 
		In this paper, we devote to solving these two problems. By examing the differential equations of single and consecutive modular additions, we find that the influence of non-independence  can be described by  relationships between constraints on the intermediate state of two additions. Specifically, constraints of the first addition can make some of  its output bits non-uniform, and when they meet the  constraints of the second addition, the differential probability of the whole CMA may be different from the value calculated under the independence assumption.
		As a result, we can build SAT models to verify the validity of a given differential trail of ARX ciphers and \#SAT models to calculate the exact probabilities of the differential propagation through CMAs in the trail, promising  a more accurate evaluation of probability of the trail.
		Our automic  methods and searching tools are applied to search related-key differential trails of SPECK and Chaskey including CMAs in the key schedule and the round function respectively. 
		Our SAT-based search strategies find that the probabilities of optimal related-key differential trails of 11-round  SPECK32/64 and 11- and 12-round SPECK48/96 are $2^3$, $2$, and $2^4$ times higher than the existing results, respectively, 
		and probabilities of some good related-key differential trails of 15-round SPECK32/64, SPECK48/96 and SPECK64/128, including several non-independent CMAs, are $2^{9}$,  $2^{5.5}$, and $2^{16}$ times higher, respectively.
		For Chaskey, we give a more accurate probability evaluation for the 8-round differential trail given by the designers considering the non-independence of additions, which turns out to be $2^{7.8}$ times higher than it calculated  under the independence assumption.
		
	}


	\maketitle

	\section{Introduction}
	\label{sec:Introduction} 
	ARX-based ciphers, constructed by the modular addition, rotation, and XOR operations, have received great attention in the design of lightweight symmetric ciphers in recent years. 
	Since these three cryptographic primitives are very simple and easy to implement in software as well as hardware, ARX-based ciphers are usually simpler and faster than Sbox-based ciphers. Besides, the combination between these three operations can provide sufficiently high security, and make the cipher resistant to side-channel attacks.
	Some famous ARX ciphers, to name a few, include the block ciphers SPECK \cite{DBLP:conf/dac/BeaulieuSSTWW15}, Sparkle \cite{ToSC:BBCGPUVW20}, the  MAC algorithm Chaskey \cite{mouha2014chaskey}, the stream ciphers Salsa20 \cite{DBLP:series/lncs/Bernstein08} and Chacha \cite{bernstein2008chacha}, the hash functions Siphash \cite{DBLP:conf/indocrypt/AumassonB12}, Skein \cite{ferguson2010skein} and BLAKE \cite{BLAKE}.
	Among them, SPECK is released by the NSA and standardized by ISO as a part of the RFID air interface standard, ISO/29167-22.
	Siphash has become part of the default hash table implementation in Python.
	With a similar structure to Siphash, Chaskey is the MAC algorithm standardized in ISO/IEC 29192-6.
	There are mainly four cryptanalytic methods for  ARX ciphers, differential attack, linear attack, rotational attack, and differential-linear attack. The core of differential cryptanalysis is to search a differential trail (characteristic) with  probability $p$ which is higher than $2^{-n}$, where $n$ is the length of the input of the cipher, and use the trail to build a distinguisher to obtain some information of the key or recover the key with a complexity of $O({1}/{p})$. To reduce complexity and increase success probability of such an attack, it is important to find trails with probabilities as high as possible, and as accurate as possible.

	Since  modular addition is the only non-linear operation in  ARX ciphers, 
	the calculation of  probability of a differential trail comes down to that of each modular addition in the trail. 
	In the past two decades,  differential propagation of modular addition has been well studied. 
	The first effective algorithm was proposed by Lipmaa and Moriai \cite{FSE:LipMor01}, 
	where both the validity verification and probability calculation of a differential  were expressed by formulas about the input and output differences. In this paper we call them Lipmaa-Moriai formulas.
	Later, Mouha et al. \cite{DBLP:conf/sacrypt/MouhaVCP10} used  S-functions and graph theory to calculate the probability of a given differential by matrix multiplications.
	Following that, Schulte-Geers \cite{DBLP:journals/dcc/Schulte-Geers13} 
	studied the differential property of the  modular addition function from the perspective of  CCZ-equivalence, and proposed a theoretical explanation for the calculation of the differential probability. 
	These three methods can accurately calculate the differential probability when two inputs of the addition operation are uniform random, and have been widely used to build models for searching optimal differential trails of ARX ciphers.
	However, when the inputs are not uniform random, their calculations may be inaccurate. 
	In most previous work, the automatic search methods of differential trails for ARX ciphers are based on  the independence assumption, where the differential propagation on each modular addition is considered to be independent, and the probability of the whole differential trail is computed as the product of differential probabilities of each addition.
	With the above three methods to calculate the differential probability of a modular addition, there are mainly four methods to automatically search differential trails.
	The first is the SMT-based method proposed in \cite{EPRINT:MouPre13,DBLP:conf/fse/Biryukov0V14}, where Lipmaa-Moriai formulas were transformed into some logical operations between bit-vectors of input and output differences, and were used in building an SMT model to capture the differential propagation with a specific target probability. In this method, the automatic search for optimal trails is realized by  solving the SMT models 
	and adjusting the target probability according to the search results.
	The second is the MILP-based method introduced in \cite{FSE:FWGSH16}, where Lipmaa-Moriai formulas were transformed into a series of linear inequalities, and used to build an MILP model to characterize the propagation of differences, the  probability of which is transformed into the objective function. The state of art solver Gurobi \cite{gurobi} are often used in this method to solve the MILP models.
	The third is  Matsui's branch-and-bound algorithm for ARX ciphers  \cite{DBLP:journals/tit/LiuLJW21}, where Lipmaa-Moriai formulas 
	are used to construct the carry-bit-dependent difference distribution table (CDDT) for computing the differential probability on modular addition.
	The last is the SAT-based method proposed in \cite{DBLP:journals/tosc/SunWW21}, where the CNF representation of Lipmaa-Moriai  formulas is used to build a SAT model for searching trails with  differential probabilities less than a target value. The authors used the sequential encoding method \cite{DBLP:conf/cp/Sinz05} to encode the probability constraint and added  Matsui's bounding constraints to the SAT model to speed up the search process with low-round optimal probabilities. The SAT solver CaDiCaL \cite{BiereFazekasFleuryHeisinger-SAT-Competition-2020-solvers} was used to solve the SAT models for searching  optimal trails.
	These four methods can effectively search differential trails of ARX ciphers, and help cryptographic researchers to analyze  security of them.
	From the results for optimal differential trails of the SPECK family in \cite{DBLP:journals/tosc/SunWW21}, it seems that the SAT- and SMT-based methods are more efficient and more suitable for characterizing differential propagation of the modular addition than other two methods. 

	However, the independence assumption sometimes does not hold, especially when the output of one modular addition function behaves as the input of another one in the differential trail. This case is called  consecutive modular addition (CMA) in this paper.
	Since the output of the first addition operation may not be uniform random when it satisfies the differential propagation, the probability of the differential propagation on the second operation calculated by the previous methods may be inaccurate and even incorrect.
	It has been found in \cite{SAC:WanKeldun07,AC:Leurent12} that 
	some published differential trails that include CMAs are in fact impossible 
	because the output of certain addition does not satisfy the differential propagation of its following one. 
	And in \cite{EPRINT:MouPre13}, it was found a case that the probability of a differential of the CMA was higher than it calculated independently.
	Therefore, the probability of the differential trail computed under the independence assumption maybe inaccurate when there are CMAs in the trail.
	It was found in \cite{AC:Leurent12,AFRICACRYPT:ElSAbdYou19} that some trails were impossible because of the contradictory constraints of differential propagation of a CMA on certain adjacent bits of the intermediate state.
	To avoid these impossible trails during the search, Leurent provided the ARX Toolkit \cite{AC:Leurent12}, which was based on finite state machines and considered the differential constraints of modular additions on several adjacent bits of the input and output states.	
	Recently, ElSheikh et al. \cite{AFRICACRYPT:ElSAbdYou19} 
	introduced an MILP-based search model with a new series of variables to represent the constraints on two adjacent bits.
	The two methods are effective to avoid those impossible trails,
	but both are limited by the scale of the problem and not efficient for large models. 
	Moreover, as can be found from the experiments in \cite{AC:Leurent12}, there are still lots of impossible trails that can not be recognized by their methods and the reasons for invalidity of the trails need further investigation.

	Since it is hard to avoid  impossible differentials of CMAs during the search, it is necessary to verify the validity of the found differential trails, that is, verify whether it includes any right pairs of inputes following the trail. 
	Following the work of \cite{C:LiuIsoMei20}, Sadeghi et al. \cite{DBLP:journals/dcc/SadeghiRB21} proposed an MILP-based method to depict the state  propagation of a given differential trail of ARX-based ciphers, and experimentally verified whether the trail admitted any right input pairs. 
	By this method, they showed several published RX-trails of SIMECK and SPECK were impossible, and found better and longer related-key differential trails for SPECK. 

	These works have taken into account the impossible differential trails caused by the non-independence of  CMAs, 
	but paid little attention to the influence of non-independence on the differential probabilities. To the best of our knowledge, a systematic method to accurately calculate probabilities of differential trails considering non-independence of modular additions is still missing in the literature.
	This paper is devoted to this problem. We try to figure out the root causes of non-independence between
	CMAs, giving more conditions that can detect impossible differentials and methods to accurate the probability calculation of differential trails. 
	On basis of these, we would like to develop automatic methods and tools to search and evaluate differential trails for ARX ciphers including CMAs,  such as SPECK and Chaskey. It is important to analysis security of these two ciphers, in particular because they are both standardized by ISO. 

	\subsection{Our contributions}
	\label{subsec:contribution}
	Our contributions in this paper are summarized as follows.
	
	\textbf{A theoretic study of differential properties of the single and consecutive modular addition functions.}
	For single modular addition, we use mathematical induction to analyze the differential equation to obtain the constraints of a differential propagation on each bit, 
	which gives a clear explanation for the calculation of differential probability and the non-uniform distribution of the output.		
	For consecutive modular additions, we theoretically reveal several cases of non-independence by comparing the differential constraints of each addition on the intermediate state, 
	including the impossible difference caused by contradictory constraints and inaccurate probability calculation  caused by the non-uniform distribution of the intermediate state.
	
	\textbf{New models for verification and probability calculation of  CMAs.}
	Given a differential trail of a CMA, we develop a SAT model to capture the state transition of each modular addition following the trail, 
	and apply the \#SAT method to calculate its accurate differential probability. 
	For a differential trail of ARX ciphers including CMAs, we verify its validity by applying the SAT method to capture the state transition and find the state that satisfies the trail quickly, 
	and estimate its differential probability by applying the \#SAT model to calculate the probability of CMAs in the trail, which is more accurate than previous methods under the independence assumption.

	\textbf{Better related-key differential (RKD) trails for SPECK.}
	We apply the SAT method to  verify the validity of  related-key differential trails of round reduced SPECK obtained under the independence assumption before. 
	For valid trails, we apply the \#SAT method to re-estimate their differential probabilities, 
	and for invalid ones, we give a method to automatically detect differences that cause contradictory constraints and record them to avoid this kind of impossible trails in the  search of new valid trails, which reduced a lot of search time.
	Our approach can find valid RKD trails as well as their weak keys, with higher probabilities than the work in \cite{DBLP:journals/dcc/SadeghiRB21}.
	For SPECK32/64, we find the optimal RKD trails of 10 to 13 rounds for the first time. The probabilities of trails of 11 and 15 rounds we find are $2^3$ and $2^9$ times higher than the exiting results.
	For SPECK48/96, we find optimal RKD trails of 11 and 12 rounds for the first time, whose probabilities are $2$ and $2^4$ times higher than  previous results, and the trails of 14 and 15 rounds  have probabilities $2^{1.5}$ and $2^{5.5}$ times higher.
	For SPECK64/128, the RKD trails of 14 and 15 rounds have a probability $2^{16}$ higher than the exiting results. The detailed results are sorted in Table \ref{tab:speck32/64}, Table \ref{tab:speck48/96}, and Table \ref{tab:speck64/128}.
	
	\textbf{A tighter security bound for Chaskey.}
	We apply our methods to analysis the differential trail of 8 rounds for Chaskey, which  was found by the designers in \cite{mouha2014chaskey}.
	We assume that the two  chains of modular additions in the rount function are independent of each other, and use the \#SAT method to calculate the differential probabilities of two chains respectively. 
	We find that the probability of this trail computed under the consideration of  non-independence of CMAs is $2^{-285.1713}$, which is  $2^{7.8}$ times higher than the value $2^{-293}$ calculated under the independence assumption. This provides a tighter security bound for Chaskey against differential attacks.

	\subsection{Outline} 
	\label{subsec:outline}
	The rest of the paper is organized as follows. 
	Sect. \ref{sec:preliminaries} recalls some definitions that will be used in the remaining contents.
	Sect. \ref{sec:diff property of mdad and consecutive mdad} introduce the differential properties of single modular addition and consecutive modular additions, including the influence of non-independent CMAs.
	Sect. \ref{sec:SAT model} presents the SAT model of states propagation on  modular additions, which is used to verify the validity of differential trails and do  accurate calculation of differential probabilities on CMAs.
	The application of our method on the SPECK family of block ciphers and Chaskey are shown in Sect. \ref{sec:application on SPECK} and Sect. \ref{sec:application on Chaskey}, respectively. Conclusions are given in Sect. \ref{sec:conclusion}.

	\section{Preliminaries}
	\label{sec:preliminaries}
	\subsection{Notations}
	\label{subsec:notations}
	In this paper, we denote an $n$-bit vector as $ x = (x_{n-1},\ldots, x_1, x_0) $, where $x_0$ is the least significant bit and $x_{i}$ is called the $i$-th bit (bit $i$) of $x$. We use $x[i,j] = \left(x_i,...,x_{j}\right)$ for $j\leq i$ to represent the $j$-th to the $i$-th bit of $x$. 
	In differential cryptanalysis, we denote two $n$-bit input states as $x$ and $x'$, and their difference as $\varDelta x  = (\varDelta x_{n-1},\ldots, \varDelta x_1, \varDelta x_0) = x \oplus x' $. 
	Some conventional operations are listed in Table \ref{talbe: operations}.
	\begin{table}[tp]
		\begin{center}
			\begin{minipage}{320pt}
				\caption{Basic operations}
				\label{talbe: operations}
				\begin{tabular}{ccccccc}
					\toprule
					XOR 		& Modular addition 	& Left rotation & Right rotation & AND & OR & NOT \\
					\midrule
					$\oplus$ 	& $\boxplus $	 	& $\lll$        & $\ggg$    &$\wedge$ & $\vee$ & $\neg$ \\  
					\botrule
				\end{tabular}
			\end{minipage}
		\end{center}
	\end{table}

	For a modular addition $z = x \boxplus y$, we denote the vector of  carry bits by $c = (c_{n-1},\ldots,c_1,c_0)$ where $c_0 = 0$. Then we have
	\begin{align}
		z_0 &= x_0 \oplus y_0,\\
		z_{i} &= x_{i} \oplus y_{i} \oplus c_{i},\\
		c_{i+1} &= x_{i}y_{i} \oplus (x_{i} \oplus y_{i})c_{i}=x_i\oplus y_i \oplus x_{i}y_{i} \oplus x_{i}z_{i} \oplus y_{i}z_{i}  ,\label{recur_carry}
	\end{align} 
	for $i \geq 1$. We use $c_{i+1} = f_{carry}(x_i,y_i,c_i)$ to denote the carry function.
	Note that for any  $0\leq j<i \leq n-1$, $c_i$ is a Boolean function in $c_j$ and $x_k$, $y_k$ for $k=j,j+1,\ldots, i-1$. So we use 
	\begin{equation}
		c_i = f_{carry}\left(x[i-1,j], y[i-1,j],c_j\right) %
	\end{equation}
	to represent the carry function on the $i$-th bit with the initial carry bit $c_j$. 
	
	We denote a differential propagation  on $z = x \boxplus y$ as an event $M:(\varDelta x,\varDelta y)\xrightarrow{\boxplus}\varDelta z$, that is, input differences $\varDelta x,\varDelta y$ propagates to an output difference $\varDelta z$. When the inputs are independent and uniform random, the differential probability of $M$ is denoted as 
	\begin{equation*}
		\label{eq: the probability of M}
		\begin{aligned}
			\Pr(M) &= \Pr((\varDelta x,\varDelta y)\xrightarrow{\boxplus}\varDelta z) \\
			&= \frac{\#\left\{x,y\in \mathbb{F}_2^n: (x\boxplus y)\oplus ((x\oplus \varDelta x) \boxplus (y \oplus \varDelta y)) =\varDelta z \right\}}{2^{2n}}.
		\end{aligned}
	\end{equation*}
	For $i = 0,1,\ldots,n-1$, we use $P_M^i$ to denote the differential probability of bit $i$ under the conditions that the differential propagation on all previous bits are successful.
	
	\begin{figure}[tp]  %
		\begin{center}
			\includegraphics[ width=0.40\textwidth]{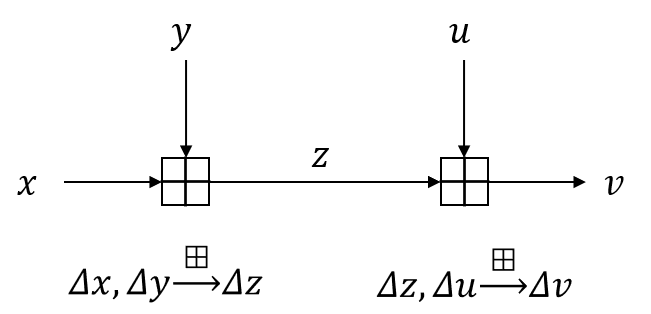}
			\caption{Consecutive modular addition}
			\label{fig:CMA-0}
		\end{center}
	\end{figure}

	The consecutive modular addition (CMA) is depicted in Fig. \ref{fig:CMA-0}, where the two additions are denoted by $z = x \boxplus y$ and $v = z \boxplus u$ and the differential propagation on them are denoted as $M_1:(\varDelta x,\varDelta y)\xrightarrow{\boxplus}\varDelta z$ and $M_2:(\varDelta z,\varDelta u)\xrightarrow{{\small \boxplus}}\varDelta v$, respectively. Note that the output difference $\varDelta z$ of the first addition coincides with the input difference of the second one. 
	The carry vectors of these two additions are $c =x\oplus y\oplus z$ and $d =z\oplus u\oplus v$, respectively.
	Assume the three inputs of the CMA, namely,  $x$, $y$ and $v$, are uniform random. Then the probability of the differential propagation on the CMA can be calculated by $\Pr(M_2,M_1)=\Pr(M_1) \times \Pr(M_2 \vert M_1)$, where $\Pr(M_2 \vert M_1)$ represents the conditional probability of $M_2$ under the condition that $M_1$ happens.  Note that when $M_1$ happens, if the output of the first addition is uniform random in the whole space, then we have $\Pr(M_2 \vert M_1) = \Pr(M_2)$. Given a differential trail, we are interested in the non-independent CMAs where $\Pr(M_2 \vert M_1) \neq \Pr(M_2)$.
	
	\subsection{SAT and \#SAT}
	\label{subsec:SAT and sharpSAT}
	The Boolean satisfiability problem (SAT) studies the satisfiability of a given Boolean formula. A SAT problem is said satisfiable if there exists at least one solution of binary variables that make the formula  true. Although the SAT problem is the first problem that was proved to be NP-complete \cite{DBLP:conf/stoc/Cook71}, the state of art SAT solvers, such as miniSAT, cryprominisat \cite{DBLP:conf/sat/SoosNC09}, CaDiCaL \cite{BiereFazekasFleuryHeisinger-SAT-Competition-2020-solvers}, can solve instances with millions of variables.
	
	The sharp satisfiability problem (\#SAT) \cite{TCS:Valiant97} is the problem that counts the number of solutions to a given Boolean formula. Although the number of solutions can be counted by repeatedly solving the SAT problem and excluding the solutions that have been found until the SAT solver returns \textit{UnSAT}, it is ineffective in solving large-scale models. There are two modern solvers: SharpSAT \cite{DBLP:conf/sat/Thurley06} which is based on modern DPLL based SAT solving technology, and GANAK \cite{DBLP:conf/ijcai/SharmaRSM19} which is a new scalable probabilistic exact model counter. They can both deal with large-scale problems.
	
	Almost all modern SAT and \#SAT solvers take the conjunctive normal form (CNF) of Boolean formulas as inputs. For example, the CNFs of several simple boolean functions are:

	\begin{equation}
		\label{eq: cnf of a oplus b  } 
		a\oplus b \Longleftrightarrow \ \left(\lnot a\vee\lnot b\right)\wedge\left(a\vee b\right),
	\end{equation}
	\begin{equation}
		\label{eq: cnf of a = b}  a=b \Longleftrightarrow \left(\lnot a \vee b\right) \wedge \left(a \vee \lnot b\right),
	\end{equation}
	\begin{equation}
		\label{eq: cnf of g = a oplus b} 
		\begin{aligned}
			g=a \oplus b \Longleftrightarrow & \left(\lnot a \vee b \vee g\right) \wedge \left(a \vee b \vee \lnot g\right)  \wedge \left(a \vee \lnot b \vee g\right) \\
			&  \wedge \left(\lnot a \vee \lnot b \vee \lnot g \right) ,
		\end{aligned}
	\end{equation}
	
	\begin{equation}
		\label{eq: cnf of g = ab} 
		g=ab \Longleftrightarrow \left(\neg g \vee a\right) \wedge \left(\neg g \vee b\right) \wedge \left(g\vee \neg a \vee \neg b\right).
	\end{equation}
	
	\subsection{A brief description of SPECK}
	\label{subsec:SPECK}
	\begin{figure}[tp]  %
		\centering
		\includegraphics[ width=0.50\textwidth]{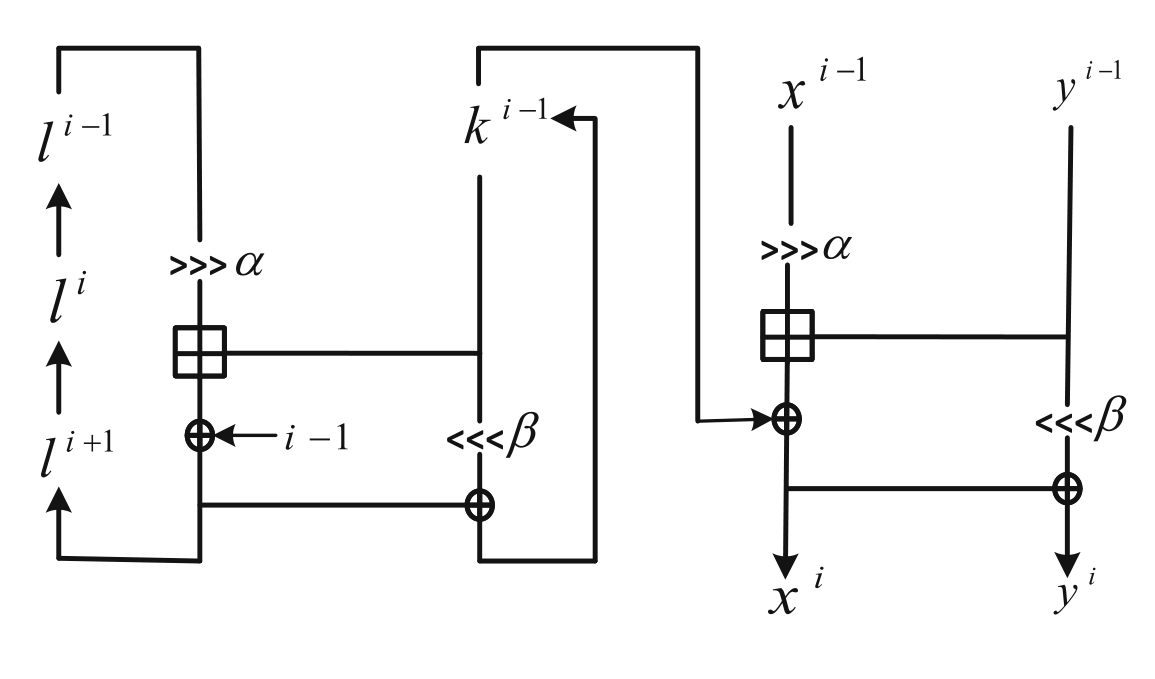}
		\caption{One round of SPECK(for $m =4$)}
		\label{fig:speck}
	\end{figure}    
	SPECK is a family of lightweight block ciphers designed by NSA in 2015 \cite{DBLP:conf/dac/BeaulieuSSTWW15}. For SPECK$2n/mn$,  the  block size is $2n$ bits, while the master key size is $mn$ bits. 
	The round function  and the key schedule function  of SPECK are shown in Fig.\ref{fig:speck}. Let $i\geq 1$. For round $i-1$, the input states and the sub-key are denoted as $(x^{i-1}, y^{i-1})$ and $k^{i-1}$ respectively, then the output state is computed as
	\begin{equation*}
		\left( x^{i}, y^{i} \right)  = \left(\left(\left(x^{i-1} \ggg \alpha\right) \boxplus y^{i-1}\right)\oplus k^{i-1}, \left(y^{i-1}\lll \beta\right) \oplus x^{i} \right),
	\end{equation*}
	where $(\alpha, \beta) = (7,2)$ for $n = 16$ and $(\alpha, \beta) = (8,3)$ in other cases. The round function of key schedule is same to the round function, except that the inserted sub-keys are replaced by round constants
	. Taking $K = (l^{m-2},l^{m-3},\ldots , l^{0},k^{0})$ as the master key, for the $(i-1)$-th round of the key schedule, the output  is
	\begin{equation}
		\left( l^{i+m-2}, k^{i}\right)  = \left(\left(\left(l^{i-1} \ggg \alpha\right) \boxplus k^{i-1}\right)\oplus r^{i-1}, \left(k^{i-1}\lll \beta\right)\oplus l^{i+m-2} \right),
	\end{equation}
	where $r^{i-1} = i-1$.

	Note that in the data encryption part, there is a sub-key insertion between two modular additions of consecutive rounds, so  the input of the second addition can be considered as uniform random, while in the key schedule part, there is only linear operations between two additions, which means the input of the second addition is related to the first one and thus can not be regarded as uniform random. 
	In this paper, we set $m=4$ and study the related-key differential trails for SPECK32/64, SPECK48/96, and SPECK64/128 in Sect. \ref{sec:application on SPECK}.
	
	\subsection{A brief description of Chaskey}
	\label{subsec:Chaskey}
	\begin{figure}[tp]  %
		\centering
		\includegraphics[ width=0.80\textwidth]{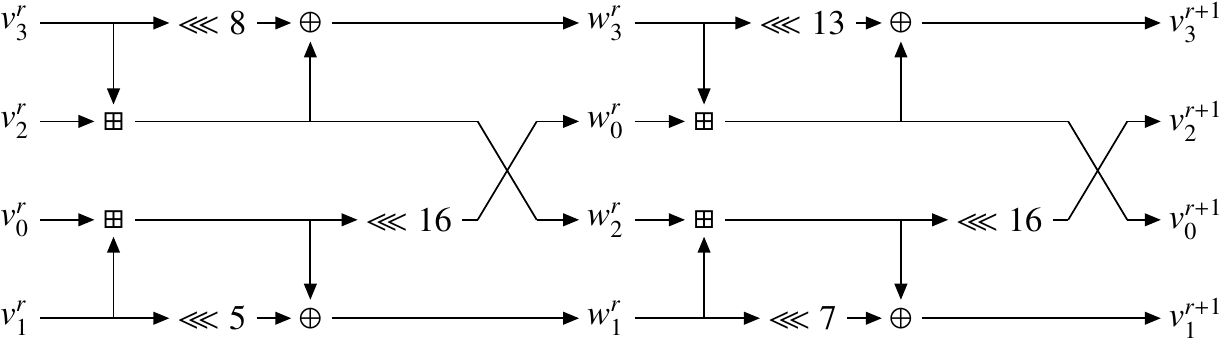}
		\caption{One round of the Chaskey permutation $\pi$ (12 rounds in full)}
		\label{fig:Chaskey}
	\end{figure}
	Chaskey, designed by Mouha et al. \cite{mouha2014chaskey}, is a very efficient MAC algorithm for microcontrollers and was standardized in ISO/IEC 29192-6. The internal permutation $\pi$ of Chaskey has totally 12 rounds, and one of them is depicted in Fig. \ref{fig:Chaskey}. 
	It is obvious that there are two CMAs in one round of  $\pi$, which are
	$w_2^r=v_2^r\boxplus v_3^r$, $v_2^{r+1}\ggg 16=w_2^r\boxplus w_1^r$ and 
	$w_0^{r}\ggg 16=v_0^r\boxplus v_1^r$, $v_0^{r+1}=w_0^r\boxplus w_3^r$. 
	For the entire $\pi$, there are two chains of modular additions with 24 additions in each one. We call the addition chain started with $(v_0^0, v_1^0)$ as Chain 1, and the other one as Chain 2.
	
	Since  Chain 1 and Chain 2 both contain 12 pairs of CMAs which are connected by linear operations (rotations),  in  differential cryptanalysis the non-independence of them will also definitely affect the entire differential trail, and it is important to find out this influence. The designer of Chaskey reported a best found differential trail of 8 rounds in \cite{mouha2014chaskey}, whose differential probability is $2^{-289.9}$ calculated by Leurent's ARX Toolkit. We find it has a higher probability considering the non-independence of CMAs in Sect. \ref{sec:application on Chaskey}.

	\section{Differential properties of the (consecutive) modular addition function}
	\label{sec:diff property of mdad and consecutive mdad}
	\subsection{Differential properties of a single modular addition }
	\label{subsec:modular addition}
	For one $n$-bit modular addition $z=x \boxplus y $, assume that the input variables $x$ and $y$ are independent and uniform random. Given a differential propagation  $M: (\varDelta x,\varDelta y)\xrightarrow{\boxplus}\varDelta z$ for fixed differences $\varDelta x,\varDelta y,\varDelta z\in \mathbb{F}_2^n$, let $x' = x \oplus \varDelta x$, $y' = y \oplus \varDelta y$ and $z' = x'\boxplus y'$. Then the differential equation of $M$ is refereed to as $z\oplus z^\prime=\varDelta z$. Let the carry vectors of $z=x \boxplus y $ and $z' = x'\boxplus y'$ be $c$ and $c'$, respectively.

	\subsubsection{Differential equation on each bit}
	\label{subsubsec:DEB}  
	We first derive the constraints on the input and output of a modular addition function (called differential constraints) under the differential propagation $M$  by studying the differential equation in the bit level. For $i = 0,1,\ldots, n-1$, the output of $z=x \boxplus y $ and $z' = x'\boxplus y'$ satisfies
	\begin{equation}\nonumber
		\begin{aligned}
			&z_i = x_i \oplus y_i \oplus c_i, \\
			&z'_i = x'_i \oplus y'_i \oplus c'_i = x_i \oplus y_i  \oplus \varDelta x_i \oplus \varDelta y_i \oplus c'_i,
		\end{aligned} 
	\end{equation}
	so the differential equation $z\oplus z' = \varDelta z$ on this bit is
	\begin{equation}
		\label{eq:initial diff equation on bits}
		c_i \oplus c'_i = \varDelta x_i \oplus \varDelta y_i \oplus \varDelta z_i .
	\end{equation}
	Then the propagation rule of the difference on each bit can be analyzed as follows. Note here that the main argument comes from the recursive relation \eqref{recur_carry} for the carry function, which makes the differential equation can be analyzed recursively.
	
	For the least significant bit, since $c_0=c'_0= 0$, the differential equation is a constant equation, namely,
	\[\varDelta x_0 \oplus \varDelta y_0 \oplus \varDelta z_0 = 0.\]
	So $x_0$, $y_0$ have no constraints and the differential probability is $ P_M^0 = 1_{\varDelta z_0 = \varDelta x_0 \oplus \varDelta y_0 }$. 	
	
	For any $0\leq i\leq n-2$, assume the differential equation \eqref{eq:initial diff equation on bits} has been solved from the 0-th to the $i$-th bit and the solutions are $\{(x_{i-1},\ldots,x_0), (y_{i-1},\ldots,y_0) \in \mathbb{F}_2^{i}\}$. 
	We consider the equation on the $(i+1)$-th bit. According to \eqref{recur_carry}, we have
	\begin{equation}\nonumber
		\begin{aligned}
			&c_{i+1}=x_{i} y_{i} \oplus x_{i} \oplus y_{i} \oplus\left(x_{i} \oplus y_{i}\right) z_{i}, \\
			&c_{i+1}^{\prime}=x_{i}^{\prime} y_{i}^{\prime} \oplus x_{i}^{\prime} \oplus y_{i}^{\prime} \oplus\left(x_{i}^{\prime} \oplus y_{i}^{\prime}\right) z_{i}^{\prime}.
		\end{aligned}  
	\end{equation}	
	Substituting $x'_i = x_i \oplus \varDelta x_i$, $y'_i = y_i \oplus \varDelta y_i$,	$z'_i = z_i \oplus \varDelta z_i$ and summing the two equalities, the differential equation (\ref{eq:initial diff equation on bits}) of the ($i+1$)-th bit becomes
	\begin{equation}\label{eq:diff equation on bits}
		\begin{aligned}
			&\left(\Delta x_{i} \oplus \Delta y_{i}\right) z_{i} 
			\oplus\left(\Delta x_{i} \oplus \Delta z_{i}\right)  y_{i} 
			\oplus\left(\Delta z_{i} \oplus \Delta y_{i}\right)  x_{i} \\
			=&\left(\Delta x_{i} \oplus \Delta y_{i}\right) \neg \Delta z_{i}\oplus\Delta x_{i} \Delta y_{i}\oplus\Delta x_{i+1} \oplus \Delta y_{i+1} \oplus \Delta z_{i+1}.
		\end{aligned}
	\end{equation}	
	Note that (\ref{eq:diff equation on bits}) is actually an equation in the variables $x_i$, $y_i$ since $c_i=x_i\oplus y_i\oplus z_i$ only depends on  $x_{i-1},\ldots,x_0, y_{i-1},\ldots,y_0$ (which can be viewed as known in this step). 
	Unless a constant equation, it is a linear equation, the solution of which gives constraints on $x_i$, $y_i$. To simply the analysis in the following, we also view $c_i$ as an variable for $0\leq i\leq n-1$, which naturally satisfies Eq. \eqref{recur_carry}.
	Let $\varDelta c_0 = 0$ and $\varDelta c_{j+1} = \varDelta z_{j+1} \oplus \varDelta x_{j+1} \oplus \varDelta y_{j+1}$ for $j = 0,1,\ldots,i$. 
	Then the constraints on $x_i$, $y_i$, $c_i$ and $c_{i+1}$ derived from Eq. \eqref{eq:diff equation on bits} iterating over all possible differences are shown in Table \ref{tab:diff_constraint}.
	
	We simply explain  Table \ref{tab:diff_constraint} by an example. See Row 8 in Table \ref{tab:diff_constraint}. 
	Given the difference $ (\varDelta x_i, \varDelta y_i, \varDelta z_i ) = (1,0,0)$ and $\varDelta c_{i+1}=\varDelta x_{i+1}\oplus\varDelta y_{i+1}  \oplus \varDelta z_{i+1} = 1$,  Eq. (\ref{eq:diff equation on bits}) is simplified to
	\[z_i\oplus y_i=x_i\oplus c_i=0.\] 
	Then the differential probability is
	\begin{equation}
		\nonumber
		\begin{aligned}
			P_M^{i+1}&=\Pr\left(x_i\oplus c_i=0 \left\vert  c_i=f_{carry}\left(x[i-1,0],y[i-1,0], c_0\right)\right.\right)\\	
			&= \Pr\left(x_i=0\right) \cdot \Pr \left(c_i = 0\right) + \Pr\left(x_i = 1\right) \cdot \Pr\left(c_i = 1\right)\\
			& = \frac{1}{2} \left(\Pr\left(c_i=0\right) + \Pr\left(c_i=1\right)\right)\\
			&=\frac{1}{2} ,
		\end{aligned} 
	\end{equation}
	since the input variable $x_i$  is uniform random. This means in this case, no matter what the probability distribution of $c_i$ is, we always have $P_M^{i+1}=1/2$. 
	Furthermore, we also know that $z_i$ is uniform since $z_i = y_i$, and $c_{i+1}$ has the same distribution with $c_i$ since $c_{i+1} = x_iy_i \oplus y_ic_i \oplus x_ic_i  = c_i$, when the differential propagation on this bit is successful.
	
	Table \ref{tab:diff_constraint} is similar to  \cite[Table 2]{AFRICACRYPT:ElSAbdYou19}, but our way to generate it is different.
	
	\begin{table}[tp]   %
		\begin{center}
			\begin{minipage}{\textwidth}
				\caption{Differential  constraints of $(\varDelta x_{i+1},\varDelta y_{i+1}) \rightarrow \varDelta z_{i+1}$} 
				\label{tab:diff_constraint}
				\renewcommand{\arraystretch}{1.6}%
				\renewcommand{\thefootnote}{\fnsymbol{footnote}}
				\begin{tabular*}{\textwidth}{ccccccccc}
					\toprule
					\multirow{2}{*}{} & \multirow{2}{*}{$\varDelta x_{i}$} & \multirow{2}{*}{$\varDelta y_{i}$} & \multirow{2}{*}{$\varDelta z_{i}$} & \multirow{2}{*}{$\varDelta c_{i+1}$} & \multicolumn{3}{@{}c@{}}{ Differential Constraints} & \multirow{2}{*}{$P_M^{i+1} $}   \\
					\cmidrule{6-8}
					&                    &                    &                    &                    &Inputs       &Output          &\quad Carry        &   \\ 
					\midrule
					1 &$a$        &$a$       &$a$       &$a$       &$ -$				&$-$      	&Carry function \eqref{recur_carry}   			  &1   \\
					2 &$a$        &$a$       &$a$       &$\neg a$  &$\times$	    		&$\times$ 		 	&$\times$        &0   \\
					3 &$a$        &$a$       &$\neg a$  &$a$       &$x_{i} = y_{i}$      	&$z_{i}=c_{i}$      &$c_{i+1}=x_{i}=y_{i}$     	&1/2 \\
					4 &$a$        &$a$       &$\neg a$  &$\neg a$  &$x_{i} = \neg y_{i}$ 	&$z_{i}=\neg c_{i}$ &$c_{i+1}=c_{i}=\neg z_{i}$ &1/2	\\
					5 &$a$        &$\neg a$  &$a$       &$a$       &$y_{i} = \neg c_{i}$ 	&$z_{i}=\neg x_{i}$ &$c_{i+1}=x_{i}=\neg z_{i}$	&1/2 \\
					6 &$\neg a$   &$a$       &$a$       &$a$       &$x_{i} = \neg c_{i}$ 	&$z_{i}=\neg y_{i}$ &$c_{i+1}=y_{i}=\neg z_{i}$ &1/2 \\
					7 &$a$        &$\neg a$  &$a$       &$\neg a$  &$y_{i} = c_{i}$      	&$z_{i}=x_{i}$      &$c_{i+1}=c_{i}=y_{i}$      &1/2	\\
					8 &$\neg a$   &$a$       &$a$       &$\neg a$  &$x_{i} = c_{i}$      	&$z_{i}=y_{i}$      &$c_{i+1}=c_{i}=x_{i}$      &1/2	\\ 
					\botrule
				\end{tabular*}
				\footnotetext{ Note: $a$ iterates over $ \left\{0,1\right\}$. 
					``$-$" means there is no constraint; ``$\times$" means this is an unsatisfactory case.}
			\end{minipage}
		\end{center}
	\end{table}
	
	In Table \ref{tab:diff_constraint}, Row 2 is the only invalid case for differential propagation, which corresponds to the case when Eq. \eqref{eq:diff equation on bits} is a constant contradictory equation. Similarly, Row 1 corresponds to the case when Eq. \eqref{eq:diff equation on bits} is a constant equality.  Row 3 to Row 8 show that if the $i$-th bit of the input and output differences  are not all equal, the differential probability $P_M^{i+1}$  is always 1/2. 
	
	From the above discussions, the differential properties of the modular addition operation can be summarized into the following two theorems, as obtained in \cite{FSE:LipMor01,DBLP:journals/dcc/Schulte-Geers13}. They will be used to build the SAT model for capturing differential propagations on modular additions in Sect. \ref{subsubsec:SAT for search}.
	
	\begin{theorem}[\cite{FSE:LipMor01,DBLP:journals/dcc/Schulte-Geers13}] \label{thm:verification of modular addition}
		Assume the two inputs of modular addition are uniform random. Then the differential $M:(\varDelta x,\varDelta y)\xrightarrow{\boxplus}\varDelta z$ is valid if and only if the differences satisfy
		\begin{enumerate}
			\item for the LSB, $\varDelta x_0 \oplus \varDelta y_0 \oplus \varDelta z_0 = 0$;
			\item for $i = 0,1,\ldots, n-2$, if $\varDelta x_i = \varDelta y_i = \varDelta z_i$, then $ \varDelta x_{i+1} \oplus \varDelta y_{i+1} \oplus \varDelta z_{i+1} = \varDelta x_i $.
		\end{enumerate}
	\end{theorem}
	\begin{theorem}[\cite{FSE:LipMor01,DBLP:journals/dcc/Schulte-Geers13}] \label{thm:probability calculation of modular addition}
		Assume the two inputs of modular addition are uniform random. Then if the differential $M:(\varDelta x,\varDelta y)\xrightarrow{\boxplus}\varDelta z$ is valid, its probability is
		\begin{equation}
			\begin{aligned}
				\operatorname{Pr}((\varDelta x, \varDelta y) \xrightarrow{\boxplus}\varDelta z) &={\Large 1}_{\varDelta z_0 = \varDelta x_0 \oplus \varDelta y_0 } \times P_M^{1} \times \cdots \times P_M^{n-1} \\ &=2^{-\sum_{i=0}^{n-2} \neg eq\left(\varDelta x_{i}, \varDelta y_{i}, \varDelta z_{i}\right)}.
			\end{aligned}
		\end{equation}
		Here, the Boolean function $ eq\left(a,b,c\right)$ outputs 1 if and only if $ a = b = c$.
	\end{theorem}
	
	\subsubsection{Non-uniformity of outputs under a differential propagation}
	\label{subsubsec:uneven output}   %

	Under a differential propagation $M$, the differential constraints on input bits of the modular addition function will cause non-uniformity of some output bits. 
	This further results in non-uniform distribution of the output vectors following $M$, for example, certain vectors may never appear as the outputs. 
	In this part, we further analyze the difference Eq. \eqref{eq:diff equation on bits} and Table \ref{tab:diff_constraint}  with a focus on the distribution of the output bit $z_i$ for $0\leq i\leq n-2$ under the differential $M:(\varDelta x, \varDelta y) \xrightarrow{\boxplus}\varDelta z$.  We divide into the following four cases.
	
	\begin{enumerate}
		\item When the differences lie in Row 1 of Table \ref{tab:diff_constraint}, there is no constraint on the output bit, which means $z_i$ is uniform random.
		\item When the differences lie in Row 8 (resp., Row 7) of Table \ref{tab:diff_constraint}, the constraints associate the output bit to one uniform random input bit, namely, $z_i = y_i$ (resp., $z_i = x_i $), and next carry bit to the current bit, namely, $c_{i+1} = c_i = x_i$ (resp., $c_{i+1} = c_i = y_i$). Obviously, $z_i$ is also uniform random.
		\item When differences lie in Row 5 (resp., Row 6) of Table \ref{tab:diff_constraint}, the output bit $z_i = \neg y_i$ (resp., $z_i = \neg x_i $) is uniform random, but the constraints on the carry bit, $c_{i+1} = \neg z_i $, may associate $z_i$ to a higher bit of output, which depends on the differential constraints on higher bits.
		\item When the differences lie in Row 3 (resp., Row 4) of Table \ref{tab:diff_constraint}, the constraints associate the output bit $z_i$ to the carry bit $c_i$ (resp., $\neg c_i$). In this case, distribution of $z_i$ is determined by lower bits, which can be non-uniform. Besides, $z_i$ may be associated to certain lower output bits, which depends on the differential constraints on lower bits.
	\end{enumerate}
	
	Since the output bit $z_i$ in the first three cases is uniform random, we focus on the fourth case. 
	We only discuss the case for Row 3 of Table \ref{tab:diff_constraint} since the discussion for Row 4 is similar. In this case we have $z_i = c_i$.
	
	\begin{table}[tp] %
		\begin{center}
			\begin{minipage}{300pt}
				\caption{Differential constraints associate carry bits: $ z_i = c_i = c_{i-1} = \cdots = c_{j+1}$}
				\label{tab:ci-eq-cj}
				\setlength{\tabcolsep}{13pt}%
				\renewcommand{\thefootnote}{\fnsymbol{footnote}}
				\begin{tabular}{cccccc}
					\toprule
					$M$	    &$i+1$\footnotemark[1]  &$i$    &$i-1$\footnotemark[2]  &$\quad \cdots$\footnotemark[2]   &$j+1$\footnotemark[2]  \\
					\midrule
					$\varDelta x$   &$\cdot$    &$a$     &$\cdot$   &$\quad \cdots$   &$\cdot$ \\
					$\varDelta y$   &$\cdot$    &$a$     &$\cdot$   &$\quad \cdots$   &$\cdot$  \\
					$\varDelta z$   &$\cdot$    &$\neg a$&$a$       &$\quad \cdots$   &$a$  \\
					$\varDelta c$   &$a$        &$\neg a$&$\neg a$  &$\quad \cdots$   &$\neg a$   \\ 
					\midrule
					Constr &$\cdot$&  $z_i=c_{i}$ &\multicolumn{3}{c}{$c_{i} = c_{i-1} =\ \cdots\ = c_{j+2} = c_{j+1}$}    \\
					\botrule		
				\end{tabular}
				\footnotetext{Note: $a$ iterates over $\{0,1\}$. $M$ is valid and  $0\leq j+1 \leq i-1$, when $j+1 = 0$, $a$ satisfy $\varDelta c_0 = 0$.} 
				\footnotetext[*]{ The differences of the $(i+1)$-th bit satisfy $\varDelta x_{i+1} \oplus \varDelta y_{i+1} \oplus \varDelta z_{i+1} =\varDelta c_{i+1} =  a$. }
				\footnotetext[\dagger]{ For any $j+1\leq t \leq i-1$, we set $\varDelta c_{t+1} = \neg a$, $\varDelta z_t = a$, and $(\varDelta x_t, \varDelta y_t) \in \{(\neg a, a), ( a,\neg a)\} $, resulting in  the constraint $c_{t+1} = c_{t} = x_{t}$ (or $= y_{t}$).   }
			\end{minipage}    
		\end{center}
	\end{table}

	Firstly, we study the case  $z_{i} = c_{i} = c_{i-1} = \cdots = c_{j+1}$ where $0\leq j+1\leq i-1 $, i.e., the differential constraints associate $z_i$ with non-adjacent lower bits by connecting some consecutive bits of the carry vector.
	As shown in Table \ref{tab:ci-eq-cj}, for any $j+1 \leq t \leq i-1$, the differences of the $t$-th and the $(t+1)$-th bit lie in Row 7 or Row 8 of Table \ref{tab:diff_constraint} with the constraint $c_{t+1} = c_{t}$ ($= x_{t}$ or $y_{t}$). Then the output bit $z_t$ is uniform and we have 
	\[c_{i} = c_{i-1} = \cdots = c_{j+2} = c_{j+1}.\] 
	The constraints make $z_i$ related to $c_{j+1}$ and independent of the output bits $z_{i-1}, z_{i-2}, \ldots, z_{j+1}$.
	Moreover, it can be learned from Table \ref{tab:ci-eq-cj} that probabilities of differential propagation from the $(j+1)$-th to the $(i-1)$-th bit  are all $1/2$.
	Therefore, the more consecutive carry bits are connected by differential constraints, the lower the overall differential probability is.
	
	In the following part, starting from $z_i= c_{j+1}$ for $0\leq j+1\leq i$, we introduce all  cases of non-uniformity of $z_i$ by analyzing differential constraints on $c_{j+1}$. 
	From Table \ref{tab:diff_constraint}, there are totally four cases, namely,
	\begin{equation} 
		\label{eq: c_j+1}
		c_{j+1} = 
		\begin{cases}
			0 & j+1 = 0, \varDelta c_0 = 0,\\ 
			x_j = y_j & j\geq 0,\ {\rm Row\ 3\ of\ Table\ 2},\\
			\neg z_j  & j\geq 0,\ {\rm Row\ 4,\ 5,\ or\ 6\ of\ Table\ 2}, \\
			f_{carry}(x_{j} ,y_{j} ,c_{j}) & j\geq 0,\ {\rm Row\ 1\ of\ Table\ 2},
		\end{cases}
	\end{equation}
	where conditions of each case except case 1 are on the differences of the $j$-th bit.
	We first find that in the second case where differences of the $j$-th bit lie in Row 3 of Table \ref{tab:diff_constraint}, we have $z_{i} = c_{j+1}= x_{j}= y_{j}$, thus $z_i$ is uniform random. 
	We also remark that $z_j$ could be non-uniform due to $z_j = c_j$ in this case, and the analysis of it is the same as (\ref{eq: c_j+1}) and will eventually fall into the other three cases.
	
	The other three cases are the sources of non-uniformity of $z$.

	\textbf{One output bit is restricted to a certain value.} This is the first case of (\ref{eq: c_j+1}), in which the differential constraints restrict $z_i$ to a fixed value, that is, $z_i = c_i = \cdots = c_0 = 0$, so the output is non-uniform. 

	\textbf{Two output bits are restricted by differential constraints.} This is the third case of (\ref{eq: c_j+1}), in which the differential constraints associate $z_i$ to a lower bit of the output. As shown in Table \ref{tab:zi-not-eq-zj}, when the differences of the $j$-th bit lie in Row 4, Row 5 or Row 6 of Table \ref{tab:diff_constraint} with the constraint $c_{j+1} = \neg z_{j}$, two output bits $z_i$ and $z_j$ are connected by carry bits, that is, $z_i= \neg z_{j}$.
	This kind of differential constraint relating the value of two output bits may make the distribution of the output $z$ non-uniform, as shown in Example \ref{example: z_i not eq z_j}. 
	In particular, during our search for differential trails, we find the case that two adjacent output bits are related, namely, $z_i = \neg z_{i-1}$ or $z_i = z_{i-1}$, often occurs. 
	\begin{table}[tp] %
		\begin{center}
			\begin{minipage}{315pt}
				\caption{Differential constraints restrict output bits: $z_i = c_{i} = \cdots = c_{j+1} = \neg z_{j}$} 
				\label{tab:zi-not-eq-zj}
				\renewcommand{\thefootnote}{\fnsymbol{footnote}}
				\begin{tabular}{ccccccccccc}
					\toprule
					$M$				&$i+1$\footnotemark[1]     &$i$        &$i-1$\footnotemark[2]   &$\ \cdots$\footnotemark[2]    & $j+1$\footnotemark[2]   &\multicolumn{5}{c}{$j$ \footnotemark[3] }  \\ 
					\midrule
					$\varDelta x$   &$\cdot$     &$a$      &$\cdot$    &$\ \cdots$&$\cdot$ &$\neg a$    & \multirow{4}{*}{or}  & $\ \ a$    &  \multirow{4}{*}{or}   & $a$\\
					$\varDelta y$   &$\cdot$     &$a$      &$\cdot$    &$\ \cdots$&$\cdot$   &$\ \ a $     &  & $\neg a$   &   & $a$   \\
					$\varDelta z$   &$\cdot$     &$\neg a$ &$a$        &$\ \cdots$&$a$       &$\neg a$     &  & $\neg a$   &   & $\neg a$ \\
					$\varDelta c$   &$a$         &$\neg a$ &$\neg a$   &$\ \cdots$&$\neg a$  &$\ \ a $     &  & $\ a$      &   & $\neg a$               \\ 
					\midrule
					\multirow{2}{*}{Constr} &$\cdot$ &$z_i = c_i$ &\multicolumn{3}{c}{$c_{i} = c_{i-1}= \cdots =c_{j+1}$}    &\multicolumn{5}{c}{$c_{j+1} = \neg z_{j}$ }   \\
					&  & & &  &  & \multicolumn{5}{c}{$c_{j} = \neg y_{j}$ or $\neg x_{j}$ or $\neg z_{j}$}  \\ 
					\botrule		
				\end{tabular}
				\footnotetext{ Note: $a$ iterates over $\{0,1\}$. $M$ is valid and $0\leq j \leq i-1 \leq n-3 $.} 
				\footnotetext[*]{ The differences of the $(i+1)$-th bit satisfy $\varDelta x_{i+1} \oplus \varDelta y_{i+1} \oplus \varDelta z_{i+1} =\varDelta c_{i+1} =  a$. }
				\footnotetext[\dagger]{ When $j< i-1$, for any $j+1\leq t \leq i-1$, differences on the $t$-th bit are the same as Table \ref{tab:ci-eq-cj} and $c_t = c_{t-1}$ holds. When $j = i-1$, this part is omitted and we have $z_i = \neg z_{i-1}$. }
				\footnotetext[\ddagger]{ At the $j$-th bit, since $c_{j+1} = \neg a$, there are totally three cases for $c_{j+1} = \neg z_j$. When $j = 0$, the value of $a$ must satisfy $\varDelta c_0 = 0$. }
			\end{minipage}    
		\end{center}
	\end{table}	
	
	\begin{figure}[tp]  %
		\centering
		\includegraphics[ width=1\textwidth]{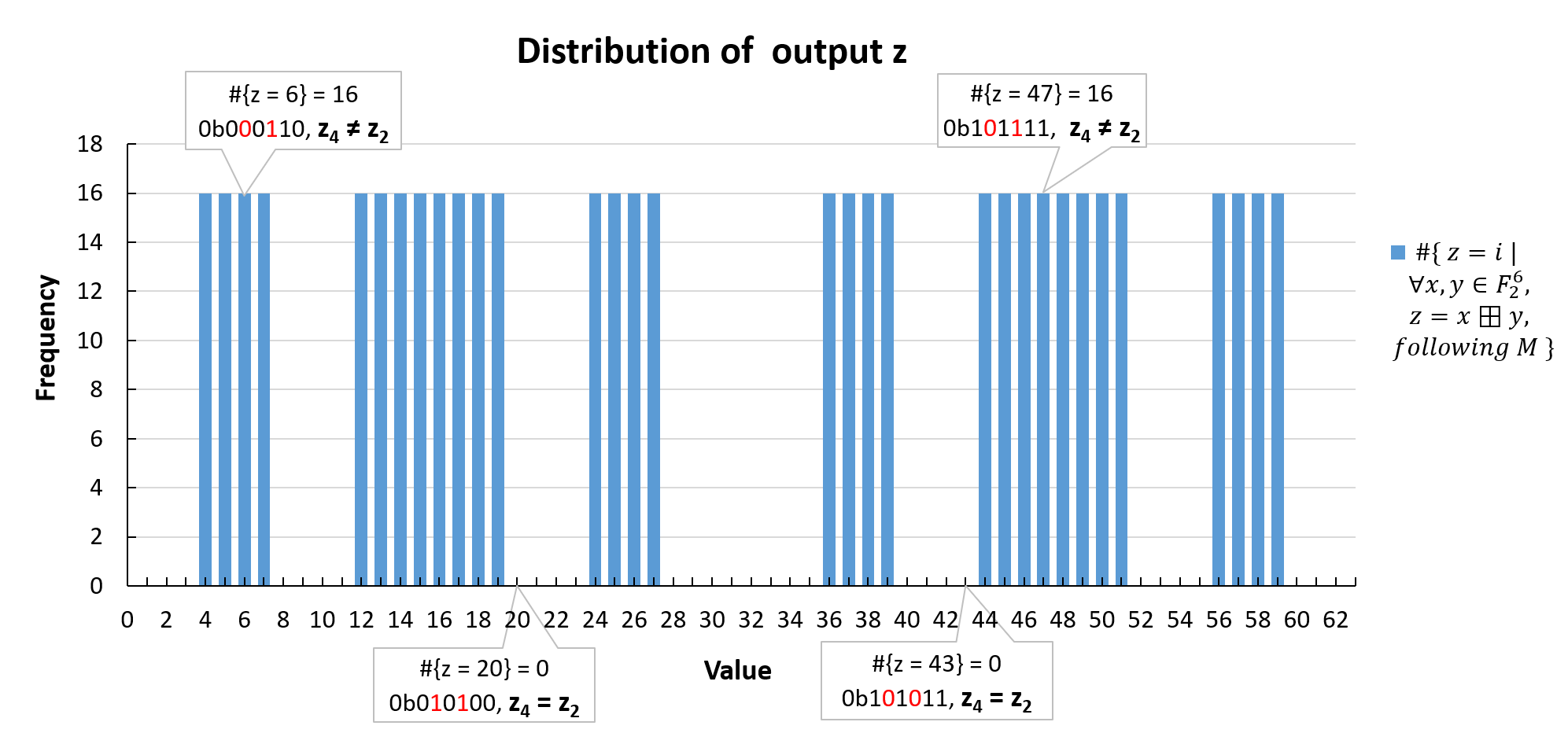}
		\caption{The uneven output caused by two restricted adjacent bits }
		\label{fig:uneven-z-eq-z}
	\end{figure}
	\begin{example} \label{example: z_i not eq z_j}
		Given the differential propagation $M$ with the differences
		\begin{equation}
			\nonumber
			\begin{aligned}
				\varDelta x = \rm 0b000000 ,\\
				\varDelta y = \rm 0b001100 ,\\
				\varDelta z = \rm 0b010100 ,\\
				\varDelta c = \rm 0b011000 ,\\
			\end{aligned}
		\end{equation}	
		the differential probability is $\Pr(M) = \Pr(\neg x_2 = c_2)\times \Pr(y_3 = c_3) \times \Pr(x_4 = y_4) = 2^{-3}$, and the constraints on the outputs are $z_2 = \neg y_2$, $z_3 = x_3$, and $z_4 = c_4 = c_3 = \neg z_2$. So the output bits are independent and uniform random except the case $z_4=\neg z_2$.
		The output distribution is shown in Fig. \ref{fig:uneven-z-eq-z}, where the horizontal and vertical axis are output values and their frequencies, respectively. 
		Obviously, the output $z$ is non-uniform  and it can be found that the frequency of the value satisfying $z_4 = z_2$ is always 0, such as $z = 20 $ and $z=43$, and other values with $z_4 = \neg z_2$ have the same frequency. 
	\end{example}
	
	\begin{table}[tp]
		\begin{center}
			\begin{minipage}{\textwidth}
				\caption{Differential constraints associate the output bit to carry bits}
				\label{tab:zi-eq-cj}
				\setlength{\tabcolsep}{3.2pt}%
				\renewcommand{\arraystretch}{1.3}%
				\renewcommand{\thefootnote}{\fnsymbol{footnote}}
				\begin{tabular*}{\textwidth}{ccccccccc}
					\toprule
					$M$	&$i+1$    &$i$   &$i-1 \cdots j+1$\footnotemark[1]   &$j$  &$\cdots$\footnotemark[2]&$k+1$ &$k$ (case 3)   &$k$ (case 4)\footnotemark[3]   \\ 
					\midrule
					$\varDelta x$   &$\cdot$    &$a$      &$\cdots$   &$\neg a$&$\cdots$&$\neg a$   &$\neg a$   &$\neg a$   \\
					$\varDelta y$   &$\cdot$    &$a$      &$\cdots$   &$\neg a$&$\cdots$&$\neg a$   &$\neg a$   &$a$        \\
					$\varDelta z$   &$\cdot$    &$\neg a$ &$a$        &$\neg a$&$\cdots$&$\neg a$   &$a$        &$\neg a$   \\
					$\varDelta c$   &$a$        &$\neg a$ &$\neg a$   &$\neg a$&$\cdots$&$\neg a$   &$a$        &$a$        \\ 
					\midrule
					\multirow{2}{*}{Constr} &          &\multicolumn{2}{c}{$z_i = c_i = \cdots = c_{j+1}$}                 &$c_{j}$&$\cdots$&$c_{k+1}$        &$c_{k+1} = x_k =  y_{k}$ & $c_{k+1} = x_{k} =\neg z_k$\\
					&          &                  &  &\multicolumn{3}{c}{Carry function}   &$z_{k} = c_{k}$ &$\neg y_k = c_{k}$\\
					\botrule												
				\end{tabular*}
				\footnotetext{ Note: $a$ iterates over $\{0,1\}$. $M$ is valid and $0\leq k+1 \leq j \leq i-1\leq n-2$. }
				\footnotetext[*]{ When $j < i-1$, for any $j+1\leq t \leq i-1$, differences on the $t$-th bit are the same as Table \ref{tab:ci-eq-cj} and $c_t = c_{t-1}$ holds. When $j = i-1$, this part is omitted. }
				\footnotetext[\dagger]{ For any $k+1 \leq t \leq j$, the differences of the $t$-th bit satisfy Row 1 of Table \ref{tab:diff_constraint}. There are no differential constraints on these bits, only carry functions. When $k+1 = 0$, the value of $a$ must satisfy $\varDelta c_0 = 0$. }
				\footnotetext[\ddagger]{ In the case 4 of (\ref{eq:c_k+1}), there are two other cases where the $k$-th bit's difference results in $ c_{k+1} = \neg z_{k}$, similar to the $j$-th bit of Table \ref{tab:zi-not-eq-zj}. When $k = 0$, the value of $a$ must satisfy $\varDelta c_0 = 0$. }
			\end{minipage}
		\end{center}
	\end{table}

	\textbf{One output bit is associated to multiple lower bits.} This is the fourth case  of (\ref{eq: c_j+1}), in which $z_i = c_{j+1} = f_{carry}(x_j,y_j,c_j)$. This case has not been discovered before in the literature. 
	We have $z_i = x_{j} y_{j} \oplus (x_{j} \oplus y_{j}) c_{j}$ and $z_{j} = x_{j} \oplus y_{j} \oplus c_{j}$. Assuming that $ \Pr(c_{j} = 0) = q$, we have $\Pr( z_i = 0 ) = (1+2q)/4$, %
	\begin{equation} \nonumber
		\left\{ 
		\begin{aligned}
			&\Pr( z_i = 0 \vert z_{j} = 1 ) = \frac{1+q}{2},\\
			&\Pr( z_i = 1 \vert z_{j} = 1) = \frac{1-q}{2},
		\end{aligned}
		\right.   \quad {\rm and} \quad
		\left\{ 
		\begin{aligned}
			&\Pr( z_i = 0 \vert z_{j} = 0 ) = \frac{q}{2},\\
			&\Pr( z_i = 1 \vert z_{j} = 0 ) = \frac{2-q}{2}.
		\end{aligned}
		\right.
	\end{equation}
	Obviously, regardless of the distribution of $c_{j}$, $z_i$ is not independent of $z_j$ and the correlation between them will make the distribution of the output $z$ non-uniform.
	So we still regard the $i$-th bit as the source of non-uniformity, even when $q = 1/2$ due to $c_{j} = x_{j-1}$ (or $=y_{j-1}$) and $\Pr( z_i = 0 ) = 1/2$. 
	
	\iffalse
	\begin{table}[tp]
		\begin{center}
			\begin{minipage}{250pt}
				\caption{The correlation between $z_i$ and $z_{j}$}
				\label{tab:1}
				\renewcommand{\arraystretch}{1.6}%
				\begin{tabular}{cccc}
					\toprule
					& $z_{j} = 0$   & $z_{j} = 1$   &                  \\
					\midrule
					$z_i = 0$ & $\frac{q}{4}$   & $\frac{1+q}{4}$ & $\frac{1+2q}{4}$ \\
					$z_i = 1$ & $\frac{2-q}{4}$ & $\frac{1-q}{4}$ & $\frac{3-2q}{4}$ \\
					\midrule
					& $\frac{1}{2}$   & $\frac{1}{2}$   &                  \\
					\bottomrule
				\end{tabular}
			\end{minipage}
		\end{center}      
	\end{table}
	\fi
	Generally, for $0\leq k+1 \leq j$, assume that differences of the $(k+1)$-th to $j$-th bit lie in Row 1 of Table \ref{tab:diff_constraint} with no constraints, and if $k\geq 0$, the $k$-th bit is the first one which has differential constraints in the lower bits, as shown in Table \ref{tab:zi-eq-cj}.
	In this case, the output bit $z_i$ are restricted by a $(j-k)$-bit carry function with the initial carry bit $c_{k+1}$, that is,
	\begin{equation}
		\label{eq: c_j - c_k+1}
		z_i = c_i = \cdots = c_{j+1} = f_{carry} \left(x[j,k+1] , y[j,k+1], c_{k+1}\right),
	\end{equation}
	For another scenario $z_i = \neg c_i$, the differential is similar to Table \ref{tab:zi-eq-cj} except that $\varDelta c_{i+1} = \neg a$. According to the carry function, $z_i$ is not independent of $\left( z_j,...,z_{k+1}\right)$, in other words, associated to those bits and their correlation makes the output $z$ uneven. 
	The distribution of $c_{k+1}$, similar to \eqref{eq: c_j+1}, can be divided into the following four cases,
	\begin{equation}
		\label{eq:c_k+1}
		c_{k+1} = 
		\begin{cases}
			0 & k+1 = 0,\\ 
			c_{k}\ (= x_{k}\ {\rm or}\ y_{k}) & k>0,\ {\rm  Row\ 7\ or\ Row\ 8\ of\ Table\ 2},\\
			x_k = y_k & k\geq 0,\ {\rm  Row\ 3\ of\ Table\ 2},\\
			\neg z_k  & k\geq 0,\ {\rm Row\ 4,\ 5,\ or\ 6\ of\ Table\ 2}, \\
		\end{cases}
	\end{equation}
	where conditions of each case except case 1 are for the differences of the $k$-th bit.
	Consider the second case, we have constraints $c_{k+1} = c_{k}= x_{k}$ (resp., $c_{k+1} = c_{k}= y_{k}$) and $z_k = y_k$ (resp., $z_k =x_k$), thus $z_k$ is uniform random and independent of $z_i$. Moreover, since $c_k = f_{carry}(x_{k-1},y_{k-1},c_{k-1})$, $z_i$ is related to more lower bits, namely, $(z_j,\ldots,z_{k+1},z_{k-1})$, and the analysis on $c_{k-1}$ is the same as (\ref{eq:c_k+1}). 

	For the third case $c_{k+1} = x_k = y_k$ as shown in Table \ref{tab:zi-eq-cj}, we have $\Pr(z_i = 0) = 1/2$ since $\Pr(c_{k+1}= 0) = 1/2$, and $z_i$ is related to $\left( z_j,...,z_{k+1}\right)$ but independent of $z_{k}$ due to the constraint $z_{k} = c_{k}$. So the output is uneven.
	
	As for case 1 where $k+1 =0$ and case 4 where the differences of the $k$-th bit ($k\geq 0$) lie in Row 4, Row 5, or Row 6 of Table \ref{tab:diff_constraint}, the differential constraints associate $z_{i}$ to $\left( z_j,...,z_{1}, z_0\right)$ and $\left( z_j,...,z_{k+1}, z_{k}\right)$ with
	\begin{align}
		&z_i = c_i = \cdots = c_{j+1} = f_{carry} \left(x[j,0] , y[j,0], c_0\right), \label{eq: z_i = c_j+1 = c_0} \\
		&z_i = c_i = \cdots = c_{j+1} = f_{carry} \left(x[j,k+1] , y[j,k+1], \neg z_{k}\right),\label{eq: z_i = c_j+1 = z_k} 
	\end{align}
	respectively. The distribution of the output $z$ in both cases is non-uniform random, but they differ from case 3 in that certain values of the output will never appear when following the differential, as shown in Proposition \ref{prop: neg z_i = c_j}. Let $i = j+1$ and $k = 0$, we give the simplest and most common example to show the uneven output. 
	\begin{proposition}
		\label{prop: neg z_i = c_j}
		When the differential constraints of $M$ are Eq. (\ref{eq: z_i = c_j+1 = c_0}), the output satisfies that if $z_j = z_{j-1} = \cdots = z_{0} = 1$, then $z_i = 0$ ($z_i = 1$ for another scenario $z_i = \neg c_i$). 
		When the differential constraints of $M$ are Eq. (\ref{eq: z_i = c_j+1 = z_k}), the output satisfies that if $z_j = z_{j-1} = \cdots = z_{k+1} = z_{k}$, then $z_i = \neg z_{j}$ ($z_i = z_{j}$ for another scenario $z_i = \neg c_i$).
	\end{proposition} 
	\begin{proof}	
		For Eq. (\ref{eq: z_i = c_j+1 = c_0}), since $c_0 = 0$, it can be derived from the carry function of several consecutive bits that if $z_{j} = \cdots = z_{1} = z_0 = 1$, the carry bits on the corresponding bits satisfy $ c_{j+1} = c_{j} = \cdots = c_{1} = c_{0} = 0$. Therefore, there is $ z_i = 0$ if $ z_i = c_{j+1}$ and $ z_i = 1$ if $ z_i = \neg c_{j+1}$.
		
		For Eq. (\ref{eq: z_i = c_j+1 = z_k}), the proof is similar. If $z_{j} = \cdots = z_{k+1} = z_k = 0 $ (or $= 1$), the carry bits on the corresponding bits satisfy $ c_{j+1} = c_{j} = \cdots = c_{k+1} = \neg z_{k} = 1 $ (or $= 0$). So $c_{j+1} = \neg z_{j}$ under those conditions, and there is $ z_i = \neg z_{j}$ if $ z_i = c_{j+1}$ and $ z_i = z_{j}$ if $ z_i = \neg c_{j+1}$. 
	\end{proof}
	
	Note that for the scenario consists of case 4 and case 2, Proposition \ref{prop: neg z_i = c_j} needs some modification. Assume that there is one bit in case 4 (Eq. (\ref{eq: z_i = c_j+1 = z_k})), for example, the $t$-th bit ($k+1 \leq t < j$), whose differential constraint is $c_{t+1} = c_{t}$, then $z_i$ is associated to $\left( z_j,...,z_{t+1}, z_{t-1},\ldots, z_{k}\right)$, and the modification is if $z_j = \cdots =z_{t+1}= z_{t-1}= \cdots = z_{k}$, then $z_i = \neg z_{j}$. The same applies to the combination of case 2 and case 1.
	
	\begin{example}
		Given the differential propagation $M$ with the difference:
		\begin{equation}
			\nonumber
			\begin{aligned}
				\varDelta x = \rm 0b010000, \\
				\varDelta y = \rm 0b010000, \\
				\varDelta z = \rm 0b000000, \\
			\end{aligned}
		\end{equation}
		the differential probability is $\Pr(M) = \Pr(x_4 = \neg y_4) =1/2$, and the constraints on output is $ z_4 = \neg c_4 = \neg f_{carry}(x[3,0],y[3,0],c_0)$. 
		The output distribution is shown in Fig. \ref{fig:uneven-z-eq-c}, and the distribution of each output bit except $z_4$ is shown in Table \ref{tab:the bits of uneven bits}. 
		It can be found from Table \ref{tab:the bits of uneven bits} that all bits of the output are independent of each other and have uniform random values, except for $z_4$. 
		From Fig. \ref{fig:uneven-z-eq-c}, it can be learned that:
		\begin{enumerate}
			\item The output distribution is uneven.
			\item The output has the same distribution on the intervals [0,31] and [32,63], indicating that $z_5$ is independent of other bits.
			\item From the distribution on [0,15] and [16,31], it is obvious that the probability of  $z_4 = 0$ changes with the values of the lower four bits. For example, $z=6$ and $z = 22$ have different frequencies, which means $\Pr(z_4 = 0) \neq \Pr(z_4 = 1)$ when $(z_3,z_2,z_1,z_0) = (0,1,1,0)$.
			Therefore, $z_4$ related to $(z_3,z_2,z_1,z_0)$ is non-uniform random.
			\item The frequency of $z = 15$ (0b001111) and  $z = 47$ (0b101111) is 0, which means that if $z_0=z_1=z_2=z_3=1$, the value of $z_4$ must be 1.
		\end{enumerate}
	\end{example}
	\begin{figure}[tp]  %
		\centering
		\includegraphics[ width=1\textwidth]{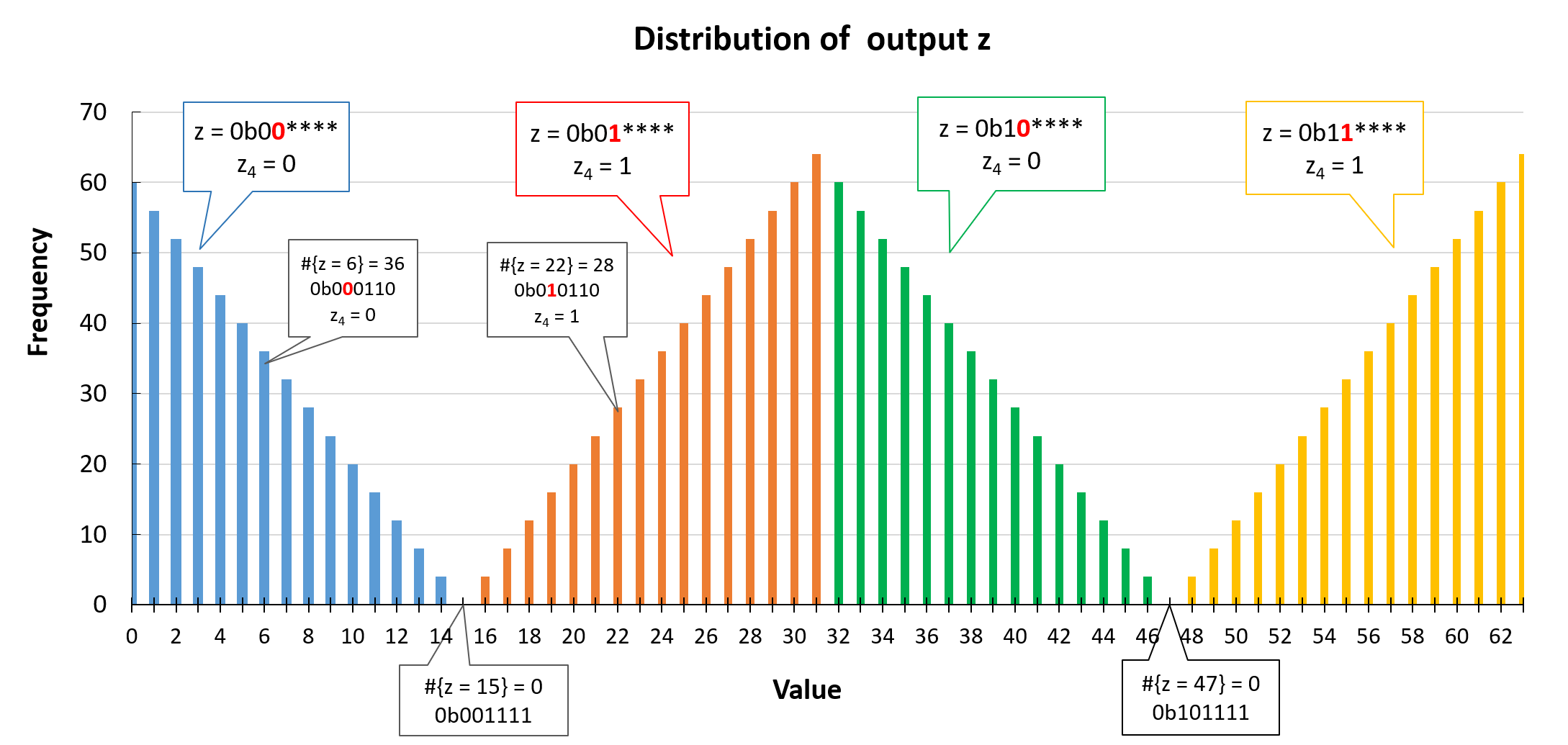}
		\caption{The uneven output of the modular addition following differential $M$}
		\label{fig:uneven-z-eq-c}
	\end{figure}
	\begin{table}[tp]
		\begin{center}
			\begin{minipage}{290pt}
				\caption{The probability that each bit of the output takes a value of zero }
				\label{tab:the bits of uneven bits}
				\begin{tabular}{cccccc}
					\toprule
					Bit $i$	& 5    & 3    & 2    & 1    & 0    \\ 
					\midrule
					$\#\{z: (\varDelta x,\varDelta y)\xrightarrow{\boxplus}\varDelta z, z_i = 0 \}$& 1024 & 1024 & 1024 & 1024 & 1024 \\ 
					$\Pr(z_i = 0)$& 0.5  & 0.5  & 0.5  & 0.5  & 0.5  \\\botrule
				\end{tabular}
			\end{minipage}
		\end{center}
	\end{table}
	
	In a summary, the uneven output of the addition following a differential can be divided into the following three part: $z_i  = c_i= \cdots = c_0 = 0$, $z_i = c_i= \cdots = c_j =\neg z_j$, and $z_i = c_i= \cdots = c_j = f_{carry}(x[j,k+1],y[j,k+1], c_{k+1})$
	where $0 \leq k < j \leq i-1$ and the cases of $z_i = \neg c_i$ are similar.
	Therefore, it is the differential constraints $z_i = c_i$ and $\neg z_i= c_i$ that cause the distribution of the output uneven, and we call the output bit with these two constraints as the non-uniform random bit.
	
	\begin{lemma}
		\label{lemma: uneven output}
		If the differences of a valid differential $M: (\varDelta x,\varDelta y)\xrightarrow{\boxplus}\varDelta z$ satisfy that $\exists\ 0\leq i < n-1$ and $a \in \{0,1\}$, $(\varDelta x_i, \varDelta y_i, \varDelta z_i) = (a,a,\neg a)$, then the output state is uneven. %
	\end{lemma}
	
	The differences mentioned above are the basic cases that cause the uneven output, and a differential $M$ may includes several of them at the same time.
	When these uneven outputs is accepted as the input of another modular addition, it may cause inaccurate calculation of its differential probability and even make the difference to fail to propagate.

	\subsection{Differential properties of the consecutive modular addition}
	\label{subsec:consecutive modular addition}  %
	
	The consecutive modular addition is depicted in Fig. \ref{fig:CMA-0}. Two $n$-bit modular additions are $z = x \boxplus y $ and $v = z\boxplus u$ with differential trails $M_1: (\varDelta x,\varDelta y)\xrightarrow{\boxplus}\varDelta z$ and $M_2: (\varDelta z,\varDelta u)\xrightarrow{\boxplus}\varDelta v$ and carry vectors $c = x\oplus y \oplus z$ and $d = z\oplus u \oplus v$ respectively.
	The intermediate state $z$ connects two additions, and the difference $\varDelta z$ has only one value. The values of three inputs $x$, $y$ and $v$ of the CMA are assumed to be uniform random.
	As the output of the first addition is the input of the second, the probability of $M_2$ is $\Pr(M_2\vert M_1)$. There are $\Pr(M_2)= \Pr(M_2\vert M_1)$ in the following three cases.
	\begin{enumerate}
		\item Following $M_1$, the output of the first addition is uniform random.
		\item If there are non-uniform random bits of $z$ following $M_1$, they are not included in the differential equations (constraints on inputs) of $M_2$.
		\item If the non-uniform random bit $z_i$ are included in the differential equation of $M_2$, this equation must be $z_i \oplus u_j = 0$ or $z_i \oplus \neg u_j = 0$ ($j\leq i$) with the probability of 1/2 since $u_j$ is uniform random.
	\end{enumerate}

	In this section, we introduce two important effects of non-independence on differential propagation of the CMA.
	The most obvious one is impossible trails caused by the contradictory differential constraints imposed by $M_1$ and $M_2$ on the intermediate state.
	Another one is the error between the true probability of $M_2$ and that calculated under the independence assumption when the output $z$ of the first addition is uneven. %
	
	\subsubsection{Invalid differential trails caused by non-independence}
	\label{subsubsec:impossible difference} %
	
	According to Sect. \ref{subsubsec:uneven output}, there are three kinds of uneven output state of $M_1$, and when they are accepted as the input state of $M_2$, the conflict between the constraints of two modular additions will cause the differential propagation to be invalid. Three basic cases are introduced below. 
	
	The first case is that the output bit $z_i$ is restricted to a certain value by the constraint $z_i = c_i = \cdots = c_0 = 0$ (resp., $z_i = \neg c_0 = 1$) of $M_1$, while the differential equation of $M_2$ is $ z_i = \neg d_i = \cdots = \neg d_0 = 1$ (resp., $z_i = d_0 = 0$). It is obvious that the differential trail of the CMA is invalid.
	
	\begin{table}[tp] %
		\begin{center}
			\begin{minipage}{\textwidth}
				\caption{Conflicting constraints of the CMA} 
				\label{tab:conflicting constraints 1}
				\setlength{\tabcolsep}{3.0pt}%
				\renewcommand{\arraystretch}{1.2}%
				\renewcommand{\thefootnote}{\fnsymbol{footnote}} 
				\begin{tabular}{ccccccclccccccc}
					\cline{1-7} \cline{9-15}
					$M_1$ &$i+1$\footnotemark[1]   &$i$ &$i-1$\footnotemark[2] &$\cdots$&$j+1$\footnotemark[2]  &$j$\footnotemark[3]        &  &$M_2$  & $i+1$\footnotemark[1]   &$i$     &$i-1$\footnotemark[2] &$\cdots$&$j+1$\footnotemark[2]      &$j$\footnotemark[3]          \\ 
					\cline{1-7} \cline{9-15}
					$\varDelta x$   &$\cdot$    &$a$     &$\cdot$ &$\cdots$&$\cdot$ &$\neg a$       & 
					&$\varDelta z$  &{\color[HTML]{CB0000} $b$}       & {\color[HTML]{CB0000} $\neg a$} &{\color[HTML]{CB0000} $a$} & {\color[HTML]{CB0000} $\cdots$} &{\color[HTML]{CB0000} $a$} & {\color[HTML]{CB0000} $\neg a$}  \\
					$\varDelta y$   &$\cdot$   &$a$     &$\cdot$  &$\cdots$ &$\cdot$  &$a$&  
					&$\varDelta u$  &$\cdot$   &$a$     &$\cdot$  &$\cdots$ &$\cdot$  &$\neg a$ \\
					$\varDelta z$   &{\color[HTML]{CB0000} $b$}       &{\color[HTML]{CB0000}$\neg a$} &{\color[HTML]{CB0000} $a$} & {\color[HTML]{CB0000} $\cdots$} &{\color[HTML]{CB0000} $a$}  &{\color[HTML]{CB0000}$\neg a$}  &  
					&$\varDelta v$  &$\cdot$   &$a$     &$\cdot$  &$\cdots$&$\cdot$    &$a$    \\
					$\varDelta c$   &$a$       &$\neg a$&$\neg a$ &$\cdots$&$\neg a$  &$a$ &  
					&$\varDelta d$  &$\neg a$   & $\neg a$  &$\neg a$ &$\cdots$ &$\neg a$ &$a$ \\ 
					\cline{1-7} \cline{9-15}
					Constr     & \multicolumn{6}{c}{$ z_i =c_i = \cdots= c_{j+1} = \neg z_{j}$} &  & Constr     & \multicolumn{6}{c}{$ z_i = d_i =\cdots = d_{j+1} =  z_{j}$} \\ 
					\cline{1-7} \cline{9-15}
				\end{tabular}
				\footnotetext{Note: $a$ and $b$ iterate in $\{0,1\}$. $M_1$ and $M_2$ are valid, and $0< j+1 \leq i< n-1$. When $j+1 = i$, the middle part is omitted.} 
				\footnotetext[*]{ For the $(i+1)$-th bit of $M_1$ and $M_2$, the differences satisfy $\varDelta x_{i+1} \oplus \varDelta y_{i+1} = b \oplus a$ and $\varDelta u_{i+1} \oplus \varDelta v_{i+1} = b \oplus \neg a$.}
				\footnotetext[\dagger]{ When $j+1<i$, for any $j+1 \leq t \leq i-1$, differences of $M_1$ satisfy $\varDelta c_{t+1} = \neg a$, $\varDelta z_{t} = a$, and $(\varDelta x_{t}, \varDelta y_{t}) \in \{(\neg a, a), (a, \neg a)\}$ holding the constraint $c_{t+1} = c_{t} = x_{t}$ (or $= y_{t}$) and differences of $M_2$ satisfy $\varDelta d_{t+1} = \neg a$ and $(\varDelta u_{t}, \varDelta v_{t}) \in \{(\neg a, a), (a, \neg a)\}$ holding the constraint $d_{t+1} = d_{t}= u_{t}$ (or $= v_{t}$).}
				\footnotetext[\ddagger]{ There are totally nine cases for the differences on the $j$-th bit of $M_{1}$ and $M_2$. Since $\varDelta z_{j} = \varDelta c_{j+1} = \varDelta d_{j+1} = \neg a$, the differences $(\varDelta x_j, \varDelta y_j, \varDelta c_j)$ of $M_1$ and $(\varDelta u_j, \varDelta v_j, \varDelta d_j)$ of $M_2$ both iterate in $\{(\neg a, a, a), (a,\neg a, a), (a,a, \neg a)\}$ to hold the constraint $\neg c_{j+1} = z_j$ and $d_{j+1} = z_j$. }
			\end{minipage}
		\end{center}
	\end{table} 
	
	The second case is caused by contradictory constraints of $M_1$ and $M_2$ on two output bits. As shown in Table \ref{tab:conflicting constraints 1}, the differential constraint $z_i = \neg z_j$ of $M_1$ conflicts with the constraint $z_i = z_j$ of $M_2$ for $i\geq j+1$. Under the assumption of independence, the probability is wrongly calculated by   
	\begin{equation}\nonumber
		P_{M_2}^{i+1} = \Pr( z_i \oplus d_i = 0 \ \vert\  d_i = \cdots= d_{j+1}= z_j) = \Pr(z_i \oplus z_j = 0) = \frac{1}{2}.
	\end{equation}
	In fact, as the output of $M_1$, the input of $M_2$ satisfies $z_i = z_j$, making the differential propagation invalid, that is,
	\begin{equation}\nonumber
		P_{M_2\vert M_1}^{i+1} = \Pr( z_i \oplus z_j = 0 \ \vert \  z_i = \neg z_j) = 0. 
	\end{equation}
	Here, $P_{M_2\vert M_1}^{i+1}$ refers to the conditional probability of $M_2$ on bit $i+1$. 
	There is another basic scenario where the constraints of $M_1$ and $M_2$ are $\neg z_i =c_i = ...= c_{j+1} = \neg  z_{j}$ and $\neg z_i = d_i = ...= d_{j+1} =  z_{j}$, similar to Table \ref{tab:conflicting constraints 1} with $(\varDelta c_{i+1}, \varDelta d_{i+1}) = (\neg a, a)$, and the differential trail is also impossible,
	\begin{equation}\nonumber
		P_{M_2\vert M_1}^{i+1} = \Pr( \neg z_i \oplus z_j = 0 \ \vert \  z_i = z_j) = 0. 
	\end{equation}
	The third invalid case is caused by the uneven output cased by constraint (\ref{eq: z_i = c_j+1 = z_k}) (resp., constraint (\ref{eq: z_i = c_j+1 = c_0})). From the Proposition \ref{prop: neg z_i = c_j}, if the output satisfies: if $z_j = z_{j-1} = \cdots = z_{k}$ (resp., $z_j = z_{j-1} = \cdots = z_{0} = 1$), then $z_i =c_{i} = \neg z_{j}$, and the constraints of $M_2$ includes $z_i = z_{j} = \cdots = z_k$ at the same time, then the constraints of $M_1$ and $M_2$ are conflicting, and the differential propagation is invalid. For another scenario where $z_{i}=\neg c_{i}$ and $z_{i} = \neg d_{i}$, the situation is similar. 
	
	Above are the base cases of invalid differential propagation of a CMA and all of them are regarded as valid differences for probability calculation by previous search methods under the independence assumption. 
	The misjudgment of invalid trails during the search will waste lots of time, making the search method inefficient. 
	The second case, especially the invalid trails caused by conflicting constraints on adjacent bits ($j+1 = i$), are more likely to be encountered during the search.
	Therefore, in Sect. \ref{subsec:detect the contradictory}, we mainly introduce how to avoid this kind of impossible cases during the search. 
	
	\subsubsection{Inaccurate probability calculation caused by the non-independence}
	\label{subsubsec:probability calculation} %
	Similarly, there are three cases for the three kinds of the uneven output mentioned in Sect. \ref{subsubsec:uneven output}. 
	
	The most obvious one is that there exists same constraints of $M_1$ and $M_2$ on adjacent or non-adjacent bits, making the probability of $M_2$ higher than that calculated under the assumption of independence. Table \ref{tab:Same constraints 1} shows the basic scenario, where the constraints of $M_1$ and $M_2$ on the bit $i$ are $z_i = c_i = \cdots = \neg  z_j$ and $\neg z_i = d_i = \cdots = z_j$ ($i>j$) and the differential probability of $M_2$ is calculated by  
	\begin{equation}\nonumber
		P_{M_2\vert M_1}^{i+1} = \Pr(\neg z_i \oplus z_j = 0\ \vert \ z_i = c_i = \cdots = \neg z_j) = 1 > \frac{1}{2} = P_{M_2}^{i+1}.
	\end{equation}
	Since the output of $M_1$ satisfies $z_i = \neg z_j$, the differential equation $\neg z_i \oplus z_j = 0$ of $M_2$ holds with probability 1, two times higher than it computed under the independence assumption. There is another scenario similar to Table \ref{tab:Same constraints 1} where the differences of $M_1$ and $M_2$ on the bit $i+1$ satisfy $\varDelta c_{i+1} = \neg a$ and $\varDelta d_{i+1} = \neg a$, making the constraints be $\neg z_i = c_i = \cdots = \neg  z_j$ and $ z_i = d_i = \cdots = z_j$ ($i>j$), and the probability is also higher and computed by
	\begin{equation}\nonumber
		P_{M_2\vert M_1}^{i+1} = \Pr( z_i \oplus z_j = 0\ \vert \  \neg z_i = c_i = \cdots = \neg z_j) = 1 > \frac{1}{2} = P_{M_2}^{i+1}.
	\end{equation}
	
	\begin{table}[tp]  %
		\begin{center}
			\begin{minipage}{\textwidth}
				\caption{Same constraints of $M_1$ and $M_2$ } 
				\label{tab:Same constraints 1}
				\setlength{\tabcolsep}{3.0pt}%
				\renewcommand{\arraystretch}{1.2}%
				\renewcommand{\thefootnote}{\fnsymbol{footnote}} 
				\begin{tabular}{ccccccclccccccc}
					\cline{1-7} \cline{9-15}
					$M_1$ &$i+1$\footnotemark[1]  &$i$         &$i-1$\footnotemark[2]  &$\cdots$&$j+1$\footnotemark[2]       &$j$\footnotemark[3]        &  &$M_2$  & $i+1$\footnotemark[1]   &$i$     &$i-1$\footnotemark[2]  &$\cdots$&$j+1$\footnotemark[2]      &$j$\footnotemark[3]          \\ 
					\cline{1-7} \cline{9-15}
					$\varDelta x$   &$\cdot$  &$a$     &$\cdot$ &$\cdots$&$\cdot$ &$\neg a$       &  
					&$\varDelta z$  &{\color[HTML]{CB0000} $b$}       & {\color[HTML]{CB0000} $\neg a$} &{\color[HTML]{CB0000} $a$} & {\color[HTML]{CB0000} $\cdots$} &{\color[HTML]{CB0000} $a$} & {\color[HTML]{CB0000} $\neg a$}  \\
					$\varDelta y$   &$\cdot$    &$a$    &$\cdot$ &$\cdots$ &$\cdot$ &$a$&  
					&$\varDelta u$  &$\cdot$    &$a$    &$\cdot$ &$\cdots$ &$\cdot$ &$\neg a$ \\
					$\varDelta z$   &{\color[HTML]{CB0000} $b$} &{\color[HTML]{CB0000}$\neg a$} &{\color[HTML]{CB0000} $a$} & {\color[HTML]{CB0000} $\cdots$} &{\color[HTML]{CB0000} $a$}  &{\color[HTML]{CB0000}$\neg a$}  &  
					&$\varDelta v$  &$\cdot$&$a$        &$\cdot$    &$\cdots$ &$\cdot$ &$a$    \\
					$\varDelta c$   &$a$    &$\neg a$   &$\neg a$   &$\cdots$ &$\neg a$  &$a$ &  
					&$\varDelta d$  &$a$    &$\neg a$   &$\neg a$   &$\cdots$ &$\neg a$ &$a$ \\ 
					\cline{1-7} \cline{9-15}
					Constr     & \multicolumn{6}{c}{$ z_i =c_i = \cdots= c_{j+1} = \neg z_{j}$} &  & Constr     & \multicolumn{6}{c}{$ z_i =  \neg d_i =\cdots =  \neg d_{j+1} =  \neg z_{j}$} \\ 
					\cline{1-7} \cline{9-15}
				\end{tabular}
				\footnotetext{Note: $a$ and $b$ iterate in $\{0,1\}$. $M_1$ and $M_2$ are valid, and $0< j+1 \leq i< n-1$. When $j+1 = i$, the middle part is omitted.} 
				\footnotetext[*]{ For the $(i+1)$-th bit, the differences of $M_1$ and $M_2$ satisfy $\varDelta x_{i+1} \oplus \varDelta y_{i+1} = b \oplus a$ and $\varDelta u_{i+1} \oplus \varDelta v_{i+1} = b \oplus a$.}
				\footnotetext[\dagger]{ Same as Table \ref{tab:conflicting constraints 1}, for $j+1 \leq t \leq i-1$, differences of $M_1$ and $M_2$ hold constraints $c_{t+1} = c_{t}$ and $d_{t+1} = d_{t}$.}
				\footnotetext[\ddagger]{ Same as Table \ref{tab:conflicting constraints 1}, differences of the $j$-th bit hold constraints $\neg c_{j+1} = z_j$ and $d_{j+1} = z_j$. }
			\end{minipage}
		\end{center}
	\end{table}
	
	For the uneven output of the first addition caused by $z_i = c_0 = 0$ (resp., $z_i = \neg c_0 = 1$), if the differential equation of $M_2$ on the $(i+1)$-th bit is $z_i = d_0 = 0$ (resp., $ z_i =\neg d_0 = 1$), then the probability is also higher: $P_{M_2\vert M_1}^{i+1} = 1 > 1/2$.
	
	The uneven output of $M_1$ caused by $z_i=c_{j+1} = f_{carry}(x[j,k+1],y[j,k+1], c_{k+1}) $ (resp., $z_i = \neg c_{j+1}$) for $0 \leq k+1 \leq j \leq i-1 \leq n-2$, as mentioned in Sect. \ref{subsubsec:uneven output}, could also affected the differential probability of $M_2$. 
	In this scenario, as the output of $M_1$, $z_i$ is associated to $(z_j,\cdots,z_{k+1})$ through the carry bit $c_j$, and 
	if the constraints of $M_2$ restrict $z_i$ with $d_{j} = f_{carry}(z[j,k+1],u[j,k+1], d_{k+1})$, the probability $P^{i+1}_{M_2 \vert M_1}$ is very complex to calculate.
	%
	\iffalse{
		In this scenario, the differential constraints of $M_1$ and $M_2$ associate the same bit of the intermediate state $z$ to their carry bits, that is
		\begin{equation}
			\begin{aligned}
				M_1:\ & z_i= c_{j+1} {\ \rm or\ } \neg z_i = c_{j+1},\\
				M_2:\ & z_i= d_{j+1} {\ \rm or\ } \neg z_i = d_{j+1}.  \\
			\end{aligned}
		\end{equation}
		As the output of $M_1$, $z_i$ is associated to $(z_j,\cdots,z_{k+1})$ through the carry bit $c_j$, and the probability of the equation of $M_2$ consists of $z_i$ and $d_{j+1}$ is very complex, since $d_{j+1} = f_{carry}(z[j,k+1],u[j,k+1], d_{k+1})$.}
	\fi
	
	Since there is no good mathematical theory to describe the calculation, so we can only get the probability from experiments.
	In this paper, a \#SAT method introduced in Sect. \ref{subsubsec:calculate} is applied to compute the probability of a CMA. 
	We find two properties from experiments that if the relations of $z_i$ and $c_{j+1}$ are different to that of $z_i$ and $d_{j+1}$, the true probability of $M_2$ is higher than it calculated under the independence assumption, and if the relations are same, the true probability is lower.
	These two properties are shown in Observation \ref{pro: the prob is higher}	and Observation \ref{pro: the prob is lower}, and there are two basic examples with $i = j+1$ and $k+1 = 0$.
	\begin{observation}
		\label{pro: the prob is higher}
		For $i\geq 2$, if the differential constraints of the CMA are: 
		\begin{equation} \nonumber
			\begin{aligned}
				&M_1 : \neg z_i\ or\ z_i = c_{j+1} = f_{carry} \left( x[j,k+1] , y[j,k+1], c_{k+1}\right),\\
				&M_2 : z_i\ or\ \neg z_i = d_{j+1} = f_{carry} \left( z[j,k+1] , u[j,k+1], d_{k+1}\right),
			\end{aligned}
		\end{equation}
		the probability of $M_2$ is higher: $\Pr(M_2\vert M_1) > \Pr(M_2)$.
	\end{observation}
	
	\begin{observation}
		\label{pro: the prob is lower}
		For $i\geq 2$, if the differential constraints of the CMA are:
		\begin{equation} \nonumber
			\begin{aligned}
				&M_1 :z_i\ or\ \neg  z_i = c_{j+1} = f_{carry} \left( x[j,k+1] , y[j,k+1], c_{k+1}\right),\\
				&M_2 : z_i\ or\ \neg z_i = d_{j+1} = f_{carry} \left( z[j,k+1] , u[j,k+1], d_{k+1}\right),
			\end{aligned}
		\end{equation}
		the probability of $M_2$ is lower: $\Pr(M_2\vert M_1) < \Pr(M_2)$.
	\end{observation}
	
	\begin{example}
		\label{example: p_sat > p_indep}
		For Proposition \ref{pro: the prob is higher}, let the differences of the CMA be 
		\begin{equation}\nonumber
			M_1
			\left\{ \begin{matrix}
				\varDelta x = \rm 0b001000 \\
				\varDelta y = \rm 0b001000 \\
				\varDelta z = \rm 0b000000 \\
			\end{matrix}
			\right. ,
			\quad M_2
			\left\{ \begin{matrix}
				\varDelta z = \rm 0b000000 \\
				\varDelta u = \rm 0b001000 \\
				\varDelta v = \rm 0b001000 \\
			\end{matrix}
			\right. .
		\end{equation}
		The constraint of $M_1$ is $\neg z_3 = c_3$, and the differential equation of $M_2$ on bit 3 is $z_3 = d_3$. With the help of the \#SAT solver GANAK, we get the true probability of $M_2$: 
		\begin{equation}\nonumber
			\Pr(M_2\vert M_1) = \frac{\Pr(M_2,M_1)}{\Pr(M_1)} = \frac{21}{32} > \frac{1}{2} = \Pr(M_2) .
		\end{equation}
	\end{example}
	
	\begin{example}
		\label{example: p_sat < p_indep}
		For Proposition \ref{pro: the prob is lower}, let the differences of the CMA be
		\begin{equation}\nonumber
			M_1
			\left\{ \begin{matrix}
				\varDelta x = \rm 0b001000 \\
				\varDelta y = \rm 0b011000 \\
				\varDelta z = \rm 0b000000 \\
			\end{matrix}
			\right.  ,
			\quad M_2
			\left\{ \begin{matrix}
				\varDelta z = \rm 0b000000 \\
				\varDelta u = \rm 0b001000 \\
				\varDelta v = \rm 0b001000 \\
			\end{matrix}
			\right. .
		\end{equation}
		The constraint of $M_1$ is $ z_2 = c_2$, and the constraint of $M_2$ on the input is $ z_2 = d_2$. With the help of the \#SAT solver GANAK, we get the true probability of $M_2$: 
		\begin{equation}\nonumber
			\Pr(M_2\vert M_1) = \frac{\Pr(M_2,M_1)}{\Pr(M_1)} = \frac{11}{32} < \frac{1}{2} = \Pr(M_2) .
		\end{equation}
	\end{example}
	
	During the search of differential trails, there may exist more complicated scenarios consisting of the above three cases in a valid trail, which is difficult to analysis. Without an effective theory to capture the differential propagation of consecutive modular additions, the \#SAT method is an effective way for the accurate calculation of differential probabilities.  

	\begin{proposition}{\textbf{(Detect non-independent CMAs)}}
		\label{prop:detect non-independent}
		For the CMA with $M_1:(\varDelta x ,\varDelta y)\xrightarrow{\boxplus}\varDelta z\ {\rm and}\ M_2:(\varDelta z ,\varDelta u)\xrightarrow{\boxplus}\varDelta v$,
		let their carry vectors be $c,d \in \mathbb{F}_2^{n}$ respectively.
		If $M_1$ and $M_2$ are non-independent, then the differential constraints of two additions must associate the same bits of the intermediate state $z$ to their carry bits respectively, that is, 
		\begin{equation}
			\label{eq:non-independent}
			\exists\ 0\leq i\leq n-2,\ M_1:\ z_i = c_i \ or \ z_i = \neg c_i,\ and \ M_2:\ z_i = d_i \ or \ z_i = \neg d_i.
		\end{equation}
	\end{proposition}
	
	This proposition is concluded from above non-independent cases that lead to impossible trails and inaccurate calculation of probability and Lemma \ref{lemma: uneven output}.
	According to Table \ref{tab:diff_constraint}, it is obviously that differences satisfying condition (\ref{eq:non-independent}) are $\left({\varDelta x}_i,{\varDelta y}_i,{\varDelta z}_i,{\varDelta u}_i,{\varDelta v}_i\right)=\left(1,1,0,1,1\right)\ {\rm or}\  \left(0,0,1,0,0\right). $ 
	Although this property is not the sufficient condition of non-independent CMAs, it can be used to extract the CMAs which could be possibly non-independent from the searched trail.
	And this will save lots of time for calculating the accurate probability of valid trails and finding the differences that cause the trail invalid.
	
	\section{Verification and probability calculation of differential on CMA}
	\label{sec:SAT model}

	\subsection{The SAT model capturing state transition }
	\label{subsec:state transition model }
	In this section, SAT models for the state (value) transition of the modular addition, rotation and XOR operation are introduced.

	For a modular addition $x\boxplus y=z$, let its carry state be $c \in \mathbb{F}_2^n$ and differential propagation be $M:(\varDelta x, \varDelta y) \xrightarrow{\boxplus} \varDelta z$. The variables of this model are the inputs, output and carry states: $(x,y,z,c)\in \mathbb{F}_2^{4n}$, and the CNF of the carry function, output function and differential constraints are listed below.
	
	\textbf{Carry function.}  
	Since the LSB of the carry state is zero, there are $c_0 = 0 \Leftrightarrow \neg c_0 = 1$ and the carry function $c_1 = x_0y_0$, same as Eq. (\ref{eq: cnf of g = ab}). 
	For $ i = 0,1,...,n-2$ bit, the CNF of carry function is
	\begin{equation}
		\label{eq:carry fuction to cnf}
		\begin{aligned}
			c_{i+1}=x_{i}y_{i} \oplus (x_{i} \oplus y_{i}) c_{i} \Longleftrightarrow 
			\ & (c_{i+1} \vee \neg c_{i} \vee \neg y_{i})  
			\wedge (c_{i+1} \vee c_{i} \vee \neg x_{i} \vee \neg y_{i}) \\
			& \wedge (\neg c_{i+1} \vee c_{i} \vee x_{i})
			\wedge (c_{i+1} \vee \neg c_{i} \vee x_{i} \vee \neg y_{i}) \\
			& \wedge (\neg c_{i+1} \vee c_{i} \vee y_{i}) 
			\wedge (c_{i+1} \vee \neg c_{i} \vee \neg x_{i} \vee y_{i}) \\
			& \wedge (\neg c_{i+1} \vee x_{i} \vee y_{i}).	
		\end{aligned}
	\end{equation}
	
	\textbf{Output function.} 
	For the LSB, the output is $z_0 = x_0 \oplus y_0$, same as Eq. (\ref{eq: cnf of g = a oplus b}) because of $c_0 = 0$.
	For $ i = 1,...,n-1$ bit, the output is 
	\begin{equation}
		\label{eq: output formulas to cnf} 
		\begin{aligned}
			z_{i}=x_{i} \oplus y_{i} \oplus c_{i} \Longleftrightarrow
			\ & (\neg z_{i} \vee x_{i} \vee y_{i} \vee c_{i})
			\wedge (z_{i} \vee \neg x_{i} \vee \neg y_{i} \vee \neg c_{i}) \\ 
			& \wedge (z_{i} \vee \neg x_{i} \vee y_{i} \vee c_{i})
			\wedge (\neg z_{i} \vee x_{i} \vee \neg y_{i} \vee \neg c_{i}) \\ & \wedge (z_{i} \vee x_{i} \vee \neg y_{i} \vee c_{i})
			\wedge (\neg z_{i} \vee \neg x_{i} \vee y_{i} \vee \neg c_{i}) \\ & \wedge (z_{i} \vee x_{i} \vee y_{i} \vee \neg c_{i}) 
			\wedge (\neg z_{i} \vee \neg x_{i} \vee \neg y_{i} \vee c_{i}).
		\end{aligned}
	\end{equation}
	
	\textbf{The differential constraints.} 
	Given the input and output differences of a modular addition, we use Table \ref{tab:diff_constraint} to determine which bits have constraints of the differential propagation, and convert these constraints into their CNF to replace the original carry formulas. For example, if the $(i+1)$-th and $i$-th bits of difference are $(\varDelta x_{i+1},\varDelta y_{i+1},\varDelta z_{i+1}) = (1,0,0)$ and $(\varDelta x_i,\varDelta y_{i},\varDelta z_{i}) = (1,1,0)$, the CNFs of the differential constraints on the $i$-th bit are listed below.
	\begin{enumerate}
		\item The constraint on the inputs (the differential equation): 
		\begin{equation}
			\neg y_i = c_i  \Longleftrightarrow (\neg y_i \vee \neg c_i) \wedge (y_i \vee c_i).
		\end{equation}
		\item The constraint on the output: 
		\begin{equation}
			\neg z_i = x_i  \Longleftrightarrow (\neg z_i \vee \neg x_i) \wedge (z_i \vee x_i).
		\end{equation}
		\item The constraint on the carry: 
		\begin{equation}
			c_{i+1} = \neg z_i  \Longleftrightarrow (\neg c_{i+1} \vee \neg z_i) \wedge (c_{i+1} \vee z_i).
		\end{equation}
	\end{enumerate}
	Replace the original carry formula with the above three formulas, and do similar operations on other bits. 
	The equation (\ref{eq:carry fuction to cnf}) and (\ref{eq: output formulas to cnf}) capture the state propagation of the addition, and other equations limit the states to meet the differential propagation, and formulas of differential constraints are simpler than that of the carry functions, which make the SAT model easier to solve.

	For the XOR operation of variables: $z = x \oplus y$, which has no differential constraints, the Eq. (\ref{eq: cnf of g = a oplus b}) can be used to capture the propagation on each bit;
	
	For the state propagation of linear functions, the output state can be represented by input variables with no differential constrains. For the XOR operation between a variable and the round constant: $y = x \oplus r$, where $r \in \mathbb{F}_2^n$, set a variable $x = (x_{n-1},\cdots, x_0)$ as the input, and the output state $y$ satisfies
	\[y_i = \left\{
	\begin{matrix}
		x_i &if\ r_i = 0, \\
		\neg x_i &if\ r_i = 1,
	\end{matrix}
	\right. \]
	for $i = 0,1,\cdots,n-1$. 

	For the rotation: $y = x \lll l$, where $0\leq l \leq n-1$ is the rotation constant, set a variable $x = (x_{n-1},\cdots, x_0)$ as the input. Then the output state $y$ is
	$y = (x_{n-1-l},\cdots , x_{0}, x_{n-1},\cdots, x_{n-l})$.

	\subsection{The SAT method to identify incompatible differential trails}
	\label{subsubsec:verify diff-trails} %
	To verify the validity of a differential trail, MILP models are built in \cite{DBLP:journals/dcc/SadeghiRB21} and \cite{C:LiuIsoMei20} to find a right pair of states. The MILP model captures the propagation of a pair of states $(x,x')$ on the primitives and uses the XOR operation of states: $\varDelta x = x\oplus x'$ to add differential constraints.
	
	Our SAT method only considers the propagation of one state instead of a pair, which has less variables than the MILP model.
	Since the ARX-based cipher is consisted of the modular addition, the rotation, and the XOR operation, the state propagation of a given differential trail can be captured by a big SAT model combining the SAT models of each operation mention in Sect. \ref{subsec:state transition model }. 
	With the help of a SAT solver, we can quickly verify whether the differential trail includes a right state following it. If there exits states that satisfies the differential propagation, we can get \textit{SAT} and one of the states from the solver. If not, we get \textit{UNSAT} from the solver.  
	
	\subsection{The \#SAT method to calculate the differential probability of CMAs}
	\label{subsubsec:calculate}  %
	Given a CMA with differential $M_1: (\varDelta x ,\varDelta y)\xrightarrow{\boxplus}\varDelta z$ and $M_2: (\varDelta z ,\varDelta u)\xrightarrow{\boxplus}\varDelta v$, the method to calculate accurately the differential probability is introduced in this section. 
	
	Assuming that the input $(x,y,u)$ of the CMA are independent and uniform random, a SAT model is built to capture its state propagation following the given differential. Set the variables be the input, output and carry states of two additions: $ (x,y,u)\in \mathbb{F}_2^{3n}$, $(z,v)\in \mathbb{F}_2^{2n}$, and $(c,d)\in \mathbb{F}_2^{2n}$, where the variable $z$ is both the output of $M_1$ and the input of $M_2$. Build two models mentioned in Sect. \ref{subsec:state transition model } to capture the state propagation of $M_1$ and $M_2$ respectively, and combine them to get the state propagation model for the CMA.
	Since the format of the SAT and \#SAT model are the same, 
	using SAT solver to solve the model, the validity of the differential trail can be quickly verified; 
	Using \#SAT solver to solve it, the number of solutions can be obtained, which is denoted by $\#\{(x,y,u)\}$ since the values of states $(z,v)$ and carry states $(c,d)$ of two additions are determined for a given value of $(x,y,u)$. The accurate value of the differential probability is calculated by 
	\begin{equation}
		\Pr(M_1,M_2) = \frac{\#\{(x,y,u)\}}{2^{3n}}.
	\end{equation}

	If there are some linear functions between two additions of the CMA, such as the rotation and the XOR operation with round constants in the key schedule of SPECK, the state transition models of linear functions mentioned in Sect. \ref{subsec:state transition model } are added into the SAT model.

	The \#SAT solver we use is GANAK \cite{DBLP:conf/ijcai/SharmaRSM19}. For small size CMA, such as the word size $ n =$ 16 and 24, whose number of solutions is $2^{3n}\times \Pr(M_1,M_2)$, the solver output the result in a few seconds or minutes. As for the word size $n \geq 32$ and a high differential probability, which cause a large number of solutions, it takes hours for solvers to output results because the higher difference probability have the fewer differential constraints making the \#SAT model more complicated and difficult to be solved.
	Therefore, in order to speed up the solving,
	our method is to remove the bits which are not related to the differential constraints from the model.
	
	For example, given the differences of $M_1$ and $M_2$ as Example (\ref{example: p_sat < p_indep}), 
	the 4th and 5th bit of those two additions can be omitted when build the state propagation model, because there is neither differential constraints on them nor higher bits associated with them,
	while the lowest two bits can not be omitted, because the constraint of $M_1$ on the 3th bit, which is $z_2 = c_2 = f_{carry}(x[1,0],y[1,0],c_0)$, associates the lowest three bits of the output. 
	With this method, solving large models will be much faster. 
	
	To obtain a more accurate probability of a given differential trail of an ARX-based cipher, we firstly apply Proposition \ref{prop:detect non-independent} to identify the non-independent CMAs in the trail, and then build \#SAT models for these CMAs to calculate their probabilities assuming that their inputs are independent and uniform random.
	Although sometimes the inputs of additions in the trail may not all be uniformly random, this assumption is more accurate than assuming that the differential propagation of all modular additions are independent of each other.
	In order to prove the effectiveness and accuracy of our method, some experiments are implemented in Appendix \ref{sec:Experiments to prove the effectiveness of our method}. We build some toy ciphers with small block size, and find that the probabilities of the differential trails calculated by our method are more accurate and closer to the real value than that calculated under the independence assumption.

	\subsection{Detect the contradictory of differential constraints of a CMA}	
	\label{subsec:detect the contradictory}
	During the search of differential trails for ARX ciphers, invalid trails caused by contradictory constraints on adjacent bits of CMA are often misjudged under the independence assumption and wastes lots of time.
	In this section, we introduce a method to 
	represent adjacent constraints by several basic operations between input and output differences. 
	With this method, we can automatically detect the location of contradictory constraints for a given invalid differential trails of CMA, 
	and recording the exact differences of bits that cause conflicting constraints will help to avoid this kind of invalid trails in subsequent re-searches and save a lot of time. 
	This method can also be used to build SMT models to exclude invalid trails, but same as the methods mentioned above, it is not efficient when searching long trails. 
	
	In the following lemmas, we first introduce the representations of the differential constraints on the output according to Table \ref{tab:diff_constraint}.
	\begin{lemma} \label{lemma:relations between z and c} %
		\normalsize	Given a differential trail of the modular addition $M_1:(\varDelta x ,\varDelta y)\xrightarrow{\boxplus}\varDelta z$ and let $\varDelta c = \varDelta x \oplus \varDelta y \oplus \varDelta z $. Consider the following n-bit values 
		\begin{equation}
			\begin{aligned}
				&a_1 = \neg(\varDelta x \oplus \varDelta y ) \wedge \neg(\varDelta y \oplus \neg\varDelta z) \wedge \neg(\neg\varDelta z \oplus (\varDelta c \gg 1) ), \\
				&a_2 = \neg(\varDelta x \oplus \neg\varDelta y ) \wedge \neg(\varDelta z \oplus (\varDelta c \gg 1) ),\\
				&a_3 = \neg(\varDelta x \oplus \varDelta  y ) \wedge \neg(\varDelta y \oplus \neg\varDelta z) \wedge \neg(\varDelta z \oplus (\varDelta c \gg 1) ). 
			\end{aligned}
		\end{equation}
		Then, for $0\leq i \leq n-2$, we have\footnote[1]{ The '$\Leftrightarrow$' used in the lemmas of this section means that the left equation holds if and only if the differential propagation has the constraint on the right. }
		\begin{equation}
			\begin{aligned}
				&a_{1,i} = 1 \Leftrightarrow z_i=c_i, \\
				&a_{2,i} = 1 \Leftrightarrow \neg z_i=c_{i+1},\\	
				&a_{3,i} = 1  \Leftrightarrow \neg z_i=c_{i+1}=c_i,
			\end{aligned} 
		\end{equation}
		where $a_{j,i}$ denote the $i$-th bit of $a_j$, for $j \in \{1,2,3\}$.
	\end{lemma}
	\begin{proof}
		Note that $a_{1,i} = 1 $ if and only if $\varDelta x_i= \varDelta y_i = \neg \varDelta z_i = \varDelta c_{i+1}$, which is the only case that leads to the constraint $z_i = c_i$ according to Table \ref{tab:diff_constraint} (At other conditions, $z_i$ is independent of $c_i$). The proof of $a_2$ and $a_3$ are the same: $a_{2,i} = 1 $ if and only if $\varDelta z_i= \varDelta c_{i+1} = \neg \varDelta c_i$, which will lead to the constraint $c_{i+1} = \neg z_i$. 
	\end{proof}
	With the vectors of Lemma \ref{lemma:relations between z and c}, we can represent the constraints on the adjacent bits of output.
	\begin{lemma} \label{lemma:constraints on adjacent bits of z}
		\normalsize Consider the following n-bits vector,
		\begin{equation}
			\begin{aligned}
				&b_N = (a_1 \gg 1) \wedge (a_2 \vee a_3),\\	
				&b_E = (a_3 \gg 1) \wedge (a_2 \vee a_3).
			\end{aligned}
		\end{equation}
		\normalsize Then, for $0\leq i \leq n-3$, let $b_{N,i}$ and $b_{E,i}$ denote the $i$-th bit of $b_N$ and $b_E$, we have 
		\begin{equation}
			\begin{aligned}
				&b_{N,i} = 1 \Leftrightarrow z_{i+1} = \neg z_i,\\	
				&b_{E,i} = 1 \Leftrightarrow z_{i+1} = z_i.
			\end{aligned}
		\end{equation}
	\end{lemma}
	\begin{proof}
		\normalsize Note that $( a_{2,i} \vee a_{3,i} ) = 1 \Leftrightarrow \neg z_i = c_{i+1} $. So $b_{N,i} = a_{1,i+1} \wedge ( a_{2,i}\vee a_{3,i} ) = 1 $ if and only if $z_{i+1} = c_{i+1} = \neg z_i$, and $b_{E,i} = a_{3,i+1} \wedge ( a_{2,i}\vee a_{3,i} ) = 1 $ if and only if $\neg z_{i+1} = c_{i+1} = \neg z_i$.  
	\end{proof}
	Similarly, the differential constraints on adjacent bits of the input are represented in the following Lemma.
	\begin{lemma} \label{lemma:constraints between input z and carry d}
		\normalsize	Given a differential trail $M_2: (\varDelta z ,\varDelta u)\xrightarrow{\boxplus}\varDelta v\ $and $\varDelta d = \varDelta z \oplus \varDelta u \oplus \varDelta v $, the constraint between the input bit $z_i$ and carry bit $d_i$ can be represent by the following bit-vectors:
		\begin{equation}
			\begin{aligned}
				&a_1' = \neg(\varDelta z \oplus \neg \varDelta u ) \wedge \neg(\varDelta u \oplus \varDelta v) \wedge \neg(\varDelta v \oplus (\varDelta d \gg 1) ),\\
				&a_2' = \neg(\varDelta u \oplus \neg\varDelta v ) \wedge \neg(\varDelta z \oplus (\varDelta d \gg 1) ),\\
				&a_3' = \neg(\varDelta z \oplus \neg \varDelta u ) \wedge \neg(\varDelta u \oplus \varDelta v) \wedge \neg(\varDelta z \oplus (\varDelta d \gg 1) ),\\
				&b_N' = (a_1'\gg 1) \wedge (a_2' \vee a_3'),\\
				&b_E' = (a_3'\gg 1) \wedge (a_2' \vee a_3'). 
			\end{aligned}	
		\end{equation}
		\normalsize Then, the constraints on adjacent bits of the input state $z$ can be represented by $b_N'$ and $b_E'$. For $0\leq i \leq n-3$, we have 
		\begin{equation}
			\begin{aligned}
				&b_{N',i} = 1 \Leftrightarrow \neg z_{i+1}=z_i,\\
				&b_{E',i}= 1 \Leftrightarrow z_{i+1}=z_i.
			\end{aligned}	
		\end{equation}
	\end{lemma}
	\begin{theorem}
		\label{lemma: conflicting constraints on adjacent bits } 
		For	$M_1:(\varDelta x ,\varDelta y)\xrightarrow{\boxplus}\varDelta z\ $ and $M_2:(\varDelta z,\varDelta u)\xrightarrow{\boxplus}\varDelta v\ $. Get $(b_N, b_E)$ from Lemma \ref{lemma:constraints on adjacent bits of z} and $(b'_N, b'_E)$ from Lemma \ref{lemma:constraints between input z and carry d}.
		Let	$q=\left(b_N \wedge {b'}_E\right) \vee \left(b_E \wedge {b'}_N\right)$, we have:
		\begin{equation}
			\begin{aligned}
				q_i = 1\ \Longleftrightarrow\ \ &(b_N \wedge b'_E)_i = 1 \Leftrightarrow M_1\ has \ z_{i+1} = \neg z_i \ {\rm and}\ M_2\ has \ z_{i+1} = z_i, \\
				&or\ (b_E \wedge b'_N)_i = 1 \Leftrightarrow M_1\ has \ z_{i+1} =  z_i \ {\rm and}\ M_2\ has \  z_{i+1} = \neg z_i .
			\end{aligned}
		\end{equation}
		Here, $(b_N \wedge b'_E)_i$ refers to the $i$-th bit of $(b_N \wedge b'_E)$ for$\ 0\leq i \leq n-2$. Therefore, $q \neq 0$ if and only if there exists conflicting constraints of $M_1$ and $M_2$ on adjacent bits.
	\end{theorem}
	Given an invalid differential trails of consecutive modular additions, Lemma \ref{lemma: conflicting constraints on adjacent bits } can be used to determine whether there are contradictory constraints on adjacent bits. If so, the error differences and their positions can be obtained by analyzing the position of "1" in $q$, and recording them can avoid this type of impossible differential trails in the following re-search. For example, there are conflict constraints $z_{i+1} = \neg z_{i}$ of $M_1$ and $z_{i+1} = z_{i}$ of $M_2$, then we record the differences of the $i$-th, $(i+1)$-th, and $(i+2)$-th bit.
	\section{Application on SPECK family of block ciphers}
	\label{sec:application on SPECK}
	
	\subsection{Searching the related-key differential trails of SPECK family of block ciphers}
	The key schedule of SPECK includes consecutive modular additions, which would affect the search for differentials trails as mentioned in Chapter \ref{sec:diff property of mdad and consecutive mdad}. 
	In this chapter, we build SAT models to search the related-key differential (RKD) trails under the independence assumption and verify the invalidity of them. If the trail we find is valid, we find the non-independent consecutive modular additions and calculate the accurate value of their probability. If not, we find and record the reasons that cause the trails invalid to avoid them during the next search.
	
	The only non-linear operation in the SPECK round function is the modular addition, 
	and the only key-dependent operation is the XOR operation with sub-key after the modular addition, which means the the cipher operation is completely predictable until this first XOR operation with sub-key. Same as \cite{DBLP:journals/dcc/SadeghiRB21}, we ignore the modular addition in the first round when search for related-key differential trails. 
	Therefore, the input of the first round is $x^{0'} = (x^0 \ggg \alpha) \boxplus y^0$ and $y^{0'} = y^0 \lll \beta$, 
	and the probability of a $R$ round trail is $P_{DK} = P_D \times P_K = \prod_{i=1}^{R-1}P_D^i \times \prod_{i=0}^{R-2}P_K^i$, where $P_D$ and $P_K$ are differential probabilities of the data encryption and the key schedule, and $P_D^i$ and $P_K^i$ are probabilities of each round of them.
	
	\subsubsection{SAT model for the search of differential trail }
	\label{subsubsec:SAT for search}
	In order to model the differential propagation of the ARX-based cipher, it is sufficient to express XOR, bit rotation and modular addition. Since bit rotation and XOR operation are linear functions, the differential propagation on them can be modeled similar to the state propagation in Sect. \ref{subsec:state transition model }.
	
	As for modular addition $(\varDelta x, \varDelta y) \xrightarrow{\boxplus} \varDelta z$, the verification and calculation formula of Theorem \ref{thm:verification of modular addition} and Theorem \ref{thm:probability calculation of modular addition} are converted to CNF by SunLing et al.\cite{DBLP:journals/tosc/SunWW21}. But due to the definition of the LSB in this paper is different from theirs, our conversion result is reproduced below.
	
	\textbf{Verification.} 
	For the LSB, the CNF of the verification formula is
	\begin{equation}
		\begin{aligned}
			&(\neg \varDelta x_0 \vee \varDelta y_0 \vee \varDelta z_0) 
			\wedge (\varDelta x_0 \vee \neg \varDelta y_0 \vee \varDelta z_0)  
			\wedge (\varDelta x_0 \vee \varDelta y_0 \vee \neg \varDelta z_0) \\
			&\wedge (\neg \varDelta x_0 \vee \neg \varDelta y_0 \vee \neg \varDelta z_0).
		\end{aligned}
	\end{equation}
	For the $0 \leq i \leq n-2 $ bit, the CNF of the verification formula is 
	\begin{equation}
		\begin{aligned}
			&(\varDelta x_i \vee \varDelta y_i \vee \varDelta z_i \vee \varDelta x_{i+1} \vee \varDelta y_{i+1} \vee \neg \varDelta z_{i+1})  \\
			&\wedge (\varDelta x_i \vee \varDelta y_i \vee \varDelta z_i \vee \varDelta x_{i+1} \vee \neg \varDelta y_{i+1} \vee \varDelta z_{i+1}) \\
			&\wedge (\varDelta x_i \vee \varDelta y_i \vee \varDelta z_i \vee \neg \varDelta x_{i+1} \vee \varDelta y_{i+1} \vee \varDelta z_{i+1}) \\	
			&\wedge (\varDelta x_i \vee \varDelta y_i \vee \varDelta z_i \vee \neg \varDelta x_{i+1} \vee \neg \varDelta y_{i+1} \vee \neg \varDelta z_{i+1}) \\
			&\wedge (\neg \varDelta x_i \vee \neg \varDelta y_i \vee \neg \varDelta z_i \vee \varDelta x_{i+1} \vee \varDelta y_{i+1} \vee \varDelta z_{i+1}) 	\\
			&\wedge (\neg \varDelta x_i \vee \neg \varDelta y_i \vee \neg \varDelta z_i \vee \varDelta x_{i+1} \vee \neg \varDelta y_{i+1} \vee \neg \varDelta z_{i+1}) \\
			&\wedge (\neg \varDelta x_i \vee \neg \varDelta y_i \vee \neg \varDelta z_i \vee \neg \varDelta x_{i+1} \vee \varDelta y_{i+1} \vee \neg \varDelta z_{i+1}) \\	
			&\wedge (\neg \varDelta x_i \vee \neg \varDelta y_i \vee \neg \varDelta z_i \vee \neg \varDelta x_{i+1} \vee \neg \varDelta y_{i+1} \vee \varDelta z_{i+1}) .\\
		\end{aligned}
	\end{equation}
	
	\textbf{Probability calculation.} 
	Let $w \in \mathbb{F}_2^{n-1}$, so that $w_i = \neg eq\left(\varDelta x_{i}, \varDelta y_{i}, \varDelta z_{i}\right)$ for $0 \leq i \leq n-2 $ and its CNF is
	\begin{equation}
		\label{eq:cnf of prob calculation}
		\begin{aligned}
			&\neg (\varDelta x_i \vee \varDelta z_i \vee w_i) 
			\wedge (\varDelta x_i \vee \neg \varDelta z_i \vee w_i)
			\wedge (\varDelta x_i \vee \varDelta y_i \vee\varDelta z_i \vee \neg w_i) \\
			&\wedge (\neg \varDelta x_i \vee \varDelta y_i \vee w_i) 
			\wedge (\varDelta x_i \vee \neg \varDelta y_i \vee w_i) 
			\wedge (\neg \varDelta x_i \vee \neg \varDelta y_i \vee \neg \varDelta z_i \vee \neg w_i).\\
		\end{aligned}
	\end{equation}
	Then the probability is $\Pr(M) = 2^{-\sum_{i=0}^{n-2}w_i }$. 
	
	Since the modular additions are included in both data encryption and key schedule, we denote the differential probabilities on additions of these two parts as $w_d = -\log_2{P_D}$ and $w_k= -\log_2{P_K}$, and that of each round as $w_d^i = -\log_2{P_D^i}$ and $w_k^i= -\log_2{P_K^i}$. 
	The calculation of these probabilities under the independence assumption is the same as the Formula (\ref{eq:cnf of prob calculation}).
	Given the target probability $W = -\log_2{P_{DK}}$ of the search model, we convert the probability constraints $w_d + w_k \leq W$ into CNF formulas with the sequential encoding method \cite{DBLP:conf/cp/Sinz05}.
	
	\subsubsection{The verification steps of differential trails}
	\label{subsubsec:verifitcation of trails}
	Since our search is under the assumption of independence, the validity of the searched differential trails needs to be verified. 
	The entire verification process is shown in Fig. \ref{fig:verifiaction process}. Given a differential trail, we first build the SAT model in Sect. \ref{subsubsec:verifitcation of trails} to capture state propagation of the entire trail and then apply SAT solver CaDiCal \cite{BiereFazekasFleuryHeisinger-SAT-Competition-2020-solvers} to verify whether the model contains any solutions, that is, the validity of the trail. 
	
	\begin{figure}[tp]  
		\centering
		\includegraphics[ width=0.8\textwidth]{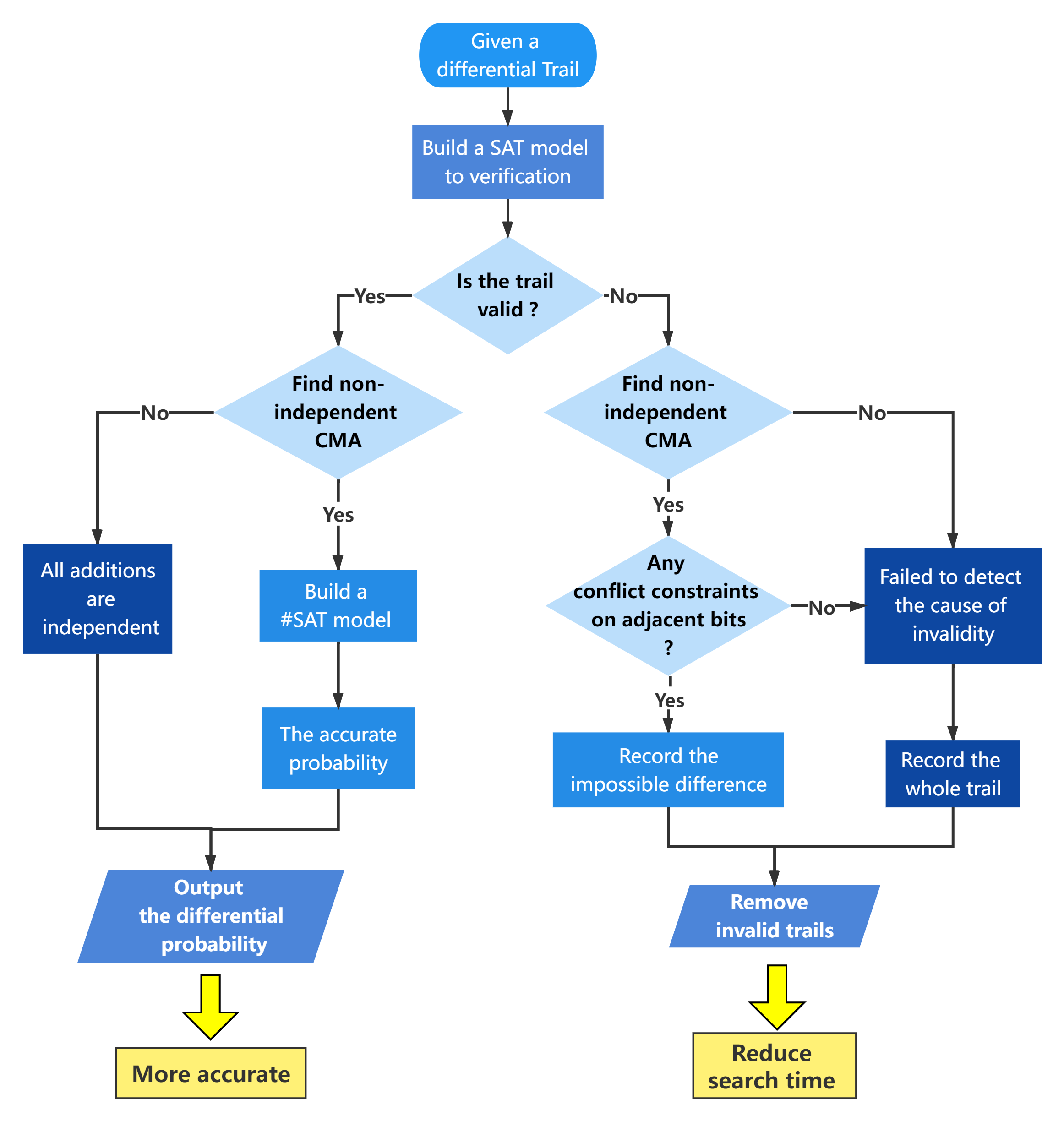}
		\caption{The verification steps of differential trails}
		\label{fig:verification steps}
		\label{fig:verifiaction process}
	\end{figure}
	
	If the trail is verified as valid, 
	we use Proposition \ref{prop:detect non-independent} to find the non-independent CMAs. 
	If there exits, we apply the \#SAT method to calculate their accurate probability as showed in Sect. \ref{subsubsec:calculate} and update the probability of the whole differential trail, 
	if not, we treat each addition as independent. 
	Our method is effective and more accurate than the probability calculated under the independence assumption, which can be proved by some experimental results in the Appendix \ref{sec:Experiments to prove the effectiveness of our method}.
	This step will help to save resources when constructing the differential distinguisher. 
	For example, a valid RKD trails of 15-round SPECK48/96 with a probability of $2^{-87}$ was found under the independence assumption. 
	Our method find that there exists a non-independent CMA consisting of additions in round 10 and 13. Their input and output differences are shown in Table \ref{tab:CMA of 15round SPECK48/96} and the accurate probability of $M_{13}$ is $\Pr(M_{13}\vert M_{10}) = 2^{-5.5081} > 2^{-9} = \Pr(M_{13}) $, more than 8 times higher than that calculated under the independence assumption. 
	Therefore, the probability of this trail should be $2^{-83.5081}$, much higher and more accurate than that computed under the independence assumption. 
	
	\begin{table}[tp]
		\begin{center}
			\begin{minipage}{\textwidth}
				\caption{The non-independent CMA in the key schedule of the 15-round trail of SPECK48/96}
				\label{tab:CMA of 15round SPECK48/96}
				\renewcommand{\thefootnote}{\fnsymbol{footnote}}
				\begin{tabular}{ccc}
					\toprule
					$M_{10}$ & Round 10      & The constraints on the output                         \\
					\midrule
					$\varDelta x$  & 0b00000000{\textcolor{red}{1}}000{\textcolor{red}{0000}}00000000 & $z_{11} = \neg z_{10} = \neg z_{9} = \neg z_8$  \\ %
					$\varDelta y$  & 0b10000000{\textcolor{red}{1}}000{\textcolor{red}{0001}}11100100 & \multirow{2}{*}{$\neg z_{15} = c_{15}= f_{carry}(x[15,12],y[15,12],x[11])$}      \\ 
					$\varDelta z$  & 0b10000000{\textcolor{red}{0}}000{\textcolor{red}{1111}}00100100 &    \\ 				
					\toprule
					$M_{13}$& Round 13      & The constraints on the input         \\ 
					\midrule
					$\varDelta zz$\footnotemark[1] & 0b0010010010000000{\textcolor{red}{0}}000{\textcolor{red}{1111}} & $\neg {zz}_3 = {zz}_2 = {zz}_1 = {zz}_0$  \\ %
					$\varDelta u$  & 0b0000010000000000{\textcolor{red}{1}}010{\textcolor{red}{0001}} &  \multirow{2}{*}{$ {zz}_7 = d_7= f_{carry}(zz[7,4],u[7,4],u[3])$}              \\
					$\varDelta v$  & 0b0010000010000000{\textcolor{red}{1}}010{\textcolor{red}{0000}} &   \\  
					\toprule
					\multicolumn{3}{c}{ $zz = (z\oplus {\rm 0b1010}) \ggg 8$, $\Pr\left(M_{13} \vert M_{10}\right) = \frac{\Pr(M_{10},M_{13})}{\Pr(M_{10})} = 2^{-5.5081} > 2^{-9}= \Pr(M_{13})$}             \\ 
					\bottomrule
				\end{tabular}
				\footnotetext[*]{ One of the inputs of the second addition is denoted as $zz$, different from the output $z$ of the first modular addition, since two additions are connected by a linear operation.}
			\end{minipage}
		\end{center}
	\end{table}
	
	If the trail is verified as invalid, 
	we try to search for the root causes of invalidation and record them in the search model to avoid them during the next search. 
	Similarly, we first search for non-independent CMA in the trail. 
	If there are, we use the method in Sect. \ref{subsec:detect the contradictory} to find the location of differences that cause conflicting constraint, and record them only. 
	If there is no conflicting constraints on the adjacent bits or no non-independent CMA, which means we do not find the reason of invalidity, we record the whole trails.
	Invalid trails or differences will be converted into CNF sentences and added to our SAT-based differential search model, so we can avoid them in the next search, which can greatly shorten the time required for search. 
	
	\begin{table}[tp]
		\begin{center}
			\begin{minipage}{310pt}
				\caption{The contradictory CMA in the key schedule of invalid 14-round trails of SPECK64/128}
				\label{tab:contradictory CMA of SPECK64/128}
				\centering
				\begin{tabular}{ccc}
					\toprule
					$M_8$          & \multicolumn{1}{c}{Round 8}       & The constraints on the output         \\
					\midrule
					$\varDelta x$  & 0b0000000000{\textcolor{red}{000}}0000000000000000000 &   \\
					$\varDelta y$  & 0b0000000000{\textcolor{red}{000}}1111000000000000000 & $z_{20} = z_{19}$		\\ 
					$\varDelta z$  & 0b0000000000{\textcolor{red}{111}}1001000000000000000 & 		\\ 				
					\toprule
					$M_{11}$       & \multicolumn{1}{c}{Round 11}      & The constraints on the input    \\ 
					\midrule
					$\varDelta zz$ & 0b000000000000000000{\textcolor{red}{111}}10010000000 &   \\
					$\varDelta u$  & 0b000000000010000010{\textcolor{red}{000}}10100000000 &  $\neg zz_{12} = zz_{11} $    \\ 
					$\varDelta v$  & 0b000000000010000010{\textcolor{red}{101}}00010000000 &   \\  
					\toprule
					\multicolumn{3}{c}{$zz = (z\oplus {\rm 0b1000}) \ggg 8$, conflict constraints, $\Pr(M_{11}\vert M_8) = 0 \neq 2^8 = \Pr(M_{11}) $}      \\
					\bottomrule
				\end{tabular}
			\end{minipage}
		\end{center}
	\end{table}
	
	It is necessary to exclude invalid trails. During the search for 14-round RK differential trails of SPECK64/128, five different trails are found with a same probability: $P_D = 2^{-35}$, $P_K = 2^{-38}$, and $P_{DK} = 2^{-73}$ (much higher than the probability $2^{-88}$ in \cite{DBLP:journals/dcc/SadeghiRB21}), but they are invalid because of the conflicting constraints of additions in round 8 and 11, as shown in Table \ref{tab:contradictory CMA of SPECK64/128}. 
	According to Table \ref{tab:diff_constraint}, the differential propagation of $M_8$ on bit 21 and 20 has constraints $z_{20} =\neg c_{20} = z_{19}$, and it becomes ${zz}_{12} = {zz}_{11}$ after the linear operation, which is contradictory to constraints $\neg {zz}_{12} = d_{12} = {zz}_{11}$ of $M_{11}$. 
	Therefore, these five trails are all invalid. These impossible differences on these three bits are detected by Lemma \ref{lemma: conflicting constraints on adjacent bits }, and to avoid such kind of contradictory in the next search, we only need add the following clause 
	\begin{equation}	
		\begin{aligned}
			&(x_{21}\vee x_{20}\vee x_{19}) 
			\vee (y_{21}\vee y_{20}\vee y_{19}) 
			\vee (\lnot z_{21} \vee\lnot z_{20} \vee\lnot z_{20}) \\		
			&\vee (u_{13}\vee u_{12}\vee u_{11}) 
			\vee (v_{13}\vee v_{12}\vee v_{11}) 
		\end{aligned} 
	\end{equation}
	to the SAT-based search model, which is much more efficient and simpler than recording all five trails.
	
	\subsection{Search strategies}
	
	\subsubsection{Search for optimal RKD trails}
	\label{subsubsec:search for optimal RKD trails}
	In this section, we introduce our strategy for searching for optimal related-key differential trails.
	According to SunLing et al.\cite{DBLP:journals/tosc/SunWW21}, Matsui's branch and bound boundary conditions are added to our SAT-based search model to make reasonable use of low-round information to speed up the search process. 
	
	Our search strategy shown in Algorithm \ref{algo:search for optimal RK differential trails} is a common method \cite{DBLP:journals/tosc/SunWW21}: 
	To search for $R$ round optimal trails, start from the search for trails of a low round $r < R$ and a high probability $2^{-W}$. 
	If a valid trail is found, record the optimal probability $2^{-W}$ of round $r$, and set $(r+1,W)$ as the target of next search until $r = R$. 
	Otherwise decrease the probability, and set $(r,W+1)$ for next search.
	The function $\operatorname{Search_1}(r,W,bound)$ in Algorithm \ref{algo:search for optimal RK differential trails} is implemented in the following steps. 
	\begin{enumerate}
		\item Build a SAT-based search model for $r$ round optimal trails with the constraint of probability $w_d + w_k \leq W$ and Matsui's boundary conditions. The variable $bound$ is a list for collecting probabilities of optimal trails of different rounds to build boundary conditions.
		\item If a trail is found, verify it as shown in Sect.\ref{subsubsec:verifitcation of trails}. Otherwise, output \textit{UNSAT}.
		\item If the found trail is valid, follow the process in Fig. \ref{fig:verifiaction process} and output \textit{SAT}. Otherwise, add the invalid reason to the search model and re-search until a valid trail is find, and if no valid trail is found, output \textit{UNSAT}.
	\end{enumerate}
	\begin{algorithm}[tp]
		\caption{Search for optimal RKD trails of $R$ rounds for SPECK}%
		\label{algo:search for optimal RK differential trails}
		\begin{algorithmic}[1]
			\renewcommand{\algorithmicrequire}{ \textbf{Input:}}
			\Require   the target round $R$.     %
			\renewcommand{\algorithmicrequire}{ \textbf{Output:}}
			\Require   the probability of an optiaml trails of $R$ round.
			\State $r \leftarrow 1$
			\State $bound \leftarrow list(R)$\
			\While{$r \leq R$}
			\If{$\operatorname{Search_1}(r,W,bound) = \rm \textit{SAT}$}
			\State $ bound[r] \leftarrow  W$
			\State $ r \leftarrow r+1$
			\Else
			\State $ W \leftarrow  W+1$
			\EndIf
			\EndWhile
			\State \Return{$2^{-W}$}
		\end{algorithmic}
	\end{algorithm}

	\subsubsection{Search for good RKD trails}
	\label{subsubsec:search for good RKD trails} 
	\begin{algorithm}[tp]
		\caption{Search for good RKD trails of $R$ rounds for SPECK}%
		\label{algo:search for good differential trails}
		\begin{algorithmic}[1]
			\renewcommand{\algorithmicrequire}{ \textbf{Input:}}
			\Require   the target round $R$ and predetermined upper and lower bounds of the probabilities: $(W^1_d,W^2_d),(W^1_{dk},W^2_{dk})$.
			\renewcommand{\algorithmicrequire}{ \textbf{Output:}}
			\Require	the probability of a good trail.
			\State $W_{dk} \leftarrow W^1_{dk}, W_d \leftarrow W^1_d$
			\While{$W^1_d - W^2_d > 1$}
			\If{$\operatorname{Search_2}(R,W_d,W_{dk}) = SAT$}
			\State $ W^1_d  \leftarrow  W_d$ 
			\Else
			\State $ W^2_d  \leftarrow  W_d$
			\EndIf
			\State $ W_d  \leftarrow  \lfloor\frac{W^1_d+W^2_d}{2} \rfloor$
			\EndWhile
			\State $W_d  \leftarrow  W^1_d$
			\While{$W^1_{dk} - W^2_{dk} > 1$}
			\If{$\operatorname{Search_2}(R,W_d,W_{dk}) = SAT$}
			\State $ W^1_{dk}  \leftarrow  W_{dk}$
			\Else
			\State $ W^2_{dk}  \leftarrow  W_{dk}$
			\EndIf
			\State $ W_{dk}  \leftarrow \lfloor\frac{W^1_{dk} + W^2_{dk}}{2}\rfloor $
			\EndWhile
			\State $W_{dk}  \leftarrow  W^1_{dk}$
			\State \Return{$2^{-W_{d}},\ 2^{-W_{dk}}$} 
		\end{algorithmic}
	\end{algorithm}
	
	For the search of optimal RKD trails of high rounds, our SAT method is inefficient. So We set some restrictions to speed up the search, but this makes our path not guaranteed to be optimal.
	According to the optimal trails found in the previous section, such as Table \ref{tab:The RKD trails of 11 rounds for SPECK32/64} and Table \ref{tab:The RKD trails of 15 rounds for SPECK64/128}, it is found that all these optimal RKD trails have one thing in common: 
	in their key schedule, there exists at most three consecutive rounds whose sub-key has no differences.
	And these three sub-keys lead to four consecutive rounds of data encryption with no differences. 
	It is worth noting that this phenomenon was first discovered in \cite{DBLP:journals/dcc/SadeghiRB21}, and helped them found good RKD trails for SPECK. 
	Same as their work, we set the differences of sub-keys of three consecutive rounds and that of the corresponding round of data encryption to zero during the search.

	The details of our strategy is shown in Algorithm \ref{algo:search for good differential trails}, similar to Algorithm 1 of \cite{liu2017rotational}.
	Different to the SAT-based search model in Sect. \ref{subsubsec:SAT for search}, we increase the probability constraint into two parts: the probability of encryption part $w_d \leq W_d $ and the probability of the entire trail $ w_{dk} = w_d + w_k \leq W_{dk}$. 
	Moreover, the Matsui's condition is not applied in this model because we have no information about the probability of optimal trails of lower rounds.
	Therefore, the function $\operatorname{Search_2}(R,W_d,W_{dk})$ is to find RKD trails of $R$ rounds for SPECK with $w_d \leq W_d $ and $w_{dk} \leq W_{dk}$, and verify them same as $\operatorname{Search_1}(r,W,bound)$. 	
	
	Preset the search range of probabilities as $W^2_d \leq W_d \leq W^1_d$ and $W^2_{dk} \leq W_{dk} \leq W^1_{dk}$.
	We firstly search for RKD trails with lowest $W_{d}$ under a certain trail probability $W_{dk} = W^1_{dk}$. The ideal situation is that there is one valid trail found in the upper bound $W^1_d$ and no trails in the lower bound $ W^2_d$, so the best data probability is $w_d = W^1_d$ when $W^1_d - W^2_d = 1$. These boundaries can be extended if necessary.
	Then, with the fixed $W_{d}$, we search for trails with best $W_{dk}$.

	\subsection{Search results}
	\label{subsec:seach result}
	With the help of SAT solver CaDiCal \cite{BiereFazekasFleuryHeisinger-SAT-Competition-2020-solvers} and \#SAT solver GANAK \cite{DBLP:conf/ijcai/SharmaRSM19}, we find better trails for SPECK32/64, SPECK48/96, and SPECK64/128 than previous works.%
	
	\subsubsection{SPECK32/64}
	\label{subsubsec:Speck32/64}		
	\begin{table}[tp]
		\begin{center}
			\begin{minipage}{\textwidth}
				\caption{The comparison of our RKD trails with the result of \cite{DBLP:journals/dcc/SadeghiRB21} for SPECK32/64}
				\label{tab:speck32/64}
				\renewcommand{\arraystretch}{1.2}%
				\renewcommand{\thefootnote}{\fnsymbol{footnote}}
				\begin{tabular*}{\textwidth}{@{\extracolsep{\fill}}cccccc@{\extracolsep{\fill}}}
					\toprule
					Round & $\log_2{P_{DK}}$ & $\log_2{P_{D}}$ & $\log_2{P_{K}}$ & Search time & Ref.         \\
					\midrule
					10  &$-20$  &$-13$ &$-7$    & $>$ 1 d & MILP (Optimal)  \\
					11  &$-31$  &$-17$ &$-14$   &       & \cite{DBLP:journals/dcc/SadeghiRB21}  \\
					12  &$-37$  &$-24$ &$-13$   &       & \cite{DBLP:journals/dcc/SadeghiRB21}  \\
					13  &$-47$  &$-24$ &$-23$   &       & \cite{DBLP:journals/dcc/SadeghiRB21}  \\
					15  &$-94$  &$-32$ &$-62$   &      	& \cite{DBLP:journals/dcc/SadeghiRB21}  \\
					\midrule
					10  &$-20$ &$-13$ &$-7 $    &176 s       & Our(Optimal) \\
					11  &$-28$ &$-17$ &$-11$    &2922 s      & Our(Optimal) \\
					12  &$-37$ &$-24$ &$-13$    &69153 s     & Our(Optimal) \\
					13  &$-47$ &$-24$ &$-23$    &18.6 d      & Our(Optimal) \\
					15  &$-85$ &$-32$ &$-53$    &       	 & Our         \\
					\botrule
				\end{tabular*}
				\footnotetext{Note: the trails we find for SPECK32/64 include no non-independent CMAs, and the probabilities in this table are computed under the independence assumption.}
			\end{minipage}
		\end{center}
	\end{table}
	We find optimal trails of 10 to 13 rounds for SPECK32/64 as shown in Table \ref{tab:speck32/64}.
	For 10 rounds, our SAT method takes 3 minutes to find the optimal trails, while our MILP method takes more than one day.
	
	For 11 rounds, the optimal trails we found have higher probability than the work of \cite{DBLP:journals/dcc/SadeghiRB21}.
	
	For 12 and 13 rounds, our optimal trails have the same probability to that of \cite{DBLP:journals/dcc/SadeghiRB21} 
	searched under the assumption that three consecutive rounds have no differences, which can not be proved optimal.
	
	In addition, the valid trails of 15 rounds we find using the strategy of Sect. \ref{subsubsec:search for good RKD trails} have higher probabilities and no non-independent CMA.
	It seems that the SAT method is more suitable for the trail search for ARX-based ciphers than MILP method, and the assumption above is useful to find high probability trails. 
	
	\subsubsection{SPECK48/96}
	\label{subsubsec:Speck48/96}
	\begin{table}[tp]
		\begin{center}
			\begin{minipage}{\textwidth}
				\caption{The comparison of our RKD trails with the result of \cite{DBLP:journals/dcc/SadeghiRB21} for SPECK48/96}
				\label{tab:speck48/96}
				\renewcommand{\arraystretch}{1.2}%
				\renewcommand{\thefootnote}{\fnsymbol{footnote}}
				\begin{tabular*}{\textwidth}{@{\extracolsep{\fill}}cccccc@{\extracolsep{\fill}}}
					\toprule
					Round & $\log_2{P_{DK}}$ & $\log_2{P_{D}}$ & $\log_2{P_{K}}$ & Search time & Ref.         \\
					\midrule
					11  &$-30$ &$-17$ &$-13$    &   & \cite{DBLP:journals/dcc/SadeghiRB21}      \\
					12  &$-44$ &$-21$ &$-23$    &   & \cite{DBLP:journals/dcc/SadeghiRB21}      \\
					14  &$-68$ &$-43$ &$-25$    &   & \cite{DBLP:journals/dcc/SadeghiRB21} \\
					15  &$-89$ &$-46$ &$-43$    &   & \cite{DBLP:journals/dcc/SadeghiRB21} \\
					\midrule
					11  &$-29$ &$-18$ &$-11$    & 10842 s     & Our(Optimal) \\
					12  &$-40$ &$-25$ &$-15$    & 5.3d        & Our(Optimal) \\
					14  &$-66.5906(-67)$\footnotemark[1] &$-43$ &$-23.5906(-24)$ & & Our \\			
					15  &$-83.5081(-87)$ &$-42$ &$-41.5081(-45)$ &      		  & Our	\\
					\botrule
				\end{tabular*}
				\footnotetext[*]{ Note: the 14- and 15-round trails we find include several non-independent CMAs and the values in brackets are their probabilities computed under the independence assumption.}
			\end{minipage}
		\end{center}
	\end{table}
	The probabilities of the trails for SPECK48/96 we found are listed in Table \ref{tab:speck48/96}.
	The optimal trails of 11 and 12 rounds we find in Sect. \ref{subsubsec:search for optimal RKD trails} have higher probabilities than the former work, and they are all proved to be valid and have no non-independent CMA.	
	
	For the 14 round version of SPECK48/96, we find several trails with probability $2^{-67}$ under the independence assumption and all of them include a non-independent CMA in Round 0 and 3 of the key schedule as shown in Table \ref{tab:CMA of 14round SPECK48/96}. It can be proved that the input states of the CMA above are uniform random, and its true probability is $2^{-8.5906}$ calculated by GANAK. Therefore, the probability of the trails we find is $2^{-66.5906}$.
	
	\begin{table}[tp]
		\begin{center}
			\begin{minipage}{280pt}
				\caption{The CMA in Round 0 and 3 of the key schedule of 14 round SPECK48/96}
				\label{tab:CMA of 14round SPECK48/96}
				\centering
				\begin{tabular}{ccc}
					\toprule
					$M_{0}$          & Round 0      & The constraints on the output                         \\ 
					\midrule
					$\varDelta x$  & 0b1100010000000000100{\textcolor{red}{1}}0010 & \multirow{3}{*}{$z_{4} = \neg c_{4}$}  \\
					$\varDelta y$  & 0b0100010000001000000{\textcolor{red}{1}}0000 &   \\ 
					$\varDelta z$  & 0b0000000000001000100{\textcolor{red}{0}}0010 &         \\ 
					\toprule
					$M_{3}$       & Round 3      & The constraints on the input         \\
					\midrule
					$\varDelta zz$ & 0b100{\textcolor{red}{0}}00100000000000001000 & \multirow{3}{*}{${zz}_{20} = d_{20}$  }    \\ 
					$\varDelta u$  & 0b000{\textcolor{red}{1}}00100000000000001000 &      \\
					$\varDelta v$  & 0b100{\textcolor{red}{1}}00000000000000000000 &      \\
					\midrule
					\multicolumn{3}{c}{$zz = (z\oplus {\rm 0}) \ggg 8$, $\Pr(M_{0},M_{3}) = 2^{-8.5906} > 2^{-9}= \Pr(M_{0}) \times\Pr(M_{3})$}   \\ 
					\botrule
				\end{tabular}
			\end{minipage}
		\end{center}
	\end{table}
	
	For the 15 round version of SPECK48/96, the trail we find under the independence assumption have the probability of $2^{-87}$ higher than the previous work, and the CMA in the Round 10 and 13 of the key schedule is non-independent as shown in Table \ref{tab:CMA of 15round SPECK48/96}. The probability of the trail we find is $2^{-83.5081}$.

	\subsubsection{SPECK64/128}
	\label{subsubsec:speck64/128}
	The trails of 14 and 15 rounds we find under the independence assumption have much higher probability than that in \cite{DBLP:journals/dcc/SadeghiRB21}, as shown in Table \ref{tab:speck64/128}.
	
	\begin{table}[tp]
		\begin{center}
			\begin{minipage}{230pt}
				\caption{The comparison of our RKD trails with the result of \cite{DBLP:journals/dcc/SadeghiRB21} for SPECK64/128}
				\label{tab:speck64/128}
				\renewcommand{\arraystretch}{1.3}%
				\setlength{\tabcolsep}{8.0pt}
				\renewcommand{\thefootnote}{\fnsymbol{footnote}}
				\begin{tabular}{ccccc}
					\toprule
					Round &$\log_2{P_{DK}}$    &$\log_2{P_{D}}$ &$\log_2{P_{K}}$  & Ref.         \\
					\midrule
					14  &$-88 $    &$-37$    &$-51$ & \cite{DBLP:journals/dcc/SadeghiRB21} \\
					15  &$-105$    &$-45$    &$-60$ & \cite{DBLP:journals/dcc/SadeghiRB21} \\
					\midrule
					14  &$-72(79)\footnotemark[1]$    &$-35$    &$-37(-44)$      		      		  & Our \\			
					15  &$-89(97)$    &$-42$    &$-47(-55)$	             		  & Our	\\
					\botrule
				\end{tabular}
				\footnotetext[*]{Note: the 14- and 15-round trails we find include several non-independent CMAs and the values in brackets are their probabilities computed under the independence assumption.}
			\end{minipage}
		\end{center}
	\end{table}
	
	As for the 14 round SPECK64/128, we find 8 trails with $P_{D} = 2^{-35}$ and $P_{DK} = 2^{-79}$ under the independence assumption, and two of them are invalid. In the valid trails, two have two non-independent CMAs, and four have three which makes their accurate probability much higher. The trails with highest probability have non-independent CMAs in Round 7 and 10, Round 8 and 11, and Round 9 and 12 of the key schedule as shown in Table \ref{tab:(7,10) CMA of 14 round SPECK64/128}, Table \ref{tab:(8,11) CMA of 14 round SPECK64/128}, and Table \ref{tab:(9,12) CMA of 14 round SPECK64/128}.  
	And the more accurate probability of this trail calculated by our method is $2^{-72}$, $2^{7}$ times higher than it calculated under the independence assumption.
	
	As for the 15 round SPECK64/128, the trail with probability $2^{-97}$ we find has four non-independent CMAs as shown in Table \ref{tab:(7,10) CMA of 15 round SPECK64/128}, Table \ref{tab:(10,13) CMA of 15 round SPECK64/128}, Table \ref{tab:(8,11) CMA of 15 round SPECK64/128}, and Table \ref{tab:(9,12) CMA of 15 round SPECK64/128}, and the more accurate probability computed by our method is $2^{-89}$. 
	
	\section{Application on Chaskey}
	\label{sec:application on Chaskey}
	From Fig. \ref{fig:Chaskey}, the round function of Chaskey includes two consecutive modular additions. Since there is no key insertion between different rounds, the modular additions are connected directly, which means there are two addition chains in the differential trails of Chaskey. Under the independence assumption, the differential probability is calculated by the product of the probabilities of each addition on these two chains. The designers of Chaskey provide a best found differential trail for 8 rounds in the Table 4 of \cite{mouha2014chaskey}, and its probability is $2^{-293}$ (under independence assumption) and $2^{-289.9}$ (calculated by Leurent’s ARX Toolkit). In this section, we are going to calculate the probability of this 8 round trail with the consideration of non-independence.
	
	Assuming that two modular addition chains in the trail of 8 rounds are independent, and every four additions on each chain are independent too. We build SAT models of state propagation on these CMAs, and use the solver GANAK to calculate their solution numbers. The results are shown in Table \ref{tab: Chain 1 in 8 round Chaskey} and Table \ref{tab: Chain 2 in 8 round Chaskey}.
	
	\begin{table}[tp]
		\begin{center}
			\begin{minipage}{310pt}
				\caption{The probability of Chain 1 in the 8 rounds differential trail for Chaskey in \cite{mouha2014chaskey}}
				\renewcommand{\arraystretch}{0.9}%
				\label{tab: Chain 1 in 8 round Chaskey}
				\begin{tabular}{cccc}
					\toprule
					\multicolumn{2}{c}{Chain 1}  & \multirow{2}{*}{$\log_2\left(\prod_{r=i}^{i+3}\Pr \left(M_r\right)  \right)$} & \multirow{2}{*}{$\log_2\left(\Pr\left(M_i,M_{i+1},M_{i+2},M_{i+3}\right)\right)$}     \\ 
					\cmidrule{1-2}
					Round               & Addition &                                   &                         \\ 
					\midrule
					\multirow{2}{*}{1}  & 0       & \multirow{4}{*}{-70}         	  & \multirow{4}{*}{-67.1106} \\ %
					& 1       &                                   &                           \\ 
					\multirow{2}{*}{2}  & 2       &                                   &                           \\ %
					& 3       &                                   &                           \\ 
					\midrule
					\multirow{2}{*}{3}  & 4       & \multirow{4}{*}{-5}           	  & \multirow{4}{*}{-5}       \\ %
					& 5       &                                   &                           \\ 
					\multirow{2}{*}{4}  & 6       &                                   &                           \\ %
					& 7       &                                   &                           \\ 
					\midrule
					\multirow{2}{*}{5}  & 8       & 0                                 & 0                   	  \\ 
					\cmidrule{2-4} 
					& 9       & \multirow{3}{*}{-11}       		  & \multirow{3}{*}{-10.5793} \\ 
					\multirow{2}{*}{6}  & 10      &                                   &                           \\ %
					& 11      &                                   &                           \\
					\midrule
					\multirow{2}{*}{7}  & 12      & \multirow{4}{*}{-66}         	  & \multirow{4}{*}{-64.5961} \\ %
					& 13      &                                   &                           \\ 
					\multirow{2}{*}{8}  & 14      &                                   &                           \\ %
					& 15      &                                   &                           \\ 
					\midrule
					\multicolumn{2}{c}{$P_1$}       & -152                          & -147.2889                 \\ \botrule
				\end{tabular}
			\end{minipage}
		\end{center}
	\end{table}
	
	\begin{table}[htp]
		\begin{center}
			\begin{minipage}{310pt}
				\caption{The probability of Chain 2 in the 8 rounds differential trail for Chaskey in \cite{mouha2014chaskey}}
				\renewcommand{\arraystretch}{0.9}%
				\label{tab: Chain 2 in 8 round Chaskey}
				\begin{tabular}{cccc}
					\toprule
					\multicolumn{2}{c}{Chain 2} & \multirow{2}{*}{$\log_2\left(\prod_{r=i}^{i+3}\Pr \left(M_r\right)  \right)$} & \multirow{2}{*}{$\log_2\left(\Pr\left(M_i,M_{i+1},M_{i+2},M_{i+3}\right)\right)$}    \\ 
					\cmidrule{1-2}
					Round               & Addition &                                   &                       \\ 
					\midrule
					\multirow{2}{*}{1}  & 0       & \multirow{4}{*}{$-61$}       & \multirow{4}{*}{$-60.1647$} \\ %
					& 1       &                                   &                           \\ 
					\multirow{2}{*}{2}  & 2       &                                   &                           \\ %
					& 3       &                                   &                           \\ 
					\midrule
					\multirow{2}{*}{3}  & 4       & \multirow{4}{*}{$-11$}      & \multirow{4}{*}{$-11$} \\ %
					& 5       &                                   &                           \\ 
					\multirow{2}{*}{4}  & 6       &                                   &                           \\ %
					& 7       &                                   &                           \\ 
					\midrule
					\multirow{2}{*}{5}  & 8       & \multirow{4}{*}{$-12$}      & \multirow{4}{*}{$-12$}      \\  %
					& 9       &              					  &        					  \\ 
					\multirow{2}{*}{6}  & 10      &                                   &                           \\ %
					& 11      &                                   &                           \\ 
					\midrule
					\multirow{2}{*}{7}  & 12      & \multirow{4}{*}{$-57$}      & \multirow{4}{*}{$-54.7177$} \\ %
					& 13      &                                   &                           \\ 
					\multirow{2}{*}{8}  & 14      &                                   &                           \\ %
					& 15      &                                   &                           \\ 
					\midrule
					\multicolumn{2}{c}{$P_2$}       & $-141$                       & $-137.8824 $                \\ 
					\botrule
				\end{tabular}
			\end{minipage}
		\end{center}
	\end{table}
	
	Calculated by the solver GANAK, the differential probabilities of the Addition (0,1,2,3) and (12,13,14,15) in Chain 1 and Chain 2 are $2^{-67.1106}$, $2^{-64.5961}$, $2^{-60.1647}$, and $2^{-54.7177}$ respectively.   
	
	As for Addition (4,5,6,7) of Chain 1 and 2, there is no association between the differential constraints of each pair of consecutive additions, so we calculate their differential probabilities by multiplying the probability of each addition.
	
	As for Addition (8,9,10,11) of Chain 2, we find the constraints of each addition are not associated with the carry states and the outputs of them are uniform random. Therefore, we believe that these four modular additions are independent of each other.
	
	As for Addition (8,9,10,11) of Chain 1, we use GANAK to calculate the total difference probability of the last three modular additions since the output of the 8th addition is uniform random.
	
	The differential probability of Chain 1 and 2 are $2^{-147.2889}$ and $2^{-137.8824}$ respectively, so the total probability of the 8 round trail is actually $2^{-285.1713}$, higher than $2^{-293}$ (under independence assumption) and $2^{-289.9}$ (calculated by Leurent’s ARX Toolkit) in \cite{mouha2014chaskey}.
	
	\section{Conclusion}
	\label{sec:conclusion}
	In this paper, we study the differential properties of the single and consecutive modular addition from differential constraints on its inputs and output. We find the influence of non-independence is connected with the relationship between the differential constraints of these two additions on the intermediate state. 
	If those constraints are contradictory, then the differential trails of the CMA is impossible (invalid).
	If these constraints make the bit values of same positions in the intermediate state related to each other, then the non-independence will affect the probability of the differential trail on the CMA, and the probabilities calculated under the independence assumption is inaccurate.
	
	We introduce a SAT-based method to capture the state propagation of a given differential on modular addition. 
	By this method, we build a big SAT model to verify whether a given differential trail of an ARX cipher includes any right pairs.
	By this method, we build the \#SAT model to accurately calculate the probability of a given differential on CMAs, which can help to get the more accurate probability of the entire trail. We are the first to consider the accurate calculation of differential probability of CMA.
	Given a differential trail of ARX-based ciphers, we introduce a set of inspection procedures to detect the validity of the trail, and if it is invalid, find and record the difference that caused the invalid. Other, find the non-independent CMAs in the trail and calculate its accurate differential probabilities.
	
	We apply these methods to search RKD trails of SPECK family of block ciphers.
	Under our search strategies, we find better trails than the work in \cite{DBLP:journals/dcc/SadeghiRB21}.
	In addition, we apply our \#SAT method to calculate the probability of 8 round differential trail of Chaskey, given by its designer \cite{mouha2014chaskey}. After a more accurate calculation, we find the probability of this trail is much high than that calculated under the independence assumption.
	
	In this paper, we only have a deeper understanding of the non-independence of the modular addition, and more efforts are needed to study the difference properties of the modular addition. If a theory is studied from the perspective of difference to describe the differential property of consecutive modulus addition, then better trails will be found for ARX ciphers. There are many ARX-based ciphers, such as ChaCha, Siphash, and Sparckle, includes CMA in their round functions, so it is necessary to take the non-independence into consideration.

	\begin{appendices}
		\section{Experiments to prove the effectiveness of our method}
		\label{sec:Experiments to prove the effectiveness of our method}
		In order to verify the effectiveness of our method, we construct several toy ciphers with small block size to evaluate the difference between the differential probability of an optimal trail calculated by our method and its real value computed by traversing. 
		For example, to compute the real probability of an $R$ round trail for a toy cipher with the 28-bit block size, we prepare all $2^{28}$ plaintext pairs to encrypt and compute the differential probability of each round, denoted as $P^i$ (for $i = 0,1,\cdots,R$), by counting the number of pairs that satisfy the difference.
		
		Firstly, we build two toys of Chaskey with a 32-bit and a 28-bit block size, denoted as toy Chaskey-32 and toy Chaskey-28, and their rotation constants are $\{\{3,2\}, 3, \{2,2\}, 3 \}$ that is
		\begin{equation}
			\begin{aligned}
				w_1^r &= (v_1^r \lll 3 ) \oplus (w_0^r \ggg 3),\\ 
				w_3^r &= (v_3^r \lll 2 ) \oplus w_2^r,\\
				v_3^{r+1} &= (w_3^r \lll 2 ) \oplus v_0^{r+1}, \\ 
				v_1^{r+1} &= (w_1^r \lll 2 ) \oplus (v_2^{r+1} \ggg 3 ),
			\end{aligned}
		\end{equation}
		with the structure shown in Fig.\ref{fig:Chaskey}.
		In Table \ref{tab:toy chaskey wordlen 8 + round 5}, we find the optimal trail of 5 rounds for the toy Chaskey-32 with a probability of $2^{27}$ under the independence assumption, and the probability $2^{-25}$ computed by our method is much closer to its real value $2^{-24.89148}$.
		For this trail, the main difference in probability calculation is in round 1 and round 3 where the CMAs $ (M_2^1, M_4^1)$ and $(M_2^3, M_4^3)$ are not independent found by our method. 
		As shown in Table \ref{tab:The CMA in the round 1 of the toy Chaskey-32} and \ref{tab:The CMA in the round 3 of the toy Chaskey-32}, they have higher probabilities than that computed under the independence assumption because there are same differential constraints between the additions. 
		So the probabilities of round 1 and 3 in our method are $2^{-5}$ and $2^{-8}$, and from the experiment, it is very close to the real value $2^{-4.99859}$ and $2^{-7.97459}$. 
		For the toy Chaskey-28, we find the optimal trails of 7 and 6 rounds and it can be find in Table \ref{tab:toy chaskey wordlen 7 + round 7} and Table \ref{tab:toy chaskey wordlen 7 + round 6} that their probabilities calculated by our method are both closer to the real value than that computed by the previous way.
		
		Secondly, we build several toy ciphers with the similar structure as the round function of SPECK in which the insertion of the subkey is replaced by the round number. The toy cipher with a 28-bit block size and rotation constant $[6,3]$ is denoted as the toy SPECK-28, and its encryption is 
		\begin{equation}
			\begin{aligned}
				x^{i} &= ((x^{i-1} \ggg 6 ) \oplus y^{i-1}) \oplus (i-1), \\ 
				y^{i} &= x^{i} \oplus (y^{i-1} \lll 3 ),
			\end{aligned}
		\end{equation}
		similar to Fig.\ref{fig:speck}. In Table \ref{tab:toy speck wordlen 8 + round 14}, we find an optimal differential trail of 14 rounds for the toy SPECK-28, and its probability $2^{-22.41504}$ computed by our method is closer to the real value $2^{-22.35614}$ than $2^{-23}$ computed under the independence assumption. Our method find that the CMA including the addition in round 3 and 4 are non-independent, and the uneven output of the previous addition causes the probability of differential propagation on the latter addition higher than that computed independently. 
		And the differential probability of round 4 computed by our method is closed to the real value.  
		
		From the experiments on these toy ciphers, it turns out that our method can calculate the probability of differential trails more accurately than the previous way. 
		
		\begin{table}[tp]
			\begin{minipage}{\textwidth}
				\caption{An optimal differential trail of 5 rounds for the toy Chaskey-32 }
				\label{tab:toy chaskey wordlen 8 + round 5} 
				\setlength{\tabcolsep}{2.2 pt}%
				\renewcommand{\arraystretch}{1.2}%
				\begin{tabular*}{\textwidth}{cccccccc}
					\toprule
					\multicolumn{5}{c}{Differential trail} & \multicolumn{3}{c}{ $-\log_2P^i$} \\
					\cmidrule{6-8}          
					\multirow{2}{*}{Round} & \multirow{2}{*}{$\varDelta v_0$} & \multirow{2}{*}{$\varDelta v_1$} & \multirow{2}{*}{$\varDelta v_2$} & \multirow{2}{*}{$\varDelta v_3$} &  Independent  & Our  & Real \\
					&       &       &       &       & Assumption  & method & value  \\
					\midrule
					0     & 0b00000000 & 0b00000000 & 0b00100000 & 0b00100000 & 1     & 1     & 1 \\
					1     & 0b10000000 & 0b00000000 & 0b00000000 & 0b10000010 & \textcolor[rgb]{ 1,  0,  0}{ 6 } & \textcolor[rgb]{ 1,  0,  0}{ 5} & \textcolor[rgb]{ 1,  0,  0}{4.99859} \\
					2     & 0b10001000 & 0b00000000 & 0b00010000 & 0b10111010 & 10    & 10    & 10.00445 \\
					3     & 0b00000100 & 0b00000000 & 0b00010001 & 0b00000101 & \textcolor[rgb]{ 1,  0,  0}{9} & \textcolor[rgb]{ 1,  0,  0}{8} & \textcolor[rgb]{ 1,  0,  0}{7.97459} \\
					4     & 0b00000000 & 0b00000000 & 0b10000000 & 0b10000000 & 1     & 1     & 0.91384 \\
					5     & 0b00000010 & 0b00000000 & 0b00000000 & 0b00001010 &       &       &  \\
					\midrule
					&       &  \multicolumn{3}{r}{ $-\log_2P =-\sum_{i=0}^{4}\log_2P^i $}      & 27    & 25    & 24.89148 \\
					\bottomrule
				\end{tabular*}
				\footnotetext{Note: our method find that the CMAs in round 1 and 3 are nonindependent and their probabilities computed by our method are close to the real value.}
			\end{minipage}
		\end{table}%
		
		\begin{table}[tp]
			\begin{minipage}[t]{0.5\textwidth}
				\centering
				\caption{The CMA in the round 1 of the toy Chaskey-32}
				\label{tab:The CMA in the round 1 of the toy Chaskey-32}	
				\begin{tabular}{ccc}
					\toprule
					$M^1_{2}$        & Round 1         & Constraint                         \\ 
					\midrule
					$\varDelta x$  &0b00000000 &   \\
					$\varDelta y$  &0b10000010 & $z_{2}=\lnot z_{1}$ \\ 
					$\varDelta z$  &0b10000110 &         \\ 
					\toprule
					$M^1_{4}$       & Round 1        & Constraint          \\
					\midrule
					$\varDelta z$  & 0b10000110 & \\ 
					$\varDelta u$  & 0b10000000 & $\lnot{z}_2={z}_1$     \\
					$\varDelta v$  & 0b00000010 &   \\
					\midrule
					\multicolumn{3}{c}{$\Pr\left(M^1_{4}\middle\vert M^1_2\right)%
						=2^{-1} >2^{-2} = \Pr(M^1_{4})$}   \\
					\botrule
				\end{tabular}
			\end{minipage}
			\begin{minipage}[t]{0.5\textwidth}
				\centering
				\caption{The CMA in the round 3 of the toy Chaskey-32}
				\label{tab:The CMA in the round 3 of the toy Chaskey-32}	
				\begin{tabular}{ccc}
					\toprule
					$M^3_{2}$        & Round 3         & Constraint                          \\ 
					\midrule
					$\varDelta x$  &0b00010001 &   \\
					$\varDelta y$  &0b00000101 & $z_{5}=\lnot z_{4}$ \\ 
					$\varDelta z$  &0b00110100 &         \\ 
					\toprule
					$M^3_{4}$       & Round 3        & Constraint          \\
					\midrule
					$\varDelta z$  & 0b00110100 & \\ 
					$\varDelta u$  & 0b00000100 & $\lnot{z}_5={z}_4$     \\
					$\varDelta v$  & 0b00010000 &   \\
					\midrule
					\multicolumn{3}{c}{$\Pr\left(M^3_{4}\middle\vert M^3_2\right) %
						=2^{-2}   >2^{-3} = \Pr(M^3_{4})$}   \\ 
					\botrule
				\end{tabular}
			\end{minipage}
		\end{table}

		\begin{table}[tp]
			\begin{minipage}{\textwidth}
				\caption{A optimal differential trail of 7 rounds for the toy Chaskey-28 }
				\label{tab:toy chaskey wordlen 7 + round 7}
				\setlength{\tabcolsep}{3pt}%
				\begin{tabular*}{\textwidth}{cccccccc}
					\toprule
					\multicolumn{5}{c}{Differential trail} & \multicolumn{3}{c}{ $-\log_2P^i$} \\
					\cmidrule{6-8}          
					\multirow{2}{*}{Round} & \multirow{2}{*}{$\varDelta v_0$} & \multirow{2}{*}{$\varDelta v_1$} & \multirow{2}{*}{$\varDelta v_2$} & \multirow{2}{*}{$\varDelta v_3$} &  Independent  & Our  & Real \\
					&       &       &       &       & Assumption  & method & value  \\
					\midrule
					0     & 0b1000000 & 0b0000000 & 0b0000011 & 0b1000001 & 4     & 4     & 4 \\
					1     & 0b1000000 & 0b0000000 & 0b0010000 & 0b1010010 & \textcolor[rgb]{ 1,  0,  0}{8} & \textcolor[rgb]{ 1,  0,  0}{6} & \textcolor[rgb]{ 1,  0,  0}{6} \\
					2     & 0b0000000 & 0b0000000 & 0b0010000 & 0b0010000 & 1     & 1     & 1 \\
					3     & 0b1000000 & 0b0000000 & 0b0000000 & 0b1000010 & \textcolor[rgb]{ 1,  0,  0}{7} & \textcolor[rgb]{ 1,  0,  0}{5} & \textcolor[rgb]{ 1,  0,  0}{4.97625} \\
					4     & 0b1000000 & 0b0000000 & 0b0010000 & 0b1010010 & \textcolor[rgb]{ 1,  0,  0}{8} & \textcolor[rgb]{ 1,  0,  0}{6} & \textcolor[rgb]{ 1,  0,  0}{5.81430} \\
					5     & 0b0000000 & 0b0000000 & 0b0010000 & 0b0010000 & 1     & 1     & 0.81714 \\
					6     & 0b1000000 & 0b0000000 & 0b0000000 & 0b1000010 & 4     & 4     & 3.39232 \\
					7     & 0b1000100 & 0b0000000 & 0b0010000 & 0b1100110 &       &       &  \\
					\midrule
					&       &    \multicolumn{3}{r}{ $-\log_2P =-\sum_{i=0}^{6}\log_2P^i $}      & 33    & 27    & 26 \\
					\bottomrule
				\end{tabular*}
				\footnotetext{Note: our method find that the CMAs in round 1, 3, and 4 are nonindependent and their probabilities computed by our method are close to the real value.}
			\end{minipage}
		\end{table}%
		
		\begin{table}[tp]
			\begin{minipage}{\textwidth}
				\caption{A optimal differential trail of 6 rounds for the toy Chaskey-28 }
				\label{tab:toy chaskey wordlen 7 + round 6}
				\setlength{\tabcolsep}{3pt}%
				\begin{tabular*}{\textwidth}{cccccccc}
					\toprule
					\multicolumn{5}{c}{Differential trail} & \multicolumn{3}{c}{ $-\log_2P^i$} \\
					\cmidrule{6-8}          
					\multirow{2}{*}{Round} & \multirow{2}{*}{$\varDelta v_0$} & \multirow{2}{*}{$\varDelta v_1$} & \multirow{2}{*}{$\varDelta v_2$} & \multirow{2}{*}{$\varDelta v_3$} &  Independent  & Our  & Real \\
					&       &       &       &       & Assumption  & method & value  \\
					\midrule
					0     & 0b1000000 & 0b0000000 & 0b0010001 & 0b1010001 & 5     & 5     & 5 \\
					1     & 0b0000000 & 0b0000000 & 0b0010000 & 0b0010000 & 1     & 1     & 1 \\
					2     & 0b1000000 & 0b0000000 & 0b0000000 & 0b1000010 & \textcolor[rgb]{ 1,  0,  0}{7} & \textcolor[rgb]{ 1,  0,  0}{5} & \textcolor[rgb]{ 1,  0,  0}{5.00565} \\
					3     & 0b1000000 & 0b0000000 & 0b0010000 & 0b1010010 & \textcolor[rgb]{ 1,  0,  0}{8} & \textcolor[rgb]{ 1,  0,  0}{6} & \textcolor[rgb]{ 1,  0,  0}{5.99154} \\
					4     & 0b0000000 & 0b0000000 & 0b0010000 & 0b0010000 & 1     & 1     & 0.96390 \\
					5     & 0b1000000 & 0b0000000 & 0b0000000 & 0b1000010 & 4     & 4     & 3.82947 \\
					6     & 0b1000100 & 0b0000000 & 0b0010000 & 0b1100110 &       &       &  \\
					\midrule
					&       &    \multicolumn{3}{r}{ $-\log_2P =-\sum_{i=0}^{5}\log_2P^i $}      & 26    & 22    & 21.79055 \\
					\bottomrule
				\end{tabular*}
				\footnotetext{Note: our method find that the CMAs in round 2 and 3 are nonindependent and their probabilities computed by our method are close to the real value.}
			\end{minipage}
		\end{table}%
		\iffalse
		\begin{table}[tp]
			\begin{minipage}{\textwidth}
				\caption{A optimal differential trail of 5 rounds for the toy chaskey with wordsize 7 }
				\label{tab:toy chaskey wordlen 7 + round 5}
				%
				\setlength{\tabcolsep}{3pt}%
				%
				\begin{tabular*}{\textwidth}{cccccccc}
					\toprule
					\multicolumn{5}{c}{Differential trail} & \multicolumn{3}{c}{ $-\log_2P^i$} \\
					\cmidrule{6-8}          
					\multirow{2}{*}{Round} & \multirow{2}{*}{$\varDelta v_0$} & \multirow{2}{*}{$\varDelta v_1$} & \multirow{2}{*}{$\varDelta v_2$} & \multirow{2}{*}{$\varDelta v_3$} &  Independent  & Our  & Real \\
					&       &       &       &       & Assumption  & method & value  \\
					\midrule
					0     & 0b0000000 & 0b0000000 & 0b0010000 & 0b0010000 & 1     & 1     & 1 \\
					1     & 0b1000000 & 0b0000000 & 0b0000000 & 0b1000010 & \textcolor[rgb]{ 1,  0,  0}{7} & \textcolor[rgb]{ 1,  0,  0}{5} & \textcolor[rgb]{ 1,  0,  0}{5} \\
					2     & 0b1000000 & 0b0000000 & 0b0010000 & 0b1010010 & \textcolor[rgb]{ 1,  0,  0}{8} & \textcolor[rgb]{ 1,  0,  0}{6} & \textcolor[rgb]{ 1,  0,  0}{6.00026} \\
					3     & 0b0000000 & 0b0000000 & 0b0010000 & 0b0010000 & 1     & 1     & 1.00724 \\
					4     & 0b1000000 & 0b0000000 & 0b0000000 & 0b1000010 & 4     & 4     & 3.98407 \\
					5     & 0b1000100 & 0b0000000 & 0b0010000 & 0b1100110 &       &       &  \\
					\midrule
					&  &\multicolumn{3}{r}{ $-\log_2P =-\sum_{i=0}^{4}\log_2P^i $} & 21    & 17    & 16.99157 \\
					\bottomrule
				\end{tabular*}
			\end{minipage}
		\end{table} \fi

		\begin{table}[tp]
			\begin{minipage}{\textwidth}
				\caption{A optimal differential trail of 8 rounds for the toy SPECK-28 }
				\label{tab:toy speck wordlen 8 + round 14}
				\setlength{\tabcolsep}{5pt}%
				\begin{tabular*}{\textwidth}{cccccc}
					\toprule
					\multicolumn{3}{c}{Differential trail}& \multicolumn{3}{c}{$-\log_2 P^i$} \\
					\cmidrule{4-6}          
					\multirow{2}{*}{Round} & \multirow{2}{*}{$\varDelta x$} & \multirow{2}{*}{$\varDelta y$} &  Independent  & Our  & Real \\
					&       &       & Assumption  & method & value  \\
					\midrule
					0     &  0b00100101000000 & 0b00000000100100 & 3     & 3     & 3 \\
					1     &  0b00000000000001 & 0b00000100100001 & 3     & 3     & 3 \\
					2     &  0b00000000100001 & 0b00100100101001 & 5     & 5     & 4.98764 \\
					3     &  0b10100000101001 & 0b10000101100000 & 5     & 5     & 4.98722 \\
					4     &  0b00101100000000 & 0b00000000000100 & \textcolor[rgb]{ 1,  0,  0}{3} & \textcolor[rgb]{ 1,  0,  0}{2.41504} & \textcolor[rgb]{ 1,  0,  0}{2.38489} \\
					5     &  0b00000000100000 & 0b00000000000000 & 0     & 0     & 0  \\
					6     &  0b10000000000000 & 0b10000000000000 & 1     & 1     & 1.10886 \\
					7     &  0b10000010000000 & 0b10000010000100 & 3     & 3     & 2.88753 \\
					8     &  0b10000000000110 & 0b10010000100010 &       &       &  \\
					\midrule
					&\multicolumn{2}{r}{ $-\log_2P =-\sum_{i=0}^{7}\log_2P^i $}& 23    & 22.41504 & 22.35614 \\
					\bottomrule
				\end{tabular*}
			\end{minipage}
		\end{table}%

		\section{SPECK64/128}
		\label{sec:SPECK64/128}
		\subsection{The trail of 14 rounds for SPECK64/128}
		\label{subsec:The CMAs in trail of 14 rounds SPECK64/128}
		In the trail we find, there are three non-independence CMAs listed in Table \ref{tab:(7,10) CMA of 14 round SPECK64/128}, Table \ref{tab:(8,11) CMA of 14 round SPECK64/128}, and Table \ref{tab:(9,12) CMA of 14 round SPECK64/128}. 
		
		\begin{table}[htp]
			\begin{center}
				\begin{minipage}{\textwidth}
					\caption{The CMA in Round 7 and Round 10 of the 14 round version SPECK64/128}
					\label{tab:(7,10) CMA of 14 round SPECK64/128}	
					\centering
					\begin{tabular}{ccc}
						\toprule
						$M_{7}$        & Round 7 of the key schedule       & The constraints on the output                         \\ 
						\midrule
						&\qquad \qquad{\textcolor{red}{15}}\qquad \qquad &   \\
						$\varDelta x$  &0b00000000000000{\textcolor{red}{000}}000000000000000 &   \\
						$\varDelta y$  &0b00000000000000{\textcolor{red}{001}}000000000000000 & $z_{17}=\lnot z_{16}=\lnot z_{15} = y_{15}$ \\ 
						$\varDelta z$  &0b00000000000000{\textcolor{red}{111}}000000000000000 &         \\ 
						\toprule
						$M_{10}$       & Round 10 of the key schedule           & The constraints on the input         \\
						\midrule
						&\qquad \qquad \qquad \qquad \qquad {\textcolor{red}{7}} &   \\
						$\varDelta zz$ & 0b0000000000000000000000{\textcolor{red}{111}}0000000 & \\ 
						$\varDelta u$  & 0b0000000000000111100000{\textcolor{red}{001}}0000000 & $\lnot{zz}_9={zz}_8={zz}_7$     \\
						$\varDelta v$  & 0b0000000000111100100000{\textcolor{red}{000}}0000000 &   \\
						\midrule
						\multicolumn{3}{c}{$zz\ =\ (z\oplus{\rm 0b0111})\ggg8$, $\Pr\left(M_{10}\middle\vert M_7\right)= \frac{\Pr(M_{10},M_{7})}{\Pr(M_{7})} =2^{-8}>2^{-10} = \Pr(M_{10})$}   \\ 
						\botrule
					\end{tabular}
				\end{minipage}
			\end{center}
		\end{table}
		
		\begin{table}[tp]
			\begin{center}
				\begin{minipage}{\textwidth}
					\caption{The CMA in Round 8 and Round 11 round of the 14 round version SPECK64/128}
					\label{tab:(8,11) CMA of 14 round SPECK64/128}		
					\centering
					\begin{tabular}{ccc}
						\toprule
						$M_{8}$        & Round 8 of the key schedule         & The constraints on the output                         \\ 
						\midrule
						&\qquad \  {\textcolor{red}{18}}\qquad \qquad \qquad & \\
						$\varDelta x$  & 0b0000000000{\textcolor{red}{0000}}000000000000000000 &   \\
						$\varDelta y$  & 0b0000000000{\textcolor{red}{0001}}111000000000000000 &$z_{21}=\lnot z_{20}=\lnot z_{19}=\lnot z_{18} = y_{18}$ \\ 
						$\varDelta z$  & 0b0000000000{\textcolor{red}{1111}}001000000000000000 &         \\ 
						\toprule
						$M_{11}$       & Round 11 of the key schedule          & The constraints on the input         \\
						\midrule
						&\qquad \qquad \qquad \qquad \ \ {\textcolor{red}{10}}\qquad \   & \\
						$\varDelta zz$ & 0b000000000000000000{\textcolor{red}{1111}}0010000000 & \\ 
						$\varDelta u$  & 0b000000000000000010{\textcolor{red}{0001}}0000000000 & $\lnot{zz}_{13}={zz}_{12}={zz}_{11}={zz}_{10}$ \\
						$\varDelta v$  & 0b000000000000000010{\textcolor{red}{0000}}0010000000 &   \\
						\midrule
						\multicolumn{3}{c}{$zz\ =\ (z\oplus{\rm 0b1000})\ggg8$, $\Pr\left(M_{11}\middle\vert  M_8\right)= \frac{\Pr(M_{8},M_{11})}{\Pr(M_{8})} =2^{-3}>2^{-6} = \Pr(M_{11})$}   \\ 
						\botrule
					\end{tabular}
				\end{minipage}
			\end{center}
		\end{table}
		
		\begin{table}[tp]
			\begin{center}
				\begin{minipage}{\textwidth}
					\caption{The CMA in Round 9 and Round 12 of the 14 round version SPECK64/128}
					\label{tab:(9,12) CMA of 14 round SPECK64/128}	
					\centering
					\begin{tabular}{ccc}
						\toprule
						$M_{9}$        & Round 9 of the key schedule       & The constraints on the output                         \\ 
						\midrule
						&\qquad \qquad{\textcolor{red}{15}}\qquad \qquad &   \\
						$\varDelta x$  &0b00000000000000{\textcolor{red}{000}}000000000000000 &   \\
						$\varDelta y$  &0b00000000000000{\textcolor{red}{001}}000000000000000 &$z_{17}=\lnot z_{16}=\lnot z_{15} = y_{15}$ \\ 
						$\varDelta z$  &0b00000000000000{\textcolor{red}{111}}000000000000000 &         \\ 
						\toprule
						$M_{12}$       & Round 12 of the key schedule          & The constraints on the input         \\
						\midrule
						&\qquad \qquad \qquad \qquad \qquad {\textcolor{red}{7}} &   \\
						$\varDelta zz$ & 0b0000000000000000000000{\textcolor{red}{111}}0000000 & \\ 
						$\varDelta u$  & 0b0000000000000100101000{\textcolor{red}{001}}0000000 & $\lnot{zz}_9={zz}_8={zz}_7$ \\
						$\varDelta v$  & 0b0000000000000100101000{\textcolor{red}{000}}0000000 &   \\
						\midrule
						\multicolumn{3}{c}{$zz\ =\ (z\oplus{\rm 0b1001})\ggg8$, $\Pr\left(M_{12}\middle\vert M_9\right)= \frac{\Pr(M_{12},M_{9})}{\Pr(M_{9})} =2^{-4}>2^{-6} = \Pr(M_{12})$}   \\ 
						\botrule
					\end{tabular}
				\end{minipage}
			\end{center}
		\end{table}
		
		\subsection{The trail of 15 rounds for SPECK64/128}
		\label{subsec:The CMAs in trail of 15 rounds SPECK64/128}
		Four non-independence CMAs in this trail are listed in Table \ref{tab:(7,10) CMA of 15 round SPECK64/128}, \ref{tab:(8,11) CMA of 15 round SPECK64/128}, \ref{tab:(9,12) CMA of 15 round SPECK64/128}, and \ref{tab:(10,13) CMA of 15 round SPECK64/128}.

		\begin{table}[htb]
			\begin{center}
				\begin{minipage}{\textwidth}
					\caption{The CMA in the 9th and 12th round of the 15 round version SPECK64/128}
					\label{tab:(9,12) CMA of 15 round SPECK64/128}		
					\centering
					\begin{tabular}{ccc}
						\toprule
						$M_{9}$        & Round 9 of the key schedule       & The constraints on the output      \\ 
						\midrule
						$\varDelta x$  & 0b00000000000000000000000010000000 & $z_{17} = \neg z_{16} = \neg z_{15} = y_{15}$ \\ 				
						$\varDelta y$  & 0b00000000000000001000000000000000 &  \multirow{2}{*}{$z_{7} =x_7$}         \\ 
						$\varDelta z$  & 0b00000000000000111000000010000000 &                             \\     
						\toprule
						$M_{12}$       & Round 12 of the key schedule          & The constraints on the input         \\
						\midrule
						$\varDelta zz$ & 0b10000000000000000000001110000000 &  \\ 
						$\varDelta u$  & 0b00000000000001011010000010000000 & $\lnot{zz}_9={zz}_8, {zz}_7 = u_7$ \\ 
						$\varDelta v$  & 0b10000000000011011010000100000000 &                                    \\ 
						\midrule
						\multicolumn{3}{c}{$zz\ =\ (z\oplus{\rm 0b1001})\ggg8$, $\Pr\left(M_{12}\middle\vert  M_9\right)= \frac{\Pr(M_{9},M_{12})}{\Pr(M_{9})} =2^{-7}>2^{-8} = \Pr(M_{12})$}   \\ 
						\botrule
					\end{tabular}
				\end{minipage}
			\end{center}
		\end{table}
		
		\begin{table}[tp]
			\begin{center}
				\begin{minipage}{\textwidth}
					\caption{The CMA in the 7th and 10th round of the 15 round version SPECK64/128}
					\label{tab:(7,10) CMA of 15 round SPECK64/128}		
					\centering
					\begin{tabular}{ccc}
						\toprule
						$M_{7}$        & Round 7 of the key schedule                            & The constraints on the output      \\ 
						\midrule
						$\varDelta x$  & 0b00000000000000000000000000000000 &  \\ 				
						$\varDelta y$  & 0b00000000000000001000000000000000 & $z_{17} = \neg z_{16} = \neg z_{15} = y_{15}$          \\ 
						$\varDelta z$  & 0b00000000000000111000000000000000 &            \\    
						\toprule
						$M_{10}$       & Round 10 of the key schedule          & The constraints on the input         \\
						\midrule
						$\varDelta zz$ & 0b00000000000000000000001110000000 & \multirow{3}{*}{$\lnot{zz}_9 = {zz}_8 = {zz}_7$}  \\  
						$\varDelta u$  & 0b00000000000001111000000010000000 &                                                                  \\   
						$\varDelta v$  & 0b00000000001111001000000000000000 &                                                                  \\ 
						\midrule
						\multicolumn{3}{c}{ $zz\ =\ (z\oplus{\rm 0b0111})\ggg8$, $\Pr\left(M_{10}\middle\vert M_7\right)= \frac{\Pr(M_{7},M_{10})}{\Pr(M_{7})} =2^{-8}>2^{-10} = \Pr(M_{10})$}   \\ 
						\botrule
					\end{tabular}
				\end{minipage}
			\end{center}
		\end{table}
		
		\begin{table}[tp]
			\begin{center}
				\begin{minipage}{\textwidth}
					\caption{The CMA in the 10th and 13th round of the 15 round version SPECK64/128}
					\label{tab:(10,13) CMA of 15 round SPECK64/128}		
					\centering
					\begin{tabular}{ccc}
						\toprule
						$M_{10}$        & Round 10 of the key schedule                            & The constraints on the output      \\ 
						\midrule
						$\varDelta x$  & 0b00000000000000000000001110000000 & \multirow{3}{*}{$z_{21}=\lnot z_{20}=\lnot z_{19}=\lnot z_{18}$} \\   					
						$\varDelta y$  & 0b00000000000001111000000010000000 &                                                                  \\  
						$\varDelta z$  & 0b00000000001111001000000000000000 &                                                                  \\ 
						\toprule
						$M_{13}$       & Round 13 of the key schedule          & The constraints on the input         \\
						\midrule
						$\varDelta zz$ & 0b00000000000000000011110010000000 & \multirow{3}{*}{${zz}_{12} = {zz}_{11} = {zz}_{10} $}  \\   
						$\varDelta u$  & 0b10000000001000001010010100000000 &                                                               \\   
						$\varDelta v$  & 0b10000000001000000110000010000000 &                                                                  \\
						\midrule
						\multicolumn{3}{c}{$zz\ =\ (z\oplus{\rm 0b01010})\ggg8$, $\Pr\left(M_{13}\middle\vert M_{10}\right)= \frac{\Pr(M_{10},M_{13})}{\Pr(M_{10})} =2^{-7}>2^{-9} = \Pr(M_{13})$}   \\ 
						\botrule
					\end{tabular}
				\end{minipage}
			\end{center}
		\end{table}
		
		\begin{table}[tp]
			\begin{center}
				\begin{minipage}{\textwidth}
					\caption{The CMA in the 8th and 11th round of the 15 round version SPECK64/128}
					\label{tab:(8,11) CMA of 15 round SPECK64/128}		
					\centering
					\begin{tabular}{ccc}
						\toprule
						$M_{8}$        & Round 8 of the key schedule                            & The constraints on the output      \\ 
						\midrule
						$\varDelta x$  & 0b00000000000000000000000000000000 & \multirow{3}{*}{$z_{21}=\lnot z_{20}=\lnot z_{19}=\lnot z_{18}$  } \\   					
						$\varDelta y$  & 0b00000000000001111000000000000000 &                                                                  \\  
						$\varDelta z$  & 0b00000000001111001000000000000000 &                                                                  \\ 
						\toprule
						$M_{11}$       & Round 11 of the key schedule        & The constraints on the input         \\
						\midrule
						$\varDelta zz$ & 0b00000000000000000011110010000000 & \multirow{3}{*}{$ \neg z_{23} = {zz}_{12} = {zz}_{11} ={zz}_{10}$}  \\   
						$\varDelta u$  & 0b00000000000000001000010000000000	&                                                               \\   
						$\varDelta v$  & 0b00000000000000011000000010000000 &                                                                  \\ 
						\midrule
						\multicolumn{3}{c}{$zz\ =\ (z\oplus{\rm 0b1000})\ggg8$, $\Pr\left(M_{11}\middle\vert M_{8}\right)= \frac{\Pr(M_{8},M_{11})}{\Pr(M_{8})} =2^{-4}>2^{-7} = \Pr(M_{11})$}   \\ 
						\botrule
					\end{tabular}
				\end{minipage}
			\end{center}
		\end{table}
		
		\section{The detail of the trails }
		\label{sec:The detail of the trails}  
		The RKD trails of SPECK family which are found in Sect. \ref{subsec:seach result} and have higher probability than the previous work are listed in this section.
		\subsection{The trail for SPECK32/64}
		For SPECK32/64, the trails of 11 and 15 rounds are shown in Table \ref{tab:The RKD trails of 11 rounds for SPECK32/64} and Table \ref{tab:The RKD trails of 15 rounds for SPECK32/64}.	
		\begin{table}[htp]
			\begin{center}
				\begin{minipage}{300pt}
					\caption{The optimal RKD trails of 11 rounds for SPECK32/64}
					\label{tab:The RKD trails of 11 rounds for SPECK32/64}
					\centering
					\begin{tabular}{ccccccc}
						\toprule
						Round & $\varDelta l$        & $\varDelta k$                 & $-log_2{P_K^i}$& $\varDelta x$                 & $\varDelta y$                 & $-log_2{P_D^i}$ \\ 
						\midrule
						0     & 0x0011               & 0x4a00                        & 4                                   & 0x48e1                        & 0x42e0                        &                                     \\
						1     & 0x0080               & 0x0001                        & 1                                   & 0x02e1                        & 0x4001                        & 4                                   \\
						2     & 0x0200               & 0x0004                        & 1                                   & 0x0205                        & 0x0200                        & 3                                   \\
						3     & 0x2800               & 0x0010                        & 2                                   & 0x0800                        & 0x0000                        & 1                                   \\
						4     & 0x0000               & {\color[HTML]{FE0000} 0x0000} & 0                                   & {\color[HTML]{FE0000} 0x0000} & {\color[HTML]{FE0000} 0x0000} & 0                                   \\
						5     & 0x0000               & {\color[HTML]{FE0000} 0x0000} & 0                                   & {\color[HTML]{FE0000} 0x0000} & {\color[HTML]{FE0000} 0x0000} & 0                                   \\
						6     & 0x0040               & {\color[HTML]{FE0000} 0x0000} & 0                                   & {\color[HTML]{FE0000} 0x0000} & {\color[HTML]{FE0000} 0x0000} & 0                                   \\
						7     & 0x0000               & 0x8000                        & 0                                   & {\color[HTML]{FE0000} 0x0000} & {\color[HTML]{FE0000} 0x0000} & 0                                   \\
						8     & 0x0000               & 0x8002                        & 1                                   & 0x8000                        & 0x8000                        & 1                                   \\
						9     & 0x8000               & 0x8008                        & 2                                   & 0x0102                        & 0x0100                        & 3                                   \\
						10    & 0x8000               & 0x812a                        &                                     & 0x850a                        & 0x810a                        & 5                                   \\
						11    & \multicolumn{1}{l}{} & \multicolumn{1}{l}{}          &                                     & 0x152a                        & 0x1100                        &                                     \\ 
						&                      & $-log_2{P_K}$                 & 11              &                               & $-log_2{P_D}$                 & 17           \\ 
						\midrule
						\multicolumn{7}{c}{One of the Weak key: $(l^2,l^1,l^0,k^0)$ = 0xb90d, 0x6d3, 0x2d46, 0xdf0b.}\\
						\botrule                                                                                   
					\end{tabular}
				\end{minipage}
			\end{center}
		\end{table}
		\begin{table}[htp]
			\begin{center}
				\begin{minipage}{300pt}
					\caption{The RKD trails of 15 rounds for SPECK32/64}
					\label{tab:The RKD trails of 15 rounds for SPECK32/64}
					\centering
					\begin{tabular}{ccccccc}
						\toprule
						Round & $\varDelta l$        & $\varDelta k$                 & $-log_2{P_K^i}$& $\varDelta x$                 & $\varDelta y$                 & $-log_2{P_D^i}$ \\ 
						\midrule
						0     & 0x0400        & 0x0009                        & 2               & 0x524b                        & 0x064a                        &                 \\
						1     & 0x1880        & 0x0025                        & 4               & 0x5242                        & 0x5408                        & 7               \\
						2     & 0x4000        & 0x0080                        & 1               & 0x5081                        & 0x00a0                        & 4               \\
						3     & 0x0001        & 0x0200                        & 1               & 0x0281                        & 0x0001                        & 3               \\
						4     & 0x0014        & 0x0800                        & 2               & 0x0004                        & 0x0000                        & 1               \\
						5     & 0x0000        & {\color[HTML]{FE0000} 0x0000} & 0               & {\color[HTML]{FE0000} 0x0000} & {\color[HTML]{FE0000} 0x0000} & 0               \\
						6     & 0x0000        & {\color[HTML]{FE0000} 0x0000} & 0               & {\color[HTML]{FE0000} 0x0000} & {\color[HTML]{FE0000} 0x0000} & 0               \\
						7     & 0x2000        & {\color[HTML]{FE0000} 0x0000} & 1               & {\color[HTML]{FE0000} 0x0000} & {\color[HTML]{FE0000} 0x0000} & 0               \\
						8     & 0x0000        & 0x0040                        & 2               & {\color[HTML]{FE0000} 0x0000} & {\color[HTML]{FE0000} 0x0000} & 0               \\
						9     & 0x0000        & 0x01c0                        & 5               & 0x0040                        & 0x0040                        & 2               \\
						10    & 0x0040        & 0x0140                        & 2               & 0x8100                        & 0x8000                        & 2               \\
						11    & 0x00c0        & 0x8440                        & 15              & 0x8042                        & 0x8040                        & 3               \\
						12    & 0x0640        & 0x6afd                        & 12              & 0x8100                        & 0x8002                        & 2               \\
						13    & 0x8140        & 0xc01e                        & 6               & 0xebfd                        & 0xebf7                        & 5               \\
						14    & 0x7bff        & 0x416b                        &                 & 0x2fc0                        & 0x801f                        & 3               \\
						15    &               &                               &                 & 0x4155                        & 0x412b                        &                 \\
						&               & $-log_2{P_K}$                 & 53              &                               & $-log_2{P_D}$                 & 32              \\
						\midrule
						\multicolumn{7}{c}{One of the Weak key: $(l^2,l^1,l^0,k^0)$ = 0xb349, 0x973a, 0x786f, 0x31cb.}\\
						\botrule                                                                                   
					\end{tabular}
				\end{minipage}
			\end{center}
		\end{table}
		
		\subsection{The trails for SPECK48/96}
		For SPECK48/96, the trails of 11, 12, 14, and 15 rounds are shown in Table \ref{tab:The RKD trails of 11 rounds for SPECK48/96}, Table \ref{tab:The RKD trails of 12 rounds for SPECK48/96}, Table \ref{tab:The RKD trails of 14 rounds for SPECK48/96} and Table \ref{tab:The RKD trails of 15 rounds for SPECK48/96} respectively.
		\begin{table}[tp]
			\begin{center}
				\begin{minipage}{300pt}
					\caption{The RKD trails of 11 rounds for SPECK48/96}
					\label{tab:The RKD trails of 11 rounds for SPECK48/96}
					\centering
					\begin{tabular}{ccccccc}
						\toprule
						Round & $\varDelta l$        & $\varDelta k$                 & $-log_2{P_K^i}$& $\varDelta x$                 & $\varDelta y$                 & $-log_2{P_D^i}$ \\ 
						\midrule
						0     & 0x820008             & 0x081200                        & 4                                   & 0x4a12d0                        & 0x4040d0                        &                                     \\
						1     & 0x000040             & 0x400000                        & 1                                   & 0x4200d0                        & 0x024000                        & 5                                   \\
						2     & 0x000200             & 0x000002                        & 1                                   & 0x120200                        & 0x000200                        & 3                                   \\
						3     & 0x009000             & 0x000010                        & 2                                   & 0x001000                        & 0x000000                        & 1                                   \\
						4     & 0x000000             & {\color[HTML]{FE0000} 0x000000} & 0                                   & {\color[HTML]{FE0000} 0x000000} & {\color[HTML]{FE0000} 0x000000} & 0                                   \\
						5     & 0x000000             & {\color[HTML]{FE0000} 0x000000} & 0                                   & {\color[HTML]{FE0000} 0x000000} & {\color[HTML]{FE0000} 0x000000} & 0                                   \\
						6     & 0x000080             & {\color[HTML]{FE0000} 0x000000} & 0                                   & {\color[HTML]{FE0000} 0x000000} & {\color[HTML]{FE0000} 0x000000} & 0                                   \\
						7     & 0x000000             & 0x800000                        & 0                                   & {\color[HTML]{FE0000} 0x000000} & {\color[HTML]{FE0000} 0x000000} & 0                                   \\
						8     & 0x000000             & 0x800004                        & 1                                   & 0x800000                        & 0x800000                        & 1                                   \\
						9     & 0x800000             & 0x800020                        & 2                                   & 0x008004                        & 0x008000                        & 3                                   \\
						10    & 0x800000             & 0x808124                        &                                     & 0x8480a0                        & 0x8080a0                        & 5                                   \\
						11    &  					 &  					           &                                     & 0xa08504                        & 0xa48000                        &                                     \\  
						&                      & $-log_2{P_K}$                   & 11             					 &                                 & $-log_2{P_D}$                   & 18                  \\  
						\midrule
						\multicolumn{7}{c}{One of the Weak key: $(l^2,l^1,l^0,k^0)$ = 0x79f899, 0x47b57e, 0x19fe62, 0xfbf6e3.}\\
						\botrule                                                                                   
					\end{tabular}
				\end{minipage}
			\end{center}
		\end{table}
		
		\begin{table}[tp]
			\begin{center}
				\begin{minipage}{300pt}
					\caption{The RKD trails of 12 rounds for SPECK48/96}
					\label{tab:The RKD trails of 12 rounds for SPECK48/96}
					\centering
					\begin{tabular}{ccccccc}
						\toprule
						Round & $\varDelta l$        & $\varDelta k$                 & $-log_2{P_K^i}$& $\varDelta x$                 & $\varDelta y$                 & $-log_2{P_D^i}$ \\ 
						\midrule
						0     & 0x820008             & 0x081200                        & 4                                   & 0x4a1250                        & 0x404050                        &                                     \\
						1     & 0x000040             & 0x400000                        & 1                                   & 0x420050                        & 0x024000                        & 5                                   \\
						2     & 0x000200             & 0x000002                        & 1                                   & 0x120200                        & 0x000200                        & 3                                   \\
						3     & 0x009000             & 0x000010                        & 2                                   & 0x001000                        & 0x000000                        & 1                                   \\
						4     & 0x000000             & {\color[HTML]{FE0000} 0x000000} & 0                                   & {\color[HTML]{FE0000} 0x000000} & {\color[HTML]{FE0000} 0x000000} & 0                                   \\
						5     & 0x000000             & {\color[HTML]{FE0000} 0x000000} & 0                                   & {\color[HTML]{FE0000} 0x000000} & {\color[HTML]{FE0000} 0x000000} & 0                                   \\
						6     & 0x000080             & {\color[HTML]{FE0000} 0x000000} & 0                                   & {\color[HTML]{FE0000} 0x000000} & {\color[HTML]{FE0000} 0x000000} & 0                                   \\
						7     & 0x000000             & 0x800000                        & 0                                   & {\color[HTML]{FE0000} 0x000000} & {\color[HTML]{FE0000} 0x000000} & 0                                   \\
						8     & 0x000000             & 0x800004                        & 1                                   & 0x800000                        & 0x800000                        & 1                                   \\
						9     & 0x800000             & 0x800020                        & 2                                   & 0x008004                        & 0x008000                        & 3                                   \\
						10    & 0x800000             & 0x808124                        & 4                                   & 0x8480a0                        & 0x8080a0                        & 5                                   \\
						11    & 0x800004             & 0x840800                        &                                     & 0xa08504                        & 0xa48000                        & 7                                   \\
						12    & 					 &  				               &                                     & 0x242885                        & 0x002880                        &                                     \\ 
						&                      & $-log_2{P_K}$                   & 15             					 &                                 & $-log_2{P_D}$                   & 25         \\  
						\midrule
						\multicolumn{7}{c}{One of the Weak key: $(l^2,l^1,l^0,k^0)$ = 0x9c88b5, 0x482d8, 0x941556, 0xec0bfa.}\\
						\botrule                                                                                   
					\end{tabular}
				\end{minipage}
			\end{center}
		\end{table}
		
		\begin{table}[tp]
			\begin{center}
				\begin{minipage}{\textwidth}
					\caption{The RKD trails of 14 rounds for SPECK48/96}
					\label{tab:The RKD trails of 14 rounds for SPECK48/96}
					\begin{tabular*}{\textwidth}{ccccccc}
						\toprule
						Round & $\varDelta l$        & $\varDelta k$                 & $-log_2{P_K^i}$& $\varDelta x$                 & $\varDelta y$                 & $-log_2{P_D^i}$ \\ 
						\midrule
						0     & 0x0092c4             & 0x440810                        & 6                                   & 0x6d0831                        & 0x212821                        &                                     \\
						1     & 0x080020             & 0x204800                        & 3                                   & 0x290021                        & 0x082800                        & 7                                   \\
						2     & 0x000100             & 0x020001                        & 2                                   & 0x090900                        & 0x484900                        & 9                                   \\
						3     & 0x000882             & 0x120008                        & 2.5906 (3)*                                   & 0x4a420a                        & 0x080a08                        & 10                                  \\
						4     & 0x004000             & 0x000040                        & 1                                   & 0x005042                        & 0x400002                        & 5                                   \\
						5     & 0x020000             & 0x000200                        & 1                                   & 0x020012                        & 0x020000                        & 3                                   \\
						6     & 0x900000             & 0x001000                        & 2                                   & 0x100000                        & 0x000000                        & 1                                   \\
						7     & 0x000000             & {\color[HTML]{FE0000} 0x000000} & 0                                   & {\color[HTML]{FE0000} 0x000000} & {\color[HTML]{FE0000} 0x000000} & 0                                   \\
						8     & 0x000000             & {\color[HTML]{FE0000} 0x000000} & 0                                   & {\color[HTML]{FE0000} 0x000000} & {\color[HTML]{FE0000} 0x000000} & 0                                   \\
						9     & 0x008000             & {\color[HTML]{FE0000} 0x000000} & 1                                   & {\color[HTML]{FE0000} 0x000000} & {\color[HTML]{FE0000} 0x000000} & 0                                   \\
						10    & 0x000000             & 0x000080                        & 1                                   & {\color[HTML]{FE0000} 0x000000} & {\color[HTML]{FE0000} 0x000000} & 0                                   \\
						11    & 0x000000             & 0x000480                        & 2                                   & 0x000080                        & 0x000080                        & 1                                   \\
						12    & 0x000080             & 0x002080                        & 2                                   & 0x800400                        & 0x800000                        & 2                                   \\
						13    & 0x000080             & 0x812480                        &                                     & 0x80a084                        & 0x80a080                        & 5                                   \\
						14    & 					 & 						           &                                     & 0x8504a0                        & 0x8000a4                        &                                     \\
						&                      & $-log_2{P_K}$                   & 23.5906 (24)          &                                 & $-log_2{P_D}$                   & 43             \\ 
						\midrule
						\multicolumn{7}{c}{One of the Weak key: $(l^2,l^1,l^0,k^0)$ = 0xb67424, 0xd2a212, 0x3cadda, 0x65c7df.}\\
						\botrule                                                                                   
					\end{tabular*}
					\footnotetext[*]{The values in brackets are the probability computed under the independence assumption.}
				\end{minipage}
			\end{center}
		\end{table}
		
		\begin{table}[tp]
			\begin{center}
				\begin{minipage}{\textwidth}
					\caption{The RKD trails of 15 rounds for SPECK48/96}
					\label{tab:The RKD trails of 15 rounds for SPECK48/96}
					\centering
					\begin{tabular*}{\textwidth}{ccccccc}
						\toprule
						Round & $\varDelta l$        & $\varDelta k$                 & $-log_2{P_K^i}$& $\varDelta x$                 & $\varDelta y$                 & $-log_2{P_D^i}$ \\ 
						\midrule
						0     & 0x820008      & 0x081200                        & 4               & 0x4a1250                        & 0x404050                        &                 \\
						1     & 0x000040      & 0x400000                        & 1               & 0x420050                        & 0x024000                        & 5               \\
						2     & 0x000200      & 0x000002                        & 1               & 0x120200                        & 0x000200                        & 3               \\
						3     & 0x009000      & 0x000010                        & 2               & 0x001000                        & 0x000000                        & 1               \\
						4     & 0x000000      & {\color[HTML]{FE0000} 0x000000} & 0               & {\color[HTML]{FE0000} 0x000000} & {\color[HTML]{FE0000} 0x000000} & 0               \\
						5     & 0x000000      & {\color[HTML]{FE0000} 0x000000} & 0               & {\color[HTML]{FE0000} 0x000000} & {\color[HTML]{FE0000} 0x000000} & 0               \\
						6     & 0x000080      & {\color[HTML]{FE0000} 0x000000} & 0               & {\color[HTML]{FE0000} 0x000000} & {\color[HTML]{FE0000} 0x000000} & 0               \\
						7     & 0x000000      & 0x800000                        & 0               & {\color[HTML]{FE0000} 0x000000} & {\color[HTML]{FE0000} 0x000000} & 0               \\
						8     & 0x000000      & 0x800004                        & 1               & 0x800000                        & 0x800000                        & 1               \\
						9     & 0x800000      & 0x800020                        & 4               & 0x008004                        & 0x008000                        & 3               \\
						10    & 0x800000      & 0x8081e4                        & 9               & 0x8480a0                        & 0x8080a0                        & 6               \\
						11    & 0x800004      & 0x840000                        & 5               & 0xa08584                        & 0xa48080                        & 7               \\
						12    & 0x8080e0      & 0xab8004                        & 9               & 0xa42005                        & 0x802400                        & 8               \\
						13    & 0x800f24      & 0x0400a1                        & 5.5008 (9)               & 0x210024                        & 0x202020                        & 5               \\
						14    & 0x8b8000      & 0x0085a8                        &                 & 0x000181                        & 0x010080                        & 3               \\
						15    &               &                                 &                 & 0x808529                        & 0x888129                        &                 \\  
						&               & $-log_2{P_K}$                   & 41.5008 (45)              &                                 & $-log_2{P_D}$                   & 42              \\  
						\midrule
						\multicolumn{7}{c}{One of the Weak key: $(l^2,l^1,l^0,k^0)$ = 0x345351, 0x6c76c1, 0xa3cadf, 0xb4f9f9.} \\
						\botrule                                                                                   
					\end{tabular*}
				\end{minipage}
			\end{center}
		\end{table}
		
		\subsection{The trails for SPECK64/128}
		For SPECK64/128, the trails of 14 and 15 rounds are shown in Table \ref{tab:The RKD trails of 14 rounds for SPECK64/128} and Table \ref{tab:The RKD trails of 15 rounds for SPECK64/128} respectively.
		\begin{table}[htp]
			\begin{center}
				\begin{minipage}{\textwidth}
					\caption{The RKD trails of 14 rounds for SPECK64/128}
					\label{tab:The RKD trails of 14 rounds for SPECK64/128}
					\centering
					\begin{tabular*}{\textwidth}{ccccccc}
						\toprule
						Round & $\varDelta l$        & $\varDelta k$                 & $-log_2{P_K^i}$& $\varDelta x$                 & $\varDelta y$                 & $-log_2{P_D^i}$ \\ 
						\midrule
						0     & 0x00080082    & 0x12000800                        & 3               & 0x12401842                        & 0x40401040                        &                 \\
						1     & 0x00400000    & 0x00004000                        & 1               & 0x00401042                        & 0x40000002                        & 5               \\
						2     & 0x02000000    & 0x00020000                        & 1               & 0x02000012                        & 0x02000000                        & 3               \\
						3     & 0x90000000    & 0x00100000                        & 2               & 0x10000000                        & 0x00000000                        & 1               \\
						4     & 0x00000000    & {\color[HTML]{FE0000} 0x00000000} & 0               & {\color[HTML]{FE0000} 0x00000000} & {\color[HTML]{FE0000} 0x00000000} & 0               \\
						5     & 0x00000000    & {\color[HTML]{FE0000} 0x00000000} & 0               & {\color[HTML]{FE0000} 0x00000000} & {\color[HTML]{FE0000} 0x00000000} & 0               \\
						6     & 0x00800000    & {\color[HTML]{FE0000} 0x00000000} & 1               & {\color[HTML]{FE0000} 0x00000000} & {\color[HTML]{FE0000} 0x00000000} & 0               \\
						7     & 0x00000000    & 0x00008000                        & 3               & {\color[HTML]{FE0000} 0x00000000} & {\color[HTML]{FE0000} 0x00000000} & 0               \\
						8     & 0x00000000    & 0x00078000                        & 7               & 0x00008000                        & 0x00008000                        & 4               \\
						9     & 0x00008000    & 0x00008000                        & 4               & 0x00040080                        & 0x00000080                        & 2               \\
						10    & 0x00038000    & 0x00078080                        & 8 (10)              & 0x80008480                        & 0x80008080                        & 6               \\
						11    & 0x003c8000    & 0x00008400                        & 3 (6)               & 0x00840084                        & 0x00800480                        & 5               \\
						12    & 0x00038080    & 0x0004a080                        & 4 (6)               & 0x84800480                        & 0x80802080                        & 6               \\
						13    & 0x003c8000    & 0x8021a400                        &                 & 0x00000004                        & 0x04010400                        & 3               \\
						14    &               &                                   &                 & 0x8020a000                        & 0xa0288000                        &                 \\   
						&               & $-log_2{P_K}$                     & 37 (44)              &                                   & $-log_2{P_D}$                     & 35              \\   
						\midrule
						\multicolumn{7}{c}{One of the Weak key: $(l^2,l^1,l^0,k^0)$ = 0x748e0a7d, 0x928c0d5b, 0x29084dba, 0x49b9e7a2.} \\
						\botrule                                                                                   
					\end{tabular*}
				\end{minipage}
			\end{center}
		\end{table}
		
		\begin{table}[htp]
			\begin{center}
				\begin{minipage}{\textwidth}
					\caption{The RKD trails of 15 rounds for SPECK64/128}
					\label{tab:The RKD trails of 15 rounds for SPECK64/128}
					\centering
					\begin{tabular*}{\textwidth}{ccccccc}
						\toprule
						Round & $\varDelta l$        & $\varDelta k$                 & $-log_2{P_K^i}$& $\varDelta x$                 & $\varDelta y$                 & $-log_2{P_D^i}$ \\ 
						\midrule
						0     & 0x00080082 & 0x12000800                        & 3             & 0x124018c2                        & 0x404010c0                        &              \\
						1     & 0x00400000 & 0x00004000                        & 1             & 0x004010c2                        & 0x40000002                        & 5             \\
						2     & 0x02000000 & 0x00020000                        & 1             & 0x02000012                        & 0x02000000                        & 3             \\
						3     & 0x90000000 & 0x00100000                        & 2             & 0x10000000                        & 0x00000000                        & 1             \\
						4     & 0x00000000 & {\color[HTML]{FE0000} 0x00000000} & 0             & {\color[HTML]{FE0000} 0x00000000} & {\color[HTML]{FE0000} 0x00000000} & 0             \\
						5     & 0x00000000 & {\color[HTML]{FE0000} 0x00000000} & 0             & {\color[HTML]{FE0000} 0x00000000} & {\color[HTML]{FE0000} 0x00000000} & 0             \\
						6     & 0x00800000 & {\color[HTML]{FE0000} 0x00000000} & 1             & {\color[HTML]{FE0000} 0x00000000} & {\color[HTML]{FE0000} 0x00000000} & 0             \\
						7     & 0x00000000 & 0x00008000                        & 3             & {\color[HTML]{FE0000} 0x00000000} & {\color[HTML]{FE0000} 0x00000000} & 0             \\
						8     & 0x00000000 & 0x00078000                        & 7             & 0x00008000                        & 0x00008000                        & 4             \\
						9     & 0x00008000 & 0x00008000                        & 4             & 0x00040080                        & 0x00000080                        & 2             \\
						10    & 0x00038000 & 0x00078080                        & 8 (10)            & 0x80008480                        & 0x80008080                        & 6             \\
						11    & 0x003c8000 & 0x00008400                        & 4 (7)             & 0x00840084                        & 0x00800480                        & 5             \\
						12    & 0x00038080 & 0x0005a080                        & 7 (8)             & 0x84800480                        & 0x80802080                        & 6             \\
						13    & 0x003c8000 & 0x8020a500                        & 7 (8)             & 0x00010004                        & 0x04000400                        & 3             \\
						14    & 0x00018080 & 0x81254884                        &               & 0x8020a000                        & 0xa0208000                        & 7             \\
						15    &            &                                   &               & 0x2185e824                        & 0x2081e821                        &               \\   
						&            & $-log_2{P_K}$                     & 48 (55)            &                                   & $-log_2{P_D}$                     & 42            \\    
						\midrule
						\multicolumn{7}{l}{One of the Weak key: $(l^2,l^1,l^0,k^0)$ = 0xde9b2c6a, 0x5a50cb19, 0x8f42899e, 0xf8bbaeae.}\\
						\botrule                                                                                   
					\end{tabular*}
				\end{minipage}
			\end{center}
		\end{table}
		
	\end{appendices}

\end{document}